\documentclass[a4paper,reqno]{amsart}

\usepackage[T1]{fontenc}

\usepackage[margin=2cm]{geometry}
\usepackage{float}
\usepackage[export]{adjustbox}

\makeatletter

\renewcommand\subsubsection{\@secnumfont}{\bfseries}%
\renewcommand\subsubsection{\@startsection{subsubsection}{3}
  \z@{.5\linespacing\@plus.7\linespacing}{-.5em}%
  {\normalfont\bfseries}}
  
\makeatother

\usepackage[
  colorlinks,
  linkcolor = blue,
  citecolor = blue,
  urlcolor = blue]{hyperref}

\usepackage{amsmath,amsthm,amssymb,amsfonts}
\usepackage[foot]{amsaddr}

\usepackage[table,xcdraw]{xcolor}
\usepackage[skins]{tcolorbox}

\usepackage{tikz}
\usepackage{pgfplots}
\pgfplotsset{width=10cm,compat=1.13}
\usetikzlibrary{calc}
\usetikzlibrary{decorations.pathmorphing, decorations.markings}
\usetikzlibrary{decorations.pathreplacing}
\definecolor{darkred}{RGB}{179, 16, 32}

\usepackage{subcaption}

\usepackage{ytableau}
\ytableausetup{centertableaux}

\usepackage{nicematrix}
\usepackage{nicefrac}

\usepgfplotslibrary{external}
\usepackage[backend=biber,style=alphabetic,maxnames=9,maxalphanames=5,isbn=false,url=false,backref]{biblatex}

\usepackage{cleveref}%[nameinlink]
\crefname{figure}{Figure}{Figures}
\crefname{example}{Example}{Examples}
\crefname{theorem}{Theorem}{Theorems}
\crefname{proposition}{Proposition}{Propositions}
\crefname{lemma}{Lemma}{Lemmas}
\crefname{corollary}{Corollary}{Corollaries}
\crefname{remark}{Remark}{Remarks}
\crefname{section}{Section}{Sections}
\crefname{subsection}{Section}{Sections}
\crefname{subsubsection}{Section}{Sections}
\crefname{appendix}{Appendix}{Appendices}

\newtheorem{theorem}{Theorem}[section]
\newtheorem{conjecture}[theorem]{Conjecture}
\newtheorem*{theorem*}{Theorem}
\newtheorem{proposition}[theorem]{Proposition}
\newtheorem*{proposition*}{Proposition}
\newtheorem{lemma}[theorem]{Lemma}
\newtheorem{corollary}[theorem]{Corollary}
\newtheorem{remark}[theorem]{Remark}
\newtheorem*{remark*}{Remark}
\newtheorem{example}[theorem]{Example}
\newtheorem*{example*}{Example}
\newtheorem{cor}[theorem]{Corollary}

\usepackage{mathtools}
\DeclarePairedDelimiter{\set}{\lbrace}{\rbrace}
\DeclarePairedDelimiter{\abs}{\lvert}{\rvert}

\DeclarePairedDelimiter{\of}{\lparen}{\rparen}
\DeclarePairedDelimiter{\frob}{\langle}{\rangle}
\DeclarePairedDelimiter{\ceil}{\lceil}{\rceil}

\DeclareMathOperator{\Tr}{Tr} % trace
\renewcommand{\S}{\mathrm{S}} % symmetric group
\newcommand{\Br}{\mathrm{Br}} % Brauer algebra
\newcommand{\U}{\mathrm{U}} % unitary group
\newcommand{\R}{\mathbb{R}} % real numbers
\newcommand{\C}{\mathbb{C}} % complex numbers
\newcommand{\pt}{\mathbin{\vdash}} % partitions
\newcommand{\defeq}{\vcentcolon=}
\newcommand{\cont}{\mathrm{cont}}

\newcommand{\x}{\otimes}
\newcommand{\xp}[1]{^{\otimes #1}}
\newcommand{\ct}{^{*}}

\newcommand{\End}{\mathrm{End}}
\newcommand{\Orth}{\mathrm{O}}
\newcommand{\GL}{\mathrm{GL}}
\renewcommand{\H}{\mathcal{H}} % Hilbert space
\newcommand{\Herm}{\mathrm{Herm}} % Hermitian matrices

\newcommand{\0}{\varnothing} %empty set

\newcommand{\bra}[1]{\left\langle #1 \right|}
\newcommand{\ket}[1]{\left| #1 \right\rangle}
\newcommand{\braket}[2]{\left\langle #1 \middle| #2 \right\rangle}
\newcommand{\ketbra}[2]{\left| #1 \middle\rangle \mkern -\medmuskip \middle\langle #2 \right|}

\newcommand{\ydsm}[1]{\ytableausetup{boxsize = 2pt} \ydiagram{#1}}
\newcommand{\yd}[1]{\ytableausetup{centertableaux, boxsize = 4pt} \ydiagram{#1}}

\newcommand{\cU}{\mathcal{U}}
\newcommand{\cO}{\mathcal{O}}
\newcommand{\cS}{\mathcal{S}}
\newcommand{\cB}{\mathcal{B}}

\newcommand{\F}{\mathrm{F}} % Flip (Swap) operator
\newcommand{\I}{\mathrm{I}} % Identity operator
\newcommand{\W}{\mathrm{W}} % Maximally entangled state projector operator
\newcommand{\Irr}[1]{\widehat#1} % Irreducible representations
\newcommand{\spec}{\mathrm{spec}} % Specturm
\newcommand{\spansp}{\mathrm{span}_\C} % span

\DeclareMathSymbol{\shortminus}{\mathbin}{AMSa}{"39}

\usepackage{xcolor}
\usepackage{todonotes}

\title{Monogamy of highly symmetric states}

\author{Rene Allerstorfer$^{1,2}$}
\email{rene.allerstorfer@cwi.nl}
\address{$^1$QuSoft, Amsterdam, The Netherlands}
\address{$^2$CWI, Amsterdam, The Netherlands}

\author{Matthias Christandl$^3$}
\email{christandl@math.ku.dk}
\address{$^3$Department of Mathematical Sciences, University of Copenhagen, Copenhagen, Denmark}

\author{Dmitry Grinko$^{1,4}$}
\email{d.grinko@uva.nl}
\address{$^4$Institute for Logic, Language and Computation, University of Amsterdam, Amsterdam, The Netherlands}

\author{Ion Nechita$^5$}
\email{ion.nechita@univ-tlse3.fr}
\address{$^5$Laboratoire de Physique Théorique, Université de Toulouse, CNRS, Toulouse, France}

\author{Maris Ozols$^{1,4,6,7}$}
\email{marozols@gmail.com}
\address{$^6$Korteweg-de Vries Institute for Mathematics, University of Amsterdam, The Netherlands}
\address{$^7$Institute for Theoretical Physics, University of Amsterdam, The Netherlands}

\author{Denis Rochette$^{8,9}$}
\email{denis.rochette@uottawa.ca}
\address{$^8$Institut de Mathématiques, Université de Toulouse, Toulouse, France}
\address{$^9$Department of Mathematics and Statistics, University of Ottawa, Ottawa, Canada}

\author{Philip Verduyn Lunel$^{1,2}$}
\email{philip.verduyn.lunel@cwi.nl}

\begin{document}

\begin{abstract}
    We investigate the extent to which two particles can be maximally entangled when they are also similarly entangled with other particles on a complete graph, focusing on Werner, isotropic, and Brauer states. To address this, we formulate and solve optimization problems that draw on concepts from many-body physics, computational complexity, and quantum cryptography. We approach the problem by formalizing it as a semi-definite program (SDP), which we solve analytically using tools from representation theory. Notably, we determine the exact maximum values for the projection onto the maximally entangled state and the antisymmetric Werner state, thereby resolving long-standing open problems in the field of quantum extendibility. Our results are achieved by leveraging SDP duality, the representation theory of symmetric, unitary and orthogonal groups, and the Brauer algebra.
\end{abstract}

\maketitle

\tableofcontents

\newpage
%%%%%%%%%%%%%%%%%%%%%%%%%%%%%%%%%%%%%%%%%%%%%%%%%%%%%%%%%%%%%%%%%%%%%%%%%%%%%%%%%%%%%%%%%%%%%%%%%%%%%%%%%%%%%%%%%%%%%%%%%%%%%%%%%%%%%%%%%%%%%%%%%%%%%%%%%%%%%%%%%%%%%%%%%%%%%%%%%%%%%%%%%%%%%%%%%%%%%%%%%%%%%%%%%%%%%%%%%%
%%%%%%%%%%%%%%%%%%%%%%%%%%%%%%%%%%%%%%%%%%%%%%%%%%%%%%%%%%%%%%%%%%%%%%%%%%%%%%%%%%%%%%%%%%%%%%%%%%%%%%%%%%%%%%%%%%%%%%%%%%%%%%%%%%%%%%%%%%%%%%%%%%%%%%%%%%%%%%%%%%%%%%%%%%%%%%%%%%%%%%%%%%%%%%%%%%%%%%%%%%%%%%%%%%%%%%%%%%
%%%%%%%%%%%%%%%%%%%%%%%%%%%%%%%%%%%%%%%%%%%%%%%%%%%%%%%%%%%%%%%%%%%%%%%%%%%%%%%%%%%%%%%%%%%%%%%%%%%%%%%%%%%%%%%%%%%%%%%%%%%%%%%%%%%%%%%%%%%%%%%%%%%%%%%%%%%%%%%%%%%%%%%%%%%%%%%%%%%%%%%%%%%%%%%%%%%%%%%%%%%%%%%%%%%%%%%%%%
\section{Introduction} \label{sec:introduction}

\subsection{Background}
\emph{Monogamy} of entanglement is a fundamental feature of quantum theory \cite{Terhal2004,koashi2004monogamy}. Intuitively, it states that if two quantum systems are entangled with each other, they cannot be too entangled with other systems. Incarnations of monogamy include a so-called \emph{quantum de Finetti theorem}, allowing, for example, security proofs of quantum cryptography \cite{renner2008security}, SDP relaxations for bilinear relaxations \cite{berta2021semidefinite} and ground energy approximations of local Hamiltonians via product states \cite{brandao2013product}. Studying monogamy in full generally is equivalent to so-called \emph{quantum marginal problem} \cite{walter2013entanglement,schilling2015quantum,klyachko2004quantum,Klyachko_2006}, which is a notoriously difficult problem. Restricted versions of the quantum marginal problem known as \emph{state extension} or \emph{state extendibility} problems \cite{wernerExt1989,doherty2014} are fundamental in quantum information. The main idea behind state extension is to certify or find a suitable global state on several quantum systems such that certain subsystems are in a fixed specified state. 

In this paper, we formalize state extension in a concept of $G$-\emph{extendibility} for arbitrary graphs $G$ and study it analytically for important classes of symmetric states widely used in quantum information on clique graphs $G=K_n$.

More concretely, a bipartite symmetric quantum state $\sigma$ is $G$-\emph{extendible}, for a graph $G$ if there exists a global state $\rho$, on the vertices of $G$, such that for all edges $e$ of $G$, the reduced state $\sigma_e$ is equal to $\sigma$. Our graph extendibility approach, which can be viewed as a continuation of the work described in \cite{wolf2003entanglement}, not only generalizes existing concepts but also serves as a unifying framework for various problems in quantum information theory.

Consider the scenario where the graph $G$ is a \emph{star graph} $K_{1,n}$. In this case, our $K_{1,n}$-\emph{extendibility} is equivalent to the established notion of $n$-\emph{extendibility} of a bipartite quantum state, which was first used to formalize the intuition behind \emph{monogamy of entanglement} \cite{Terhal2004}. It states that if a bipartite state is $n$-\emph{extendible} for every $n$, it must be a separable state of the form $\sum_{i} p_i \rho_i \otimes \sigma_i$ \cite{fannes1988symmetric,raggio1988quantum,doherty2004complete}.

In the instance where the graph $G$ is a \emph{complete bipartite graph} $K_{n,m}$, then our $K_{n,m}$-\emph{extendibility} correspond to the $n,m$-\emph{extendibility} (also known as \emph{symmetric extendibility} \cite{terhal2003symmetric}) of bipartite quantum state. Moreover, a bipartite state is $n,m$-\emph{extendible} for all $n,m$, if and only if it is $n$-\emph{extendible} for all $n$ \cite{johnson2013compatible}.

In the case where the graph $G$ equals a \emph{complete graph} $K_n$, then our $K_n$-\emph{extendibility} is equivalent to the $n$-\emph{exchangeability} of a bipartite symmetric quantum state. This notion is related to the celebrated \emph{quantum de Finetti theorem}. It asserts that if a bipartite state is $n$-\emph{exchangeable} for every $n$, then it is a convex combination of product states of the form $\sum_{i} p_i \rho_i \otimes \rho_i$ \cite{hudson1976locally,caves2002unknown,konig2005finetti,Christandl2007}.

All these notions have a lot of applications. For example, the $K_{1,n}$-extendibility of isotropic states is intrinsically related to $1 \to n$ quantum cloning problem \cite{werner1998optimal,keyl1999optimal,nechita2023asymmetric}. One approach to obtaining the optimal $1 \to n$ symmetric quantum cloning map is to exploit Choi--Jamiołkowski's isomorphism to translate a $K_{1,n}$-extendible isotropic state into a quantum channel. Furthermore, extendibility on \emph{circle graphs} has direct implications for quantum cryptography in the context of quantum position verification \cite{buhrman2014position,kent2011quantum,unruh2014quantum}. Our framework is also suitable for quantum network applications, notably in generating multiple EPR-pairs from an $n$-party resource state \cite{bravyi2022generating}.

In quantum information, it is common to consider classes of symmetric states, such as the one-parameter families of \emph{Werner} and \emph{isotropic states} \cite{werner1989quantum,horodecki1999reduction}. They admit a two-parameter generalization to what we call \emph{Brauer states} \cite{vollbrecht2001entanglement,park2023universal}, which are defined through the \emph{Schur--Weyl duality} of the \emph{orthogonal group} \cite{brauer1937algebras}. Brauer states can be seen as a generalization of Choi states of quantum channels commonly known as \emph{Werner--Holevo channels}.

In this work, we focus on understanding the monogamy of entanglement of such states. In particular, we determine the exact maximum values of the projection to the maximally entangled state and antisymmetric Werner state possible. Before our work, the exact values in arbitrary local dimensions for the projection overlap in the case of Werner, and isotropic states were known only in the context of the monogamy theorem, namely for $K_{1,n}$-extendibility \cite{vollbrecht2001entanglement} and for \emph{complete bipartite graph} $K_{n,m}$-extendibility \cite{jakab2022extendibility}. Some aspects of $K_n$-extendibility were also studied in \cite{jakabBilinearbiquadraticModelComplete2018,jakabQuantumPhasesCollective2021,jakabInterplayUnitaryPermutation2022}.

Understanding the properties of the mentioned symmetric states is important in quantum information. Examples of applications of these symmetries can be found in recent work focusing on developing approximation algorithms for local Hamiltonians, notably the Quantum Max-Cut problem \cite{takahashi20232,watts2023relaxations}. In these applications, the analytical values that we derive are crucial for understanding this model, as they provide insights beyond what can be obtained by asymptotic approximations. In the context of constructing approximation algorithms to Quantum Max-Cut problem, exact $K_{1,n}$-extendibility value played a crucial role in obtaining better approximation ratios beyond product state approximations \cite{anshu2020beyond,parekh2021application,lee2022optimizing,king2022improved,lee2024improved}. We expect our results to be similarly helpful for obtaining better approximation values for such optimization problems.

%%%%%%%%%%%%%%%%%%%%%%%%%%%%%%%%%%%%%%%%%%%%%%%%%%%%%%%%%%%%%%%%%%%%%%%%%%%%%%%%%%%%%%%%%%%%%%%%%%%%%%%%%%%%%%%%%%%%%%%%%%%%%%%%%%%%%%%%%%%%%%%%%%%%%%%%%%%%%%%%%%%%%%%%%%%%%%%%%%%%%%%%%%%%%%%%%%%%%%%%%%%%%%%%%%%%%%%%%%
%%%%%%%%%%%%%%%%%%%%%%%%%%%%%%%%%%%%%%%%%%%%%%%%%%%%%%%%%%%%%%%%%%%%%%%%%%%%%%%%%%%%%%%%%%%%%%%%%%%%%%%%%%%%%%%%%%%%%%%%%%%%%%%%%%%%%%%%%%%%%%%%%%%%%%%%%%%%%%%%%%%%%%%%%%%%%%%%%%%%%%%%%%%%%%%%%%%%%%%%%%%%%%%%%%%%%%%%%%
\newpage
\subsection{Summary of our results}

We informally summarize our results here. We shall consider the $n$-exchangeability of the three distinct families of symmetric quantum states: Werner, Brauer, and isotropic. We provide analytical solutions of maximum values of the projection onto the antisymmetric state in the Werner case, as well as into the maximally entangled state in the Brauer and isotropic case. We call those values $q_W(n,d), p_B(n,d)$ and $p_I(n,d)$, respectively.
\begin{theorem*}[Summary of Theorems \ref{thm:WernerStates}, \ref{thm:isotropicStates}, \ref{thm:BrauerStates} and \Cref{lem:qB=qW}]
    The maximum values of those projections are:
    \begin{align*}
        \textbf{Werner:} \qquad q_W(n,d) &= \frac{d-1}{2d} \frac{(n+k+d)(n-k)}{n(n-1)} + \frac{k(k-1)}{n(n-1)}, \text{ where } k = n \bmod d, \\[1em]
        \textbf{Brauer:} \qquad p_{B}(n,d) &= \frac{1}{d} + \of*{1-\frac{1}{d}} \frac{1}{n + n \bmod 2 - 1}, \\ 
            q_B(n,d) &= q_W(n,d), \\[1em]
        \textbf{Isotropic:} \qquad p_I(n,d) &=
            \begin{cases}
                \displaystyle
                \frac{1}{d^2} + \of*{1- \frac{1}{d^2}} \frac{1}{n + n \bmod 2 - 1} &\text{ if $d > n$ or either $d$ or $n$ is even}, \\[0.75em]
                \displaystyle
                \frac{1}{d^2} + \of*{1- \frac{1}{d^2}} \min \set*{ \frac{2 d + 1}{2 d n + 1}, \frac{1}{n - 1} } &\text{ if $n \geq d$ and both $d$ and $n$ are odd}.
            \end{cases}
    \end{align*}
\end{theorem*}
As a corollary, if we parameterize the Werner and isotropic states with a parameter $p$, and the Brauer states with two parameters $p$ and $q$ (see \cref{sec:WernerIsotropicBrauerStates}), we can summarize in the following table the parameter values for which those families of states are $K_{1,n}$-extendibility ($n$-extendible) and $K_n$-extendibility ($n$-exchangeable) for all $n$:
\begin{center}
    \begin{tabular}{r|c|c}
        & $K_{1,n}$-extendibility & $K_n$-extendibility \\ \hline
        Werner & $q \leq \tfrac{1}{2}$ & $q \leq \tfrac{d-1}{2d}$ \\[0.25em]
        Brauer & $p \leq \tfrac{1}{d} \wedge q \leq \tfrac{1}{2}$ & $p \leq \tfrac{1}{d} \wedge q \leq f(p)$  \\[0.25em]
        isotropic & $p \leq \tfrac{1}{d}$ & $ p = \tfrac{1}{d^2}$
    \end{tabular}
\end{center}
\bigskip
where $f(p)$ is some unknown function such that it is upper bounded $f(p) \leq \tfrac{d-1}{2d}$, see \Cref{conj}. However, for qubits the whole Brauer $(p,q)$-extendibility region could be obtained analytically for all $n$:
\begin{theorem*}[Theorem \ref{thm:d=2_brauer_region}]
    For $d=2$ and all $n \geq 2$ the maximal value $q(p)$ for every $p \in [0,1]$ equals to
    \begin{equation}
        q(p) =
        \begin{cases}
            \frac{\ceil{n/2} + 1}{\ceil{n/2} - 1} p &\text{if $p < \frac{\ceil{n/2}-1}{2(2\ceil{n/2}-1)}$}, \\
            \frac{\ceil{n/2}}{2\ceil{n/2}-1} - p &\text{if $\frac{\ceil{n/2}}{2\ceil{n/2}-1} \geq p \geq \frac{\ceil{n/2}-1}{2(2\ceil{n/2}-1)}$}, \\
            0 &\text{otherwise}.
        \end{cases}
    \end{equation}
\end{theorem*}
\noindent In particular, this theorem implies $f(p) = 1/4 - |1/4 - p|$ for qubits, see \Cref{fig:BrauerQubitPolytope}.

The paper is organized as follows. In \cref{sec:preliminaries} we recall some basic definitions and known representation theory results. In \cref{sec:generalFormalisationOfTheProblem}, we set up the problems through a primary and dual SDP formulation. \cref{sec:KnExtendibility} contain the main results of the paper, namely, the solution of the problems in the case of the complete graph for three families of symmetric states. Finally, \cref{sec:Cyclic} study the case of the cyclic graph in the context of Bell state discrimination.

%%%%%%%%%%%%%%%%%%%%%%%%%%%%%%%%%%%%%%%%%%%%%%%%%%%%%%%%%%%%%%%%%%%%%%%%%%%%%%%%%%%%%%%%%%%%%%%%%%%%%%%%%%%%%%%%%%%%%%%%%%%%%%%%%%%%%%%%%%%%%%%%%%%%%%%%%%%%%%%%%%%%%%%%%%%%%%%%%%%%%%%%%%%%%%%%%%%%%%%%%%%%%%%%%%%%%%%%%%
%%%%%%%%%%%%%%%%%%%%%%%%%%%%%%%%%%%%%%%%%%%%%%%%%%%%%%%%%%%%%%%%%%%%%%%%%%%%%%%%%%%%%%%%%%%%%%%%%%%%%%%%%%%%%%%%%%%%%%%%%%%%%%%%%%%%%%%%%%%%%%%%%%%%%%%%%%%%%%%%%%%%%%%%%%%%%%%%%%%%%%%%%%%%%%%%%%%%%%%%%%%%%%%%%%%%%%%%%%
%%%%%%%%%%%%%%%%%%%%%%%%%%%%%%%%%%%%%%%%%%%%%%%%%%%%%%%%%%%%%%%%%%%%%%%%%%%%%%%%%%%%%%%%%%%%%%%%%%%%%%%%%%%%%%%%%%%%%%%%%%%%%%%%%%%%%%%%%%%%%%%%%%%%%%%%%%%%%%%%%%%%%%%%%%%%%%%%%%%%%%%%%%%%%%%%%%%%%%%%%%%%%%%%%%%%%%%%%%
\section{Preliminaries} \label{sec:preliminaries}

\subsection{Notation}

In this paper, the term \emph{graph} refers to an undirected, simple graph that has no self-loops. A graph $G = (V,E)$ has vertex set $V$, edge set $E \subset V \times V$, and its number of vertices is equal to $n \defeq |V|$. We denote by $\mathrm{Aut}(G)$ the \emph{automorphism group} of the graph $G$. The \emph{complete graph} $K_n$ with $n$ vertices includes all possible edges, i.e.~$(u,v) \in E$ for each distinct pair of vertices $u$ and $v$ in $V$ (e.g., see \cref{fig:K5}). This graph has $n (n - 1)/2$ edges, which is the maximum number of edges in an $n$-vertex graph. The \emph{star graph} $K_{1,n}$ on $n+1$ vertices has a distinct central vertex $v \in V$ that is connected to each of the remaining $n$ vertices, i.e.~$E = \set{(u,v) \mid u \in V,\, u \neq v}$ (e.g., see \cref{fig:S6}). An \emph{edge-transitive} graph is a graph $G$ such that for any two edges $e_1$ and $e_2$ in $E$, there exists an automorphism of $G$ that maps $e_1$ to $e_2$ \cite{biggs1993algebraic}. Equivalently, a graph $G$ is edge-transitive if and only if $G \setminus e_1 \simeq G \setminus e_2$ for all $e_1, e_2 \in E$ \cite{andersen1992edge}. Both the complete graphs and the star graphs are edge-transitive. An example of a non edge-transitive graph is given in the \emph{path graph} $P_5$ \cref{fig:P5}.

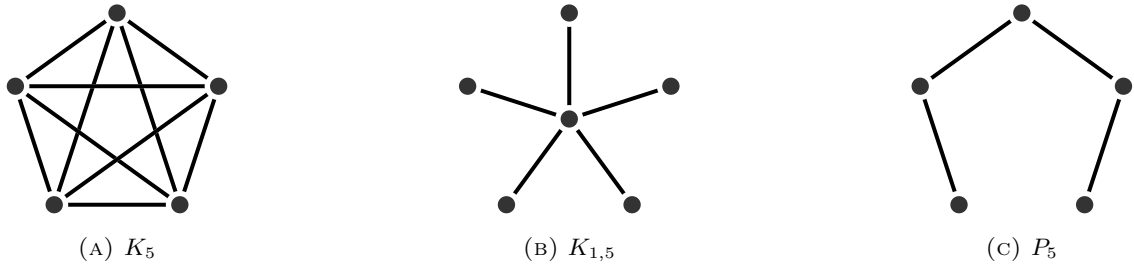
\begin{figure}[!ht]
    \centering
    \begin{subfigure}[b]{0.3\textwidth}
        \centering
        \begin{tikzpicture}[site/.style = {circle,
                                           draw = white,
                                           line width = 2pt,
                                           fill = black!80!white,
                                           inner sep = 3pt}]
        
            \coordinate (1) at (90:4em) {};
            \coordinate (2) at (162:4em) {};
            \coordinate (3) at (234:4em) {};
            \coordinate (4) at (306:4em) {};
            \coordinate (5) at (18:4em) {};

            \foreach \i in {1,...,5}
            \foreach \j in {\i,...,5} {
                \draw[-, line width = 1.5pt, draw = white] (\j.center) -- (\i.center);
                \draw[-, line width = 1.5pt, draw = black] (\j.center) -- (\i.center);
            }

            \foreach \i in {1,...,5}
                \node[site] (v\i) at (\i.center) {};
        \end{tikzpicture}
        \caption{$K_5$}
        \label{fig:K5}
     \end{subfigure}
     \hfill
     \begin{subfigure}[b]{0.3\textwidth}
         \centering
         \begin{tikzpicture}[site/.style = {circle,
                                            draw = white,
                                            line width = 2pt,
                                            fill = black!80!white,
                                            inner sep = 3pt}]

            \coordinate (0);
            \coordinate (1) at (90:4em);
            \coordinate (2) at (162:4em);
            \coordinate (3) at (234:4em);
            \coordinate (4) at (306:4em);
            \coordinate (5) at (18:4em);

            \foreach \i in {1,...,5} {
                \draw[-, line width = 1.5pt, draw = white] (0.center) -- (\i.center);
                \draw[-, line width = 1.5pt, draw = black] (0.center) -- (\i.center);
            }

            \foreach \i in {0,...,5}
                \node[site] (v\i) at (\i.center) {};
        \end{tikzpicture}
         \caption{$K_{1,5}$}
         \label{fig:S6}
     \end{subfigure}
     \hfill
     \begin{subfigure}[b]{0.3\textwidth}
         \centering
         \begin{tikzpicture}[site/.style = {circle,
                                           draw = white,
                                           line width = 2pt,
                                           fill = black!80!white,
                                           inner sep = 3pt}]

            \coordinate (0) at (90:4em) {};
            \coordinate (1) at (162:4em) {};
            \coordinate (2) at (234:4em) {};
            \coordinate (3) at (306:4em) {};
            \coordinate (4) at (18:4em) {};

            \draw[-, line width = 1.5pt, draw = white] (0.center) -- (1.center);
            \draw[-, line width = 1.5pt, draw = black] (0.center) -- (1.center);

            \draw[-, line width = 1.5pt, draw = white] (1.center) -- (2.center);
            \draw[-, line width = 1.5pt, draw = black] (1.center) -- (2.center);

            \draw[-, line width = 1.5pt, draw = white] (3.center) -- (4.center);
            \draw[-, line width = 1.5pt, draw = black] (3.center) -- (4.center);

            \draw[-, line width = 1.5pt, draw = white] (4.center) -- (0.center);
            \draw[-, line width = 1.5pt, draw = black] (4.center) -- (0.center);

            \foreach \i in {0,...,4}
                \node[site] (v\i) at (\i.center) {};
        \end{tikzpicture}
         \caption{$P_5$}
         \label{fig:P5}
     \end{subfigure}
     \caption{The complete graph $K_5$ with $\mathrm{Aut}(K_5) \simeq S_5$, the star graph $K_{1,5}$ with $\mathrm{Aut}(K_{1,5}) \simeq S_5$, and the path graph $P_5$ with $\mathrm{Aut}(P_5) \simeq \mathbb{Z}_2$.}
\end{figure}

Let $\H \defeq \C^d$ denote a complex \emph{Euclidean space} of dimension $d \geq 2$. A complex $d \times d$ matrix $H$ is \emph{Hermitian} if $H\ct = H$, where $H\ct$ is the \emph{conjugate transpose} of $H$. The collection of all Hermitian matrices acting on $\H$ is denoted as $\Herm(\H)$. For a Hermitian matrix $H \in \Herm(\H)$, we use the notation $H \succeq 0$ to indicate that $H$ is \emph{positive semi-definite}, i.e.~$\bra{\psi} H \ket{\psi} \geq 0$ for all $\ket{\psi} \in \H$. A Hermitian matrix $H \succeq 0$ is a \emph{projector} if $H^2 = H$.
A \emph{quantum state} on $\H$ is a matrix $\rho \in \Herm(\H)$ such that $\rho \succeq 0$ and $\Tr \rho = 1$. The collection of all quantum states acting on $\H$ is denoted as $\mathcal{D}_d$.

For any graph $G$, we will associate a separate copy of $\H$ to each vertex of $G$ (we will refer to $d = \dim \H$ as \emph{local dimension}). The combined space associated to $V$ is then the $n$-fold tensor power $\H\xp{n}$. If $\rho$ is a quantum state on $\H\xp{n}$ and $e = (u,v) \in E$ an edge, we will denote by $\rho_e$ the \emph{reduced state} on systems $u$ and $v$:
\begin{equation}
    \rho_e \defeq \Tr_{V \setminus \set{u,v}} \rho .
\end{equation}

In quantum information literature, a bipartite quantum state $\rho$ on $\mathcal{H}_A \otimes \mathcal{H}_B$ is called $n$-\emph{extendible} with respect to $\mathcal{H}_B$ if there exists a quantum state $\sigma$ on $\mathcal{H}_A \otimes \mathcal{H}^{\otimes n}_B$, invariant under any permutation of the $\mathcal{H}_B$ subsystems, such that
\begin{equation}
    \rho = \Tr_{B^{\otimes n-1}} \sigma.
\end{equation}
We can express this concept as a \emph{marginal problem} within the context of the star graph $K_{1,n}$, where the central vertex corresponds to the system $\mathcal{H}_A$, and the leaves represent the subsystems $\mathcal{H}_B$. A quantum state $\rho$ is $n$-extendible if and only if there exists a state on $K_{1,n}$ with reduced states along all edges equal $\rho$. Therefore, in the current paper, we refer to $n$-\emph{extendibility} as $K_{1,n}$-\emph{extendibility}.

Let $\W, \I, \F \in \Herm(\H\x\H)$ denote the unnormalized \emph{maximally entangled} and \emph{maximally mixed} states, and the \emph{flip} operator on two systems:
\begin{equation}
    \W \defeq \sum^d_{i,j = 1} \ketbra{ii}{jj}, \quad \I \defeq \sum^d_{i,j = 1} \ketbra{ij}{ij} \quad \text{and} \quad \F \defeq \sum^d_{i,j = 1} \ketbra{ij}{ji}.
\end{equation}
Note that $\Tr \W = d$, $\Tr \I = d^2$ and $\Tr \F = d$.

%%%%%%%%%%%%%%%%%%%%%%%%%%%%%%%%%%%%%%%%%%%%%%%%%%%%%%%%%%%%%%%%%%%%%%%%%%%%%%%%%%%%%%%%%%%%%%%%%%%%%%%%%%%%%%%%%%%%%%%%%%%%%%%%%%%%%%%%%%%%%%%%%%%%%%%%%%%%%%%%%%%%%%%%%%%%%%%%%%%%%%%%%%%%%%%%%%%%%%%%%%%%%%%%%%%%%%%%%%
%%%%%%%%%%%%%%%%%%%%%%%%%%%%%%%%%%%%%%%%%%%%%%%%%%%%%%%%%%%%%%%%%%%%%%%%%%%%%%%%%%%%%%%%%%%%%%%%%%%%%%%%%%%%%%%%%%%%%%%%%%%%%%%%%%%%%%%%%%%%%%%%%%%%%%%%%%%%%%%%%%%%%%%%%%%%%%%%%%%%%%%%%%%%%%%%%%%%%%%%%%%%%%%%%%%%%%%%%
\subsection{Schur--Weyl dualities for unitary and orthogonal groups}

Before introducing Schur--Weyl dualities for unitary and orthogonal groups, we need to explain \emph{Young diagrams}. They are combinatorial objects which are used to label irreducible representations of unitary and orthogonal groups.

\subsubsection{Young diagrams}

Let $\lambda \vdash n$ denote an ordered \emph{partition} of a positive integer $n$ into $l$ parts, i.e.~a non-increasing sequence of positive integers $(\lambda_1, \ldots, \lambda_l)$ that sum up to $n$:
\begin{equation}
    \lambda_1 \geq \cdots \geq \lambda_l \geq 1 \qquad \text{and} \qquad \sum^l_{i = 1} \lambda_i = n.
\end{equation}
A partition $\lambda \vdash n$ can be represented as a \emph{Young diagram}, which is a set of $n$ boxes arranged in rows, from top to bottom, that are justified to the left, and where the $i$-th row contains $\lambda_i$ boxes. The \emph{conjugate} of the partition $\lambda$, denoted $\lambda^\prime$, is the partition corresponding to transposing the Young diagram representing $\lambda$. E.g.~if $\lambda \defeq (3, 1)$ then
\begin{equation}
    \lambda = \ydiagram{3, 1} \qquad \text{and} \qquad \lambda^\prime = \ydiagram{2, 1, 1}.
\end{equation}

For a Young diagram $\lambda$, let $u \in \lambda$ be a box at row $i$ and column $j$ of $\lambda$. The \emph{content} of $u$, denoted by $\cont(u)$, is defined as $\cont(u) \defeq j - i$. The total content of the Young diagram $\lambda$, denoted as $\cont(\lambda)$, is defined as the sum of the content of all boxes in $\lambda$. E.g.~if $\lambda \defeq (3, 3, 1)$ then
\begin{equation}
    \lambda = \ydiagram{3, 3, 1} \qquad \text{and} \qquad \text{contents of $\lambda$ are} \quad
    \begin{ytableau}
        0 & 1 & 2 \\
        \shortminus 1 & 0 & 1 \\
        \shortminus 2
    \end{ytableau}\,,
\end{equation}
so the total content of $\lambda$ is $\cont (\lambda) = 1$.

%%%%%%%%%%%%%%%%%%%%%%%%%%%%%%%%%%%%%%%%%%%%%%%%%%%%%%%%%%%%%%%%%%%%%%%%%%%%%%%%%%%%%%%%%%%%%%%%%%%%%%%%%%%%%%%%%%%%%%%%%%%%%%%%%%%%%%%%%%%%%%%%%%%%%%%%%%%%%%%%%%%%%%%%%%%%%%%%%%%%%%%%%%%%%%%%%%%%%%%%%%%%%%%%%%%%%%%%%%
\subsubsection{Schur--Weyl duality for the unitary group}\label{sec:Schur_Weyl_duality_unitary}

A classic result in representation theory is \emph{Schur--Weyl duality}. It states that the commutant of the diagonal action of the \emph{unitary group} $\U_d$ (or equivalently, the \emph{general linear group} $\GL_d(\C)$) on $(\C^d)\xp{n}$ is the linear span of the image of the \emph{symmetric group} $\S_n$. More precisely, define a diagonal action $\phi : \GL_d(\C) \rightarrow \End((\C^d)\xp{n})$ as
\begin{equation}\label{def:phi}
    \phi(g) \ket{x_1} \otimes \ket{x_2} \otimes \dotsb \otimes \ket{x_n}  \defeq g\ket{x_1} \otimes g\ket{x_2} \otimes \dotsb \otimes g\ket{x_n}
\end{equation}
for every $g \in \GL_d(\C)$. We will be interested in the matrix algebra 
\begin{equation}
    \cU^d_n \defeq \spansp \set*{\phi(g) \mid g \in \GL_d(\C)},
\end{equation}
which can equivalently be defined by considering only $g \in \U_d(\C) \subset \GL_d(\C)$, i.e., $\cU^d_n = \spansp \set*{\phi(g) \mid g \in \U_d(\C)}$ \cite{aubrun,grinko2022linear}. Similarly, there is an action $\psi: \S_n \rightarrow \End((\C^d)\xp{n})$ of the symmetric group $\S_n$ on $(\C^d)\xp{n}$ defined by:
\begin{equation}
    \psi(\pi) \ket{x_1} \otimes \ket{x_2} \otimes \dotsb \otimes \ket{x_n} \defeq \ket{x_{\pi^{-1}(1)}} \otimes \ket{x_{\pi^{-1}(2)}} \otimes \dotsb \otimes \ket{x_{\pi^{-1}(n)}}
\end{equation}
for every permutation $\pi \in \S_n$. We will be interested in the matrix algebra $\cS^d_n$ spanned by the image $\psi(\S_n)$, namely, 
\begin{equation}
    \cS^d_n \defeq  \spansp \set*{\psi(\pi) \mid \pi \in \S_n}.
\end{equation}

Now let $\Irr{\cS^d_n}$ be the set of all \emph{irreducible representations} (irreps, for short) of the algebra $\cS^d_n$. As a quotient algebra of the symmetric group algebra $\C\S_n$, its irreducible representations are also labelled by Young diagrams:
\begin{equation}
    \Irr{\cS^d_n} \defeq \set*{ \mu \pt n \mid \mu'_1 \leq d}.
\end{equation}
We will denote by $V^{\cS^d_n}_{\mu}$ and $V^{\cU^d_n}_{\mu}$ the corresponding irreducible representations of $\cS^d_n$ and $\cU^d_n$. Now we are ready to formally state the Schur--Weyl duality.

\begin{theorem}[Schur--Weyl duality]
    The matrix algebras $\cS^d_n$ and $\cU^d_n$ are mutual commutants of each other. Equivalently, the space $(\C^d)\xp{n}$ decomposes into isotypic sectors labelled by $\mu \in \Irr{\cS_n^d}$ consisting of tensor product of irreducible representations of $\cS^d_n$ and $\cU^d_n$:
        \begin{equation}
            (\C^d)\xp{n} \simeq \bigoplus_{\mu \in \Irr{\cS_n^d}} V^{\cS^d_n}_{\mu} \otimes V^{\cU^d_n}_{\mu}.
        \end{equation}
\end{theorem}

%%%%%%%%%%%%%%%%%%%%%%%%%%%%%%%%%%%%%%%%%%%%%%%%%%%%%%%%%%%%%%%%%%%%%%%%%%%%%%%%%%%%%%%%%%%%%%%%%%%%%%%%%%%%%%%%%%%%%%%%%%%%%%%%%%%%%%%%%%%%%%%%%%%%%%%%%%%%%%%%%%%%%%%%%%%%%%%%%%%%%%%%%%%%%%%%%%%%%%%%%%%%%%%%%%%%%%%%%%
\subsubsection{Schur--Weyl duality for the orthogonal group} \label{sec:Schur_Weyl_duality_orthogonal}

An analogous Schur--Weyl duality result was discovered by Richard Brauer for \emph{orthogonal} and \emph{symplectic groups} \cite{brauer1937algebras}. In the following, we only focus on the complex orthogonal group $\Orth_d(\C)$. Brauer's variant of Schur--Weyl duality states that the commutant of the diagonal action of the complex orthogonal group $\Orth_d(\C)$ on the space $(\C^d)\xp{n}$ is the image of the \emph{Brauer algebra} $\Br^d_n$. More precisely, using the same diagonal action $\phi$ from \cref{def:phi} for the subgroup $\Orth_d(\C) \subset \GL_d(\C)$
\begin{equation}
    \phi(g) \ket{x_1} \otimes \ket{x_2} \otimes \dotsb \otimes \ket{x_n} = g\ket{x_1} \otimes g\ket{x_2} \otimes \dotsb \otimes g\ket{x_n}
\end{equation}
for every $g \in \Orth_d(\C)$, we define the matrix algebra 
\begin{equation}
    \cO^d_n \defeq \spansp \set*{\phi(g) \mid g \in \Orth_d(\C)}.
\end{equation}
The commutant of $\cO^d_n$ will turn out to be the image of the Brauer algebra $\Br^d_n$ which we define now.

A \emph{Brauer diagram} is a diagram with two columns of $n$ vertices which are paired up in an arbitrary way, i.e.~it is a pairing on a set of $2n$ elements. The set of all Brauer diagrams is denoted $\Br_n$. For example, $\pi, \sigma \in \Br_5$ may look like
\begin{equation}
    \begin{tikzpicture}[site/.style = {circle,
                                       draw = white,
                                       outer sep = 0.5pt,
                                       fill = black!80!white,
                                       inner sep = 1.25pt}]
        \node (lm1) {};
        \node (lm2) [yshift = -0.75em] at (lm1.center) {};
        \node (lm3) [yshift = -0.75em] at (lm2.center) {};
        \node (lm4) [yshift = -0.75em] at (lm3.center) {};
        \node (lm5) [yshift = -0.75em] at (lm4.center) {};

        \node (mm) [xshift = 2.5em, yshift = 0.125em] at (lm3.center) { and };

        \node (rm1) [xshift = 7em] at (lm1.center) {};
        \node (rm2) [yshift = -0.75em] at (rm1.center) {};
        \node (rm3) [yshift = -0.75em] at (rm2.center) {};
        \node (rm4) [yshift = -0.75em] at (rm3.center) {};
        \node (rm5) [yshift = -0.75em] at (rm4.center) {};

        \node[site] (ll1) [xshift = -0.75em] at (lm1.center) {};
        \node[site] (lr1) [xshift = 0.75em] at (lm1.center) {};

        \node[site] (ll2) [yshift = -0.75em] at (ll1.center) {};
        \node[site] (lr2) [yshift = -0.75em] at (lr1.center) {};

        \node[site, label={left:{$\pi =$}}] (ll3) [yshift = -0.75em] at (ll2.center) {};
        \node[site] (lr3) [yshift = -0.75em] at (lr2.center) {};

        \node[site] (ll4) [yshift = -0.75em] at (ll3.center) {};
        \node[site] (lr4) [yshift = -0.75em] at (lr3.center) {};

        \node[site] (ll5) [yshift = -0.75em] at (ll4.center) {};
        \node[site] (lr5) [yshift = -0.75em] at (lr4.center) {};

        \draw[-, line width = 2.5pt, draw = white] (lr3) .. controls (lm3) and (lm5) .. (lr5);
        \draw[-, line width = 1.25pt, draw = black] (lr3) .. controls (lm3) and (lm5) .. (lr5);

        \draw[-, line width = 2.5pt, draw = white] (lr1) .. controls (lm1) and (lm4) .. (lr4);
        \draw[-, line width = 1.25pt, draw = black] (lr1) .. controls (lm1) and (lm4) .. (lr4);

        \draw[-, line width = 2.5pt, draw = white] (ll1) .. controls (lm1) and (lm2) .. (ll2);
        \draw[-, line width = 1.25pt, draw = black] (ll1) .. controls (lm1) and (lm2) .. (ll2);

        \draw[-, line width = 2.5pt, draw = white] (ll4) .. controls (lm4) and (lm5) .. (ll5);
        \draw[-, line width = 1.25pt, draw = black] (ll4) .. controls (lm4) and (lm5) .. (ll5);

        \draw[-, line width = 2.5pt, draw = white] (ll3) to (lr2);
        \draw[-, line width = 1.25pt, draw = black] (ll3) to (lr2);

        \node[site] (rl1) [xshift = -0.75em] at (rm1.center) {};
        \node[site] (rr1) [xshift = 0.75em] at (rm1.center) {};

        \node[site] (rl2) [yshift = -0.75em] at (rl1.center) {};
        \node[site] (rr2) [yshift = -0.75em] at (rr1.center) {};

        \node[site, label={left:{$\sigma =$}}] (rl3) [yshift = -0.75em] at (rl2.center) {};
        \node[site] (rr3) [yshift = -0.75em] at (rr2.center) {};

        \node[site] (rl4) [yshift = -0.75em] at (rl3.center) {};
        \node[site] (rr4) [yshift = -0.75em] at (rr3.center) {};

        \node[site] (rl5) [yshift = -0.75em] at (rl4.center) {};
        \node[site] (rr5) [yshift = -0.75em] at (rr4.center) {};

        \draw[-, line width = 2.5pt, draw = white] (rr1) .. controls (rm1) and (rm2) .. (rr2);
        \draw[-, line width = 1.25pt, draw = black] (rr1) .. controls (rm1) and (rm2) .. (rr2);

        \draw[-, line width = 2.5pt, draw = white] (rr3) .. controls (rm3) and (rm5) .. (rr5);
        \draw[-, line width = 1.25pt, draw = black] (rr3) .. controls (rm3) and (rm5) .. (rr5);

        \draw[-, line width = 2.5pt, draw = white] (rl3) .. controls (rm3) and (rm5) .. (rl5);
        \draw[-, line width = 1.25pt, draw = black] (rl3) .. controls (rm3) and (rm5) .. (rl5);

        \draw[-, line width = 2.5pt, draw = white] (rl1) .. controls (rm1) and (rm4) .. (rl4);
        \draw[-, line width = 1.25pt, draw = black] (rl1) .. controls (rm1) and (rm4) .. (rl4);

        \draw[-, line width = 2.5pt, draw = white] (rl2) to (rr4);
        \draw[-, line width = 1.25pt, draw = black] (rl2) to (rr4);
    \end{tikzpicture}
\end{equation}
The Brauer algebra $\Br^d_n$ for any $d \in \C$ is defined as the complex vector space spanned by all Brauer diagrams $\pi \in \Br_n$, i.e.~$\Br^d_n \defeq \spansp \set*{\pi \in \Br_n}$. The multiplication in $\Br^d_n$ is defined by concatenation of such diagrams, counting the number of loops $l$ formed, erasing them and multiplying the result by a factor of $d^l$. For example, the multiplication of $\pi, \sigma \in \Br^d_5$ defined above gives

\begin{equation}
    \begin{tikzpicture}[site/.style = {circle,
                                       draw = white,
                                       outer sep = 0.5pt,
                                       fill = black!80!white,
                                       inner sep = 1.25pt}]
        \node (lm1) {};
        \node (lm2) [yshift = -0.75em] at (lm1.center) {};
        \node (lm3) [yshift = -0.75em] at (lm2.center) {};
        \node (lm4) [yshift = -0.75em] at (lm3.center) {};
        \node (lm5) [yshift = -0.75em] at (lm4.center) {};

        \node (mm1) [xshift = 3.15em] at (lm1.center) {};
        \node (mm2) [yshift = -0.75em] at (mm1.center) {};
        \node (mm3) [yshift = -0.75em] at (mm2.center) {};
        \node (mm4) [yshift = -0.75em] at (mm3.center) {};
        \node (mm5) [yshift = -0.75em] at (mm4.center) {};

        \node (rm1) [xshift = 4.5em] at (mm1.center) {};
        \node (rm2) [yshift = -0.75em] at (rm1.center) {};
        \node (rm3) [yshift = -0.75em] at (rm2.center) {};
        \node (rm4) [yshift = -0.75em] at (rm3.center) {};
        \node (rm5) [yshift = -0.75em] at (rm4.center) {};

        \node[site] (ll1) [xshift = -0.75em] at (lm1.center) {};
        \node[site] (lr1) [xshift = 0.75em] at (lm1.center) {};

        \node[site] (ll2) [yshift = -0.75em] at (ll1.center) {};
        \node[site] (lr2) [yshift = -0.75em] at (lr1.center) {};

        \node[site, label={left:{$\pi \cdot \sigma =$}}] (ll3) [yshift = -0.75em] at (ll2.center) {};
        \node[site] (lr3) [yshift = -0.75em] at (lr2.center) {};

        \node[site] (ll4) [yshift = -0.75em] at (ll3.center) {};
        \node[site] (lr4) [yshift = -0.75em] at (lr3.center) {};

        \node[site] (ll5) [yshift = -0.75em] at (ll4.center) {};
        \node[site] (lr5) [yshift = -0.75em] at (lr4.center) {};

        \draw[-, line width = 2.5pt, draw = white] (lr3) .. controls (lm3) and (lm5) .. (lr5);
        \draw[-, line width = 1.25pt, draw = gray] (lr3) .. controls (lm3) and (lm5) .. (lr5);

        \draw[-, line width = 2.5pt, draw = white] (lr1) .. controls (lm1) and (lm4) .. (lr4);
        \draw[-, line width = 1.25pt, draw = gray] (lr1) .. controls (lm1) and (lm4) .. (lr4);

        \draw[-, line width = 2.5pt, draw = white] (ll1) .. controls (lm1) and (lm2) .. (ll2);
        \draw[-, line width = 1.25pt, draw = black] (ll1) .. controls (lm1) and (lm2) .. (ll2);

        \draw[-, line width = 2.5pt, draw = white] (ll4) .. controls (lm4) and (lm5) .. (ll5);
        \draw[-, line width = 1.25pt, draw = black] (ll4) .. controls (lm4) and (lm5) .. (ll5);

        \draw[-, line width = 2.5pt, draw = white] (ll3) to (lr2);
        \draw[-, line width = 1.25pt, draw = black] (ll3) to (lr2);

        \node[site] (ml1) [xshift = -0.75em] at (mm1.center) {};
        \node[site] (mr1) [xshift = 0.75em] at (mm1.center) {};

        \node[site] (ml2) [yshift = -0.75em] at (ml1.center) {};
        \node[site] (mr2) [yshift = -0.75em] at (mr1.center) {};

        \node[site] (ml3) [yshift = -0.75em] at (ml2.center) {};
        \node[site, label={[yshift=0.2em]right:{$= d^2$}}] (mr3) [yshift = -0.75em] at (mr2.center) {};

        \node[site] (ml4) [yshift = -0.75em] at (ml3.center) {};
        \node[site] (mr4) [yshift = -0.75em] at (mr3.center) {};

        \node[site] (ml5) [yshift = -0.75em] at (ml4.center) {};
        \node[site] (mr5) [yshift = -0.75em] at (mr4.center) {};

        \draw[-, line width = 2.5pt, draw = white] (mr1) .. controls (mm1) and (mm2) .. (mr2);
        \draw[-, line width = 1.25pt, draw = black] (mr1) .. controls (mm1) and (mm2) .. (mr2);

        \draw[-, line width = 2.5pt, draw = white] (mr3) .. controls (mm3) and (mm5) .. (mr5);
        \draw[-, line width = 1.25pt, draw = black] (mr3) .. controls (mm3) and (mm5) .. (mr5);

        \draw[-, line width = 2.5pt, draw = white] (ml3) .. controls (mm3) and (mm5) .. (ml5);
        \draw[-, line width = 1.25pt, draw = gray] (ml3) .. controls (mm3) and (mm5) .. (ml5);

        \draw[-, line width = 2.5pt, draw = white] (ml1) .. controls (mm1) and (mm4) .. (ml4);
        \draw[-, line width = 1.25pt, draw = gray] (ml1) .. controls (mm1) and (mm4) .. (ml4);

        \draw[-, line width = 2.5pt, draw = white] (ml2) to (mr4);
        \draw[-, line width = 1.25pt, draw = black] (ml2) to (mr4);

        \draw[-, dotted, line width = 1.5pt, draw = gray] (lr1) to (ml1);
        \draw[-, dotted, line width = 1.5pt, draw = gray] (lr2) to (ml2);
        \draw[-, dotted, line width = 1.5pt, draw = gray] (lr3) to (ml3);
        \draw[-, dotted, line width = 1.5pt, draw = gray] (lr4) to (ml4);
        \draw[-, dotted, line width = 1.5pt, draw = gray] (lr5) to (ml5);

        \node[site] (rl1) [xshift = -0.75em] at (rm1.center) {};
        \node[site] (rr1) [xshift = 0.75em] at (rm1.center) {};

        \node[site] (rl2) [yshift = -0.75em] at (rl1.center) {};
        \node[site] (rr2) [yshift = -0.75em] at (rr1.center) {};

        \node[site] (rl3) [yshift = -0.75em] at (rl2.center) {};
        \node[site] (rr3) [yshift = -0.75em] at (rr2.center) {};

        \node[site] (rl4) [yshift = -0.75em] at (rl3.center) {};
        \node[site] (rr4) [yshift = -0.75em] at (rr3.center) {};

        \node[site] (rl5) [yshift = -0.75em] at (rl4.center) {};
        \node[site] (rr5) [yshift = -0.75em] at (rr4.center) {};

        \draw[-, line width = 2.5pt, draw = white] (rr1) .. controls (rm1) and (rm2) .. (rr2);
        \draw[-, line width = 1.25pt, draw = black] (rr1) .. controls (rm1) and (rm2) .. (rr2);

        \draw[-, line width = 2.5pt, draw = white] (rr3) .. controls (rm3) and (rm5) .. (rr5);
        \draw[-, line width = 1.25pt, draw = black] (rr3) .. controls (rm3) and (rm5) .. (rr5);

        \draw[-, line width = 2.5pt, draw = white] (rl1) .. controls (rm1) and (rm2) .. (rl2);
        \draw[-, line width = 1.25pt, draw = black] (rl1) .. controls (rm1) and (rm2) .. (rl2);

        \draw[-, line width = 2.5pt, draw = white] (rl4) .. controls (rm4) and (rm5) .. (rl5);
        \draw[-, line width = 1.25pt, draw = black] (rl4) .. controls (rm4) and (rm5) .. (rl5);

        \draw[-, line width = 2.5pt, draw = white] (rl3) to (rr4);
        \draw[-, line width = 1.25pt, draw = black] (rl3) to (rr4);
    \end{tikzpicture}
\end{equation}
where the number of grey loops removed is $2$, so the resulting factor is $d^2$.
Note that the formal complex linear span of the symmetric group $\S_n$, i.e., the symmetric group algebra $\C\S_n$, is a subalgebra of the Brauer algebra, $\C\S_n \subset \Br^d_n$, consisting only of diagrams which do not have vertical pairings on the same side of the diagram.

The Brauer algebra $\Br^d_n$ admits a action on $(\C^d)\xp{n}$ defined via a linear map $\psi: \Br^d_n \rightarrow (\C^d)\xp{n}$, such that for every diagram $\pi \in \Br^d_n$ and strings $x_{\bar{1}},\dotsc,x_{\bar{n}},x_1,\dotsc,x_n \in [d]$,
\begin{equation}
    \of[\big]{\bra{x_{\bar{1}}} \otimes \dotsb \otimes \bra{x_{\bar{n}}}} \psi(\pi) \of[\big]{\ket{x_1} \otimes \dotsb \otimes \ket{x_n}}  \defeq
    \begin{cases}
    1 &\text{if $x_k = x_l$ for all vertices $k$ and $l$ connected in $\pi$}, \\
    0 &\text{otherwise}.
    \end{cases} 
\end{equation}
For example, for all $x_{\bar{1}}, x_{\bar{2}}, x_{\bar{3}}, x_1, x_2, x_3 \in [d]$ and given the diagram $\pi \in \Br_3$ corresponding to the pairing $\{1, 3\}, \{2, \bar{3}\}, \{\bar{1}, \bar{2}\}$,
\begin{align*}
    \of[\big]{\bra{x_{\bar{1}}} \otimes \bra{x_{\bar{2}}} \otimes \bra{x_{\bar{3}}}} \psi(\pi) \of[\big]{\ket{x_1} \otimes  \ket{x_2} \otimes  \ket{x_3}} &=
        \begin{tikzpicture}[baseline = (baseline.center),
                            site/.style = {circle,
                                       draw = white,
                                       outer sep = 0.5pt,
                                       fill = black!80!white,
                                       inner sep = 1.25pt}]
        \node (m1) {};
        \node (m2) [yshift = -0.75em] at (m1.center) {};
        \node (m3) [yshift = -0.75em] at (m2.center) {};
% do not remove
        \node (baseline) [yshift = -0.2em] at (m2.center) {};
% do not remove
        \node[site, label={left:{\footnotesize$x_{\bar{1}}$}}] (l1) [xshift = -0.75em] at (m1.center) {};
        \node[site, label={right:{\footnotesize$x_1$}}] (r1) [xshift = 0.75em] at (m1.center) {};
% do not remove
        \node[site, label={left:{\footnotesize$x_{\bar{2}}$}}] (l2) [yshift = -0.75em] at (l1.center) {};
        \node[site, label={right:{\footnotesize$x_2$}}] (r2) [yshift = -0.75em] at (r1.center) {};
% do not remove        
        \node[site, label={left:{\footnotesize$x_{\bar{3}}$}}] (l3) [yshift = -0.75em] at (l2.center) {};
        \node[site, label={right:{\footnotesize$x_3$}}] (r3) [yshift = -0.75em] at (r2.center) {};
% do not remove
        \draw[-, line width = 2.5pt, draw = white] (r1) .. controls (m1) and (m3) .. (r3);
        \draw[-, line width = 1.25pt, draw = black] (r1) .. controls (m1) and (m3) .. (r3);
% do not remove
        \draw[-, line width = 2.5pt, draw = white] (l1) .. controls (m1) and (m2) .. (l2);
        \draw[-, line width = 1.25pt, draw = black] (l1) .. controls (m1) and (m2) .. (l2);
% do not remove
        \draw[-, line width = 2.5pt, draw = white] (l3) to (r2);
        \draw[-, line width = 1.25pt, draw = black] (l3) to (r2);
    \end{tikzpicture} = \braket{x_1}{x_3} \cdot \braket{x_2}{x_{\bar{3}}} \cdot \braket{x_{\bar{1}}}{x_{\bar{2}}}. 
\end{align*}
The image $\psi(\Br_n^d)$ of the diagram algebra $\Br_n^d$ is a matrix representation of $\Br_n^d$ which we will denote by
\begin{equation}
    \cB^d_n \defeq  \spansp \set*{\psi(\pi) \mid \pi \in \Br_n^d}.
\end{equation}
The matrix algebra $\cB^d_n$ is always \emph{semisimple}, while the diagram algebra $\Br_n^d$ is not for small integer $d$ \cite{wenzl1988structure,doran1999semisimplicity,rui2005criterion,rui2006criterion,andersen2017semisimplicity}. The irreducible representations of $\cB^d_n$ are labelled using the following set \cite{okada2016pieri}:
\begin{equation}
    \Irr{\cB^d_n} \defeq \set[\bigg]{\lambda \vdash n - 2r \:\bigg\lvert\: r \in \Big\{ 0, \ldots, \Big\lfloor \frac{n}{2} \Big\rfloor \Big\} \text{ and } \lambda^\prime_1 + \lambda^\prime_2 \leq d }.
\end{equation}
We will denote by $W^{\cB^d_n}_{\lambda}$ and $W^{\cO^d_n}_{\lambda}$ the spaces on which the corresponding irreducible representations of $\cB^d_n$ and $\cO^d_n$ act. Now we are ready to state the version of Schur--Weyl duality discovered by Brauer.

\begin{theorem}[Brauer]
    The matrix algebras $\cB^d_n$ and $\cO^d_n$ are mutual commutants of each other. Equivalently, the space $(\C^d)\xp{n}$ decomposes into isotypic sectors labelled by $\lambda \in \Irr{\cB_n^d}$ consisting of tensor product of irreducible representations of $\cB^d_n$ and $\cO^d_n$:
    \begin{equation}
        (\C^d)\xp{n} \simeq \bigoplus_{\lambda \in \Irr{\cB_n^d}} W^{\cB^d_n}_{\lambda} \otimes W^{\cO^d_n}_{\lambda}.
    \end{equation}
\end{theorem}

%%%%%%%%%%%%%%%%%%%%%%%%%%%%%%%%%%%%%%%%%%%%%%%%%%%%%%%%%%%%%%%%%%%%%%%%%%%%%%%%%%%%%%%%%%%%%%%%%%%%%%%%%%%%%%%%%%%%%%%%%%%%%%%%%%%%%%%%%%%%%%%%%%%%%%%%%%%%%%%%%%%%%%%%%%%%%%%%%%%%%%%%%%%%%%%%%%%%%%%%%%%%%%%%%%%%%%%%%%
%%%%%%%%%%%%%%%%%%%%%%%%%%%%%%%%%%%%%%%%%%%%%%%%%%%%%%%%%%%%%%%%%%%%%%%%%%%%%%%%%%%%%%%%%%%%%%%%%%%%%%%%%%%%%%%%%%%%%%%%%%%%%%%%%%%%%%%%%%%%%%%%%%%%%%%%%%%%%%%%%%%%%%%%%%%%%%%%%%%%%%%%%%%%%%%%%%%%%%%%%%%%%%%%%%%%%%%%%%
\subsection{Werner, isotropic and Brauer states} \label{sec:WernerIsotropicBrauerStates}

In this paper, we will consider three classes of symmetric bipartite states $\rho_{AB}$ based on different commutation relations:
\begin{enumerate}
    \item a \emph{Werner state} $\rho_{AB}$ commutes with $U \otimes U$, i.e., $\big[ \rho, U \otimes U \big] = 0$ for every $U \in \U_d$,
    \item an \emph{isotropic state} $\rho_{AB}$ commutes with $\bar{U} \otimes U$, i.e., $\big[ \rho, \bar{U} \otimes U \big] = 0$ for every $U \in \U_d$,
    \item a \emph{Brauer state}\footnote{More precisely, in this paper we are talking about \emph{orthogonal Brauer} states. In the case of the symplectic group, the commutant is also a Brauer algebra, so the corresponding states would be called \emph{symplectic Brauer} states.} $\rho_{AB}$ commutes with $O \otimes O$, i.e., $\big[ \rho, O \otimes O  \big] = 0$ for every $O \in \Orth_d(\C)$.
\end{enumerate}
Using the Schur--Weyl dualities from \cref{sec:Schur_Weyl_duality_unitary,sec:Schur_Weyl_duality_orthogonal}, we can observe that these states are just linear combinations of the operators $\I$, $\F$, $\W$. For our purposes, it will be better to parameterize them in terms of the projectors onto irreducible representations of the Brauer algebra $\cB_2^d$:
\begin{equation}
    \Pi_{\0} \defeq \frac{\W}{d}, \quad \Pi_{\ydsm{1,1}} \defeq \frac{\I-\F}{2}, \quad \Pi_{\ydsm{2}} \defeq \frac{\I+\F}{2} - \frac{\W}{d},
\end{equation}
and the projectors onto the irreducible representations of $\cS_2^d$, known as antisymmetric and symmetric subspaces, are
\begin{equation}
    \varepsilon_{\ydsm{1,1}} \defeq \frac{\I - \F}{2}, \quad \varepsilon_{\ydsm{2}} \defeq \frac{\I + \F}{2}.
\end{equation}
In representation theory, these projectors are also known as \emph{primitive central idempotents}. More generally, we can define a primitive central idempotent $\varepsilon_\lambda$ known as \emph{Young symmetrizer} which corresponds to an irreducible representation $\lambda$ of $\cS_n^d$. We also call $\varepsilon_\lambda$ an \emph{isotypic projector} onto the relevant sector $\lambda$ in the classical Schur--Weyl duality \cref{sec:Schur_Weyl_duality_unitary}. Note that
\begin{equation}
    \Pi_{\0} + \Pi_{\ydsm{1,1}} + \Pi_{\ydsm{2}} = \I, \quad \varepsilon_{\ydsm{1,1}} = \Pi_{\ydsm{1,1}}, \quad \varepsilon_{\ydsm{2}} = \I - \Pi_{\ydsm{1,1}} = \Pi_{\ydsm{2}} + \Pi_{\0},
\end{equation}
and
\begin{equation}
    \Tr \Pi_{\0} = 1, \quad
    \Tr \Pi_{\ydsm{1,1}} = \frac{d(d-1)}{2}, \quad
    \Tr \Pi_{\ydsm{2}} = \frac{d(d+1)}{2} - 1.
\end{equation}
The three classes of states can then be parameterized as follows:
\begin{enumerate}
    \item The Werner states form a one-parameter family, parameterized by $q \in [0,1]$, given by a convex combination of the normalised \emph{antisymmetric} and \emph{symmetric} projectors:
        \begin{equation} \label{eq:WernerState}
            q \cdot \tfrac{\Pi_{\ydsm{1,1}}}{\Tr \Pi_{\ydsm{1,1}}} + (1 - q) \cdot \tfrac{\I - \Pi_{\ydsm{1,1}}}{\Tr \of{\I - \Pi_{\ydsm{1,1}}}}.
        \end{equation}
    \item The isotropic states form a one-parameter family, parameterized by $p \in [0,1]$, given by a convex combination of the \emph{maximally entangled state} and the normalized projection onto its orthogonal complement:
        \begin{equation} \label{eq:isotropicStates_proj}
            p \cdot \Pi_{\0} + (1 - p) \cdot \tfrac{\I - \Pi_{\0}}{\Tr \of{\I - \Pi_{\0}}}.
        \end{equation}
    \item The Brauer states form a two-parameter family, parameterized by $p,q \in [0,1]$ with $p + q \leq 1$, given by a convex combination of all three normalized orthogonal projectors:
        \begin{equation} \label{eq:BrauerStates_proj}
            p \cdot \Pi_{\0} + q \cdot \tfrac{\Pi_{\ydsm{1,1}}}{\Tr \Pi_{\ydsm{1,1}}} + (1 - p - q) \cdot \tfrac{\Pi_{\ydsm{2}}}{\Tr \Pi_{\ydsm{2}}}.
        \end{equation}
\end{enumerate}

A Werner state $\rho_{AB}$ is separable if and only if $q \leq \frac{1}{2}$ \cite{werner1989quantum}, an isotropic state $\rho_{AB}$ is separable if and only if $p \leq \frac{1}{d}$ \cite{horodecki1999reduction}, and a Brauer state $\rho_{AB}$ is separable if and only if $q \leq \frac{1}{2}$ and $p \leq \frac{1}{d}$ (see \cref{app:BrauerStates}).

%%%%%%%%%%%%%%%%%%%%%%%%%%%%%%%%%%%%%%%%%%%%%%%%%%%%%%%%%%%%%%%%%%%%%%%%%%%%%%%%%%%%%%%%%%%%%%%%%%%%%%%%%%%%%%%%%%%%%%%%%%%%%%%%%%%%%%%%%%%%%%%%%%%%%%%%%%%%%%%%%%%%%%%%%%%%%%%%%%%%%%%%%%%%%%%%%%%%%%%%%%%%%%%%%%%%%%%%%%
%%%%%%%%%%%%%%%%%%%%%%%%%%%%%%%%%%%%%%%%%%%%%%%%%%%%%%%%%%%%%%%%%%%%%%%%%%%%%%%%%%%%%%%%%%%%%%%%%%%%%%%%%%%%%%%%%%%%%%%%%%%%%%%%%%%%%%%%%%%%%%%%%%%%%%%%%%%%%%%%%%%%%%%%%%%%%%%%%%%%%%%%%%%%%%%%%%%%%%%%%%%%%%%%%%%%%%%%%%
\subsection{Jucys--Murphy elements}\label{sec:JM_elemnts}

Certain special elements of the symmetric group algebra $\C\S_n$ and the Brauer algebra $\Br^d_n$, called \emph{Jucys--Murphy elements}, generate maximal commutative subalgebras inside the matrix algebras $\cS^d_n$ and $\cB^d_n$, respectively. They allow us, in principle, to build a representation theory of these algebras starting purely from the knowledge of their spectrum. This is called the Okounkov--Vershik approach \cite{OV1996,VO2005}.

For $\C\S_n$, the Jucys--Murphy elements $J^\S_k$ are defined as $J^\S_1 \defeq 0$ and, for every $k \in \set{2, \dotsc, n}$,
\begin{equation}
    J^\S_k \defeq \sum_{i=1}^{k-1} (i,k),
\end{equation}
where $(i,k)$ is the transposition of $i$ and $k$. Similarly, for the Brauer algebra $\Br^d_n$ the Jucys--Murphy elements $J^\Br_k$ are defined as $J^\Br_1 \defeq 0$ and, for every $k \in \set{2, \dotsc, n}$,
\begin{equation}
    J^\Br_k \defeq \sum_{i=1}^{k-1} (i,k) - \overline{(i,k)},
\end{equation}
where $\overline{(i,k)}$ is the vertical pairing vertices $i$ and $k$. The sum of all Jucys--Murphy elements is central in the corresponding matrix algebra. Its spectrum is well-known and summarized in the next two lemmas.

\begin{lemma}[\cite{doty2019canonical}] \label{lem:eigenvalueCentralElementsSn}
    Consider the following element of the matrix algebra $\cS^d_n$:
    \begin{equation}    
        J_\cS \defeq \sum_{k=1}^n \psi \of*{J^\S_k} = \sum_{1 \leq i < j \leq n} \F_{i,j}.
    \end{equation}
    Then, for any irreducible representation $\mu \in \Irr{\cS^d_n}$, the restriction of $J_\cS$ to $V_\mu$ is a multiple of the identity:
    \begin{equation}
        J_\cS|_{V_\mu} = \cont(\mu) \cdot I.
    \end{equation}
\end{lemma}

\begin{lemma}[\cite{doty2019canonical}] \label{lem:eigenvalueCentralElementsBn}
    Consider the following element of the matrix algebra $\cB^d_n$:
    \begin{equation}
        J_\cB \defeq \sum_{k=1}^n \psi \of*{J^\Br_k} = \sum_{1 \leq i < j \leq n} \of{\F_{i,j} - \W_{i,j}}.
    \end{equation}
    Then, for any irreducible representation $\lambda \in \Irr{\cB^d_n}$, the restriction of $J_\cB$ to $W_\lambda$ is a multiple of the identity:
    \begin{equation}
        J_\cB|_{W_\lambda} = \of*{ \cont(\lambda) - \frac{n - \abs{\lambda}}{2}(d-1) } \cdot I.
    \end{equation}
\end{lemma}

\begin{remark*}
    The commutation relation $[J_\cS, J_\cB] = 0$ holds for all $n$ and $d$. This implies in particular that $J_\cS$ and $J_\cB$ share a common eigenbasis and can be simultaneously diagonalized.
\end{remark*}

%%%%%%%%%%%%%%%%%%%%%%%%%%%%%%%%%%%%%%%%%%%%%%%%%%%%%%%%%%%%%%%%%%%%%%%%%%%%%%%%%%%%%%%%%%%%%%%%%%%%%%%%%%%%%%%%%%%%%%%%%%%%%%%%%%%%%%%%%%%%%%%%%%%%%%%%%%%%%%%%%%%%%%%%%%%%%%%%%%%%%%%%%%%%%%%%%%%%%%%%%%%%%%%%%%%%%%%%%%
%%%%%%%%%%%%%%%%%%%%%%%%%%%%%%%%%%%%%%%%%%%%%%%%%%%%%%%%%%%%%%%%%%%%%%%%%%%%%%%%%%%%%%%%%%%%%%%%%%%%%%%%%%%%%%%%%%%%%%%%%%%%%%%%%%%%%%%%%%%%%%%%%%%%%%%%%%%%%%%%%%%%%%%%%%%%%%%%%%%%%%%%%%%%%%%%%%%%%%%%%%%%%%%%%%%%%%%%%%
\subsection{Restrictions of representations of the Brauer algebra to the symmetric group algebra}

We saw in \cref{sec:Schur_Weyl_duality_orthogonal} that under the action of the Brauer algebra $\Br^d_n$, the space ${\big( \C^d \big)}^{\otimes n}$ decomposes into irreducible representations as
\begin{equation}
    {\big( \mathbb{C}^d \big)}^{\otimes n} \simeq \bigoplus_{\lambda \in \Irr{\cB^d_n}} W^{\cB_n^d}_\lambda \otimes W^{\cO^d_n}_\lambda.
\end{equation}
By \emph{restricting} the irreducible representations $W_\lambda$ of the algebra $\cB^d_n$ to the algebra $\cS^d_n$, the space ${\big( \mathbb{C}^d \big)}^{\otimes n}$ decomposes further as
\begin{align}
    {\big( \mathbb{C}^d \big)}^{\otimes n} &\simeq \bigoplus_{\lambda \in \Irr{\cB^d_n}} \mathrm{Res}^{\cB^d_n}_{\cS^d_n} \of[\Big]{ W^{\cB_n^d}_\lambda } \otimes W^{\cO^d_n}_\lambda \\
    &\simeq \bigoplus_{\lambda \in \Irr{\cB^d_n}} \of[\bigg]{\bigoplus_{\mu \in \Irr{\cS^d_n}} V^{\cS^d_n}_\mu \otimes \C^{m_{\lambda, \mu}}}  \otimes W^{\cO^d_n}_\lambda \\
    &\simeq \bigoplus_{\mu \in \Irr{\cS^d_n}} V^{\cS^d_n}_\mu  \otimes  \of[\bigg]{  \bigoplus_{\lambda \in \Irr{\cB^d_n}}W^{\cO^d_n}_\lambda \otimes \C^{m_{\lambda, \mu}}}.
\end{align}
The multiplicities $m_{\lambda, \mu}$ have no known concise formula. Even the set
\begin{equation} \label{def:okada_set}
    \Omega_{n,d} \defeq \big\{(\lambda, \mu) \in \Irr{\cB^d_n} \times \Irr{\cS^d_n} \mid m_{\lambda, \mu} \neq 0 \big\}
\end{equation}
is still unknown analytically. However, Okada characterizes it through an algorithm \cite[Proposition 2.5]{okada2016pieri}. He finds a relatively simple subset $\Gamma_{n,d} \subset \Omega_{n,d}$ \cite[Theorem 5.4]{okada2016pieri} defined as
\begin{equation}
    \Gamma_{n,d} \defeq \set[\big]{ (\lambda, \mu) \in \Irr{\cB^d_n} \times \Irr{\cS^d_n} \mid \lambda = (1^m), \, r(\mu) = m \text{ for some } m \in \set{0, \dotsc, d} },
\end{equation}
where $r(\mu)$ is the number of rows with odd size in the Young diagram $\mu$ and $(1^0) \defeq \0$ is the empty partition.

The algorithm from \cite[Proposition 2.5, Proposition 2.6, Theorem 5.4]{okada2016pieri} gives an analytical characterization of some subsets of the set $\Omega_{n,d}$ in the following three easy cases:
\begin{enumerate}
    \item if $\lambda = (1^m)$ for some $m \in \set{0, \dotsc, d}$, then $(\lambda, \mu) \in  \Omega_{n,d}$ if and only if $r(\mu) = m$, i.e., $(\lambda, \mu) \in \Gamma_{n,d}$,
    \item if $\lambda \vdash n$, then $(\lambda, \mu) \in  \Omega_{n,d}$ if and only if $\mu = \lambda$,
    \item if $\mu = (n)$, then $(\lambda, \mu) \in  \Omega_{n,d}$ if and only if $\lambda = (n-2r)$ for some $r \in \set{ 0, \ldots, \lfloor \tfrac{n}{2} \rfloor }$.
\end{enumerate}

\begin{remark} \label{rem:okada_set_qubit}
    When the dimension $d$ is $2$, a complete and simple characterization of positive multiplicities $m_{\lambda, \mu}$ can be found from the Okada algorithm \cite[Proposition 2.5]{okada2016pieri} (see also \cite[Proposition 6.6]{Ryan_2022}): $(\lambda, \mu) \in \Omega_{n,2}$ if and only if $\lambda_1 \leq \mu_1 - \mu_2$, with the exceptions of $\lambda = \varnothing$, in which case both rows of $\mu$ must be even, and $\lambda = (1,1)$, in which case both rows of $\mu$ must be odd.
\end{remark}

%%%%%%%%%%%%%%%%%%%%%%%%%%%%%%%%%%%%%%%%%%%%%%%%%%%%%%%%%%%%%%%%%%%%%%%%%%%%%%%%%%%%%%%%%%%%%%%%%%%%%%%%%%%%%%%%%%%%%%%%%%%%%%%%%%%%%%%%%%%%%%%%%%%%%%%%%%%%%%%%%%%%%%%%%%%%%%%%%%%%%%%%%%%%%%%%%%%%%%%%%%%%%%%%%%%%%%%%%%
%%%%%%%%%%%%%%%%%%%%%%%%%%%%%%%%%%%%%%%%%%%%%%%%%%%%%%%%%%%%%%%%%%%%%%%%%%%%%%%%%%%%%%%%%%%%%%%%%%%%%%%%%%%%%%%%%%%%%%%%%%%%%%%%%%%%%%%%%%%%%%%%%%%%%%%%%%%%%%%%%%%%%%%%%%%%%%%%%%%%%%%%%%%%%%%%%%%%%%%%%%%%%%%%%%%%%%%%%%
%%%%%%%%%%%%%%%%%%%%%%%%%%%%%%%%%%%%%%%%%%%%%%%%%%%%%%%%%%%%%%%%%%%%%%%%%%%%%%%%%%%%%%%%%%%%%%%%%%%%%%%%%%%%%%%%%%%%%%%%%%%%%%%%%%%%%%%%%%%%%%%%%%%%%%%%%%%%%%%%%%%%%%%%%%%%%%%%%%%%%%%%%%%%%%%%%%%%%%%%%%%%%%%%%%%%%%%%%%
\section{General formalisation of the problem} \label{sec:generalFormalisationOfTheProblem}

In this section, we formalize our problem. We quantify the monogamy of our highly symmetric entangled states via the following general \emph{semi-definite programs} (SDPs). Let $\Pi \in \Herm(\H \x \H)$ be any \emph{flip-invariant} two-qudit projector, i.e., $\F \Pi \F\ct = \Pi$. Later $\Pi$ will be either $\Pi_{\0}$ or $\Pi_{\ydsm{1,1}}$, and one should think about these projectors as the ones which select the entangled subspace of interest of the total Hilbert space. Now we want to solve the following SDPs for different choices of graph $G = (V,E)$ and projectors $\Pi$:
\begin{align}
    \tilde{p}^{w}_{\Pi}(G,d) &\defeq \max_{\rho,p} \enspace p
    & \textrm{s.t.} &&
    \Tr [\Pi_e \rho] \geq p \quad \forall e \in E, \quad
    \Tr [\rho] &= 1, \quad \rho \succeq 0,
    \label{def:SDPworstIneq} \\
    p^{w}_{\Pi}(G,d) &\defeq \max_{\rho,p} \enspace p
    & \textrm{s.t.} &&
    \Tr [\Pi_e \rho] = p \quad \forall e \in E, \quad
    \Tr [\rho] &= 1, \quad \rho \succeq 0,
    \label{def:SDPworst} \\
    p^{avg}_{\Pi}(G,d) &\defeq \max_{\rho} \enspace \frac{1}{\abs{E}} \sum_{e \in E} \Tr [\Pi_e \rho]
    & \textrm{s.t.} &&
    \Tr [\rho] &= 1, \quad \rho \succeq 0,
\end{align}
where $\Pi_e$ is defined for every edge $e \in E$ as $\Pi_e \defeq \Pi \otimes I_{\bar{e}}$, and $I_{\bar{e}}$ is the identity on all vertices except those of $e$. Note that $p^{avg}_{\Pi}(G,d)$ is just the largest eigenvalue of the Hamiltonian $\frac{1}{\abs{E}} \sum_{e \in E} \Pi_e$:
\begin{equation}
    p^{avg}_{\Pi}(G,d) = \lambda_{\mathrm{max}}\of*{\sum_{e \in E} \frac{1}{\abs{E}} \Pi_e}.
\end{equation}

We can also formalize our intuition from \cref{sec:introduction} a bit differently. Namely, given a one-parameter family of bipartite states $\sigma(p)$ and a graph $G$, we want to maximize $p$ by finding a global state $\rho$ such that $\rho_e = \sigma(p)$ for every edge $e \in E$:
\begin{align}\label{def:sdp_state_primal}
    p_\sigma(G,d) &\defeq \max_{\rho,p} \enspace  p \qquad \textrm{s.t.} \quad
    \rho_e = \sigma(p) \quad \forall e \in E, \quad
    \Tr [\rho] = 1, \quad \rho \succeq 0,
\end{align}
The parameter $p$ should be understood as some measure of entanglement. We focus on the cases, where $\sigma$ states are Werner, Brauer or isotropic with $p$ being the overlap onto the relevant projector $\Pi$ defined as $p = \Tr [ \Pi \sigma(p) ]$.

%%%%%%%%%%%%%%%%%%%%%%%%%%%%%%%%%%%%%%%%%%%%%%%%%%%%%%%%%%%%%%%%%%%%%%%%%%%%%%%%%%%%%%%%%%%%%%%%%%%%%%%%%%%%%%%%%%%%%%%%%%%%%%%%%%%%%%%%%%%%%%%%%%%%%%%%%%%%%%%%%%%%%%%%%%%%%%%%%%%%%%%%%%%%%%%%%%%%%%%%%%%%%%%%%%%%%%%%%%
%%%%%%%%%%%%%%%%%%%%%%%%%%%%%%%%%%%%%%%%%%%%%%%%%%%%%%%%%%%%%%%%%%%%%%%%%%%%%%%%%%%%%%%%%%%%%%%%%%%%%%%%%%%%%%%%%%%%%%%%%%%%%%%%%%%%%%%%%%%%%%%%%%%%%%%%%%%%%%%%%%%%%%%%%%%%%%%%%%%%%%%%%%%%%%%%%%%%%%%%%%%%%%%%%%%%%%%%%%
\subsection{Dual SDP approach}

The \emph{Lagrangian} \cite{boyd2004convex} associated with the optimisation problem \eqref{def:SDPworstIneq} is defined as
\begin{equation}
    L \of[\big]{p, \rho, Z, {(x_e)}_e, y} \defeq p + y \of[\big]{\Tr[\rho] - 1} + \sum_{e \in E} x_e \of[\big]{\Tr [\Pi_e \rho] - p} + \frob[\big]{Z, \rho},
\end{equation}
where $y\in \mathbb{R}$ and $\of{x_e \in \mathbb{R}}_e$ are real \emph{Lagrange multipliers}, $Z \in \Herm\of[\big]{(\C^d)\xp{n}}$ and $\frob{\cdot , \cdot}$ denotes the \emph{Frobenius inner product}
\begin{equation}
    \frob[\big]{A, B} \defeq \Tr \big{[} A^* B \big{]}.
\end{equation}
The \emph{Min-Max principle} states that
\begin{equation}
    \max_{p,\rho \succeq 0} \quad \min_{\substack{y, x_e \geq 0 \\ Z \succeq 0 }} \quad L \of[\big]{ p, \rho, Z,{(x_e)}_e,y} \leq \min_{\substack{y, x_e \geq 0 \\ Z \succeq 0 }} \quad \max_{p,\rho \succeq 0} \quad L \of[\big]{ p, \rho, Z, {(x_e)}_e, y}.
\end{equation}
In fact, \emph{Slater's condition} holds true for our SDP (take $\rho = \frac{1}{d^n} \cdot \I^{\otimes n}$ and $p=0$) and we have the equality. Rewriting the Lagrangian gives
\begin{equation}
    L \of[\big]{ p, \rho, Z,{(x_e)}_e,y} = - y + p \of[\bigg]{ 1 - \sum_{e \in E} x_e} + \frob[\big]{ H + Z + y I, \rho},
\end{equation}
where $H \defeq \sum_{e \in E} x_e\Pi_e$. Making the constraints of the Lagrangian explicit, that is,
\begin{equation}
    \max_{p,\rho \succeq 0} \quad L \of[\big]{ p, \rho, Z, {(x_e)}_e, y} =
    \begin{cases}
        -y &\text{if } \sum_{e \in E} x_e = 1 \text{ and } 0 \succeq H + Z + y I \\
        \infty &\text{otherwise}
    \end{cases}
\end{equation}
after substitution $y \mapsto -y$ the dual SDP of \cref{def:SDPworst} becomes 
\begin{equation}\label{SDPworst_a_bit_simplified}
    \tilde{p}^{w}_{\Pi}(G,d) = \min_{\mathclap{{(x_e)}_e,y,Z}} \quad y
    \qquad \textrm{s.t.} \quad
    \sum_{e \in E} x_e = 1, \quad x_e \geq 0 \quad \forall e \in E, \quad
    0 \succeq  H + Z - yI, \quad Z \succeq 0.
\end{equation}

Recall that for any Hermitian matrix $M$ the smallest $\lambda \in \R$ such that $\lambda I \succeq M$ is equal to the largest eigenvalue $\lambda_{\mathrm{max}}(M)$. Then \cref{SDPworst_a_bit_simplified} can be, after simplifying the variable $Z$, rewritten as
\begin{equation}
    \tilde{p}^{w}_{\Pi}(G,d) = \min_{\substack{\sum_{e \in E} x_e = 1 \\ x_e \geq 0, \, \forall e \in E }} \lambda_{\mathrm{max}}\of[\bigg]{\sum_{e \in E} x_e \Pi_e}.
    \label{def:sdp_dual_worst_ineq}
\end{equation}
A similar calculation for $p^{w}_{\Pi}(G,d)$ from \cref{def:SDPworstIneq} shows that
\begin{equation}
    p^{w}_{\Pi}(G,d) = \min_{\sum_{e \in E} x_e = 1}  \lambda_{\mathrm{max}}\of[\bigg]{\sum_{e \in E} x_e \Pi_e}.
    \label{def:sdp_dual_worst}
\end{equation}
We conjecture that $\tilde{p}^{w}_{\Pi}(G,d) = p^{w}_{\Pi}(G,d)$, and even more generally that:

\begin{conjecture}
For any graph $G$ and a flip-invariant orthogonal projection $\Pi$,
\begin{equation}
    \min_{\sum_{e \in E} x_e = 1}  \lambda_{\mathrm{max}}\of[\bigg]{\sum_{e \in E} x_e \Pi_e} = \min_{\substack{\sum_{e \in E} x_e = 1 \\ x_e \geq 0, \, \forall e \in E }} \lambda_{\mathrm{max}}\of[\bigg]{\sum_{e \in E} x_e \Pi_e}.
\end{equation}
\end{conjecture}

%%%%%%%%%%%%%%%%%%%%%%%%%%%%%%%%%%%%%%%%%%%%%%%%%%%%%%%%%%%%%%%%%%%%%%%%%%%%%%%%%%%%%%%%%%%%%%%%%%%%%%%%%%%%%%%%%%%%%%%%%%%%%%%%%%%%%%%%%%%%%%%%%%%%%%%%%%%%%%%%%%%%%%%%%%%%%%%%%%%%%%%%%%%%%%%%%%%%%%%%%%%%%%%%%%%%%%%%%%
%%%%%%%%%%%%%%%%%%%%%%%%%%%%%%%%%%%%%%%%%%%%%%%%%%%%%%%%%%%%%%%%%%%%%%%%%%%%%%%%%%%%%%%%%%%%%%%%%%%%%%%%%%%%%%%%%%%%%%%%%%%%%%%%%%%%%%%%%%%%%%%%%%%%%%%%%%%%%%%%%%%%%%%%%%%%%%%%%%%%%%%%%%%%%%%%%%%%%%%%%%%%%%%%%%%%%%%%%%
\subsection{Automorphism group action and edge-transitive graphs}

Let $\rho$ be an optimal solution of the primal SDP~\eqref{def:SDPworst} for a given graph $G = (V,E)$ with objective value $p$, and let us twirl it using the symmetries of $G$:
\begin{equation}
    \tilde{\rho} = \frac{1}{\abs{\mathrm{Aut}(G)}} \sum_{\pi \in \mathrm{Aut}(G)} \psi(\pi) \rho \psi(\pi^{-1}),
\end{equation}
where $\mathrm{Aut}(G)$ is the automorphism group of $G$.
Let $o(e)$ denote the orbit of edge $e \in E$ under this action.
Note that $\tilde{\rho}$ still satisfies the inequality constraints of SDP~\eqref{def:SDPworstIneq}:
\begin{equation}
    \Tr [\Pi_e \tilde{\rho} ] = \frac{1}{\abs{o(e)}} \sum_{e' \in o(e)} p_{e'} \geq p,
\end{equation}
where $p_e \defeq \Tr [\Pi_e \rho ]$ and we used that $ p_e \geq p$ for the feasible solution $\rho$. That means that we can always restrict the set of feasible solutions only to those that are invariant under the action of the automorphism group of the graph $G$. This observation simplifies the dual SDP \eqref{def:sdp_dual_worst_ineq}:
\begin{equation}
    \begin{aligned}
        \tilde{p}^{w}_{\Pi}(G,d) = \min_{\substack{\sum_{o \in O(G)} x_o = 1 \\ x_o \geq 0, \, \forall o \in O(G)}}  \lambda_{\mathrm{max}}\of[\bigg]{\sum_{o \in O(G)} x_o \sum_{e \in o} \frac{1}{\abs{o}} \Pi_e},
    \end{aligned}
    \label{def:sdp_dual_worst_aut}
\end{equation}
where $o$ denotes an orbit in the set of all orbits $O(G)$ of the induced automorphism group $\mathrm{Aut}(G)$ action on the set of edges $E$. An edge-transitive graph $G$ has only one orbit, so the condition $x_o \geq 0$ naturally holds for \cref{def:SDPworst} and \cref{def:sdp_dual_worst_aut} simplifies to
\begin{equation}
    \tilde{p}^{w}_{\Pi}(G,d) = \lambda_{\mathrm{max}}\of*{\sum_{e \in E} \frac{1}{\abs{E}} \Pi_e} = p^{avg}_{\Pi}(G,d) = p^{w}_{\Pi}(G,d).
    \label{def:sdp_dual_worst_edgetransitive}
\end{equation}

In the following sections, we will focus on understanding the values $p^{w}_{\Pi}(G,d)$ for the complete graph $G=K_n$ for the following choices of $\Pi$:
\begin{align}
    q_W(n,d) &\defeq p^{w}_{\Pi}(K_n,d) \text{ for } \Pi \defeq \Pi_{\ydsm{1,1}}, \label{def:p_values_Kn_Werner}\\
    p_B(n,d) &\defeq p^{w}_{\Pi}(K_n,d) \text{ for } \Pi \defeq \Pi_\0. \label{def:p_values_Kn_Brauer}
\end{align}
We will also formulate a similar problem for isotropic states (which form a subset of Brauer states). In that case, we will refer to the optimized value by $p_{I}(n,d)$, see \cref{sec:isotropic}.

%%%%%%%%%%%%%%%%%%%%%%%%%%%%%%%%%%%%%%%%%%%%%%%%%%%%%%%%%%%%%%%%%%%%%%%%%%%%%%%%%%%%%%%%%%%%%%%%%%%%%%%%%%%%%%%%%%%%%%%%%%%%%%%%%%%%%%%%%%%%%%%%%%%%%%%%%%%%%%%%%%%%%%%%%%%%%%%%%%%%%%%%%%%%%%%%%%%%%%%%%%%%%%%%%%%%%%%%%%
%%%%%%%%%%%%%%%%%%%%%%%%%%%%%%%%%%%%%%%%%%%%%%%%%%%%%%%%%%%%%%%%%%%%%%%%%%%%%%%%%%%%%%%%%%%%%%%%%%%%%%%%%%%%%%%%%%%%%%%%%%%%%%%%%%%%%%%%%%%%%%%%%%%%%%%%%%%%%%%%%%%%%%%%%%%%%%%%%%%%%%%%%%%%%%%%%%%%%%%%%%%%%%%%%%%%%%%%%%
%%%%%%%%%%%%%%%%%%%%%%%%%%%%%%%%%%%%%%%%%%%%%%%%%%%%%%%%%%%%%%%%%%%%%%%%%%%%%%%%%%%%%%%%%%%%%%%%%%%%%%%%%%%%%%%%%%%%%%%%%%%%%%%%%%%%%%%%%%%%%%%%%%%%%%%%%%%%%%%%%%%%%%%%%%%%%%%%%%%%%%%%%%%%%%%%%%%%%%%%%%%%%%%%%%%%%%%%%%
\section{\texorpdfstring{$K_n$}{Kn}-Extendibility} \label{sec:KnExtendibility}

\subsection{Werner states}

Let us first consider in detail the case where the reduced two-body state is a Werner state, i.e., $\Pi = \Pi_{\ydsm{1,1}} = \frac{\I-\F}{2}$. One can, in principle, directly solve \cref{def:SDPworst} using techniques from \cite{Christandl2007}, see \cref{app:werner_primal}. In this section, we will show how Jucys--Murphy elements help to solve the dual SDP in a simpler manner. It is easy to see that
\begin{equation}
    q_W(n,d) = \lambda_{\mathrm{max}}(H_{K_n}) = \frac{1}{2} \of*{1- \frac{\lambda_{\mathrm{min}}(J_\S)}{\abs{E}}},
\end{equation}
where $J_\S = \sum_{i=1}^{n} J_i$ and $J_i$ are Jucys--Murphy elements for symmetric group $\S_n$. The spectrum for Jucys--Murphy elements is well known, see \cref{sec:JM_elemnts}, so we can write:
\begin{equation}
    \lambda_{\mathrm{min}}(J_\S) = \min_{\substack{\lambda \pt n \\ l(\lambda) \leq d}} \cont(\lambda),
\end{equation}
with $\cont(\lambda)$ the content of the Young diagram $\lambda \pt n$. The optimal Young diagram $\lambda^*$ that achieves the minimum is the tallest one with the constraint that the number of rows is less than or equal to $d$, i.e. the most rectangular shape possible:
\begin{equation}\label{eq:lambda_star}
    \lambda^*_1 = \cdots = \lambda^*_k = \frac{n-k}{d}+1, \quad \lambda^*_{k+1} = \cdots = \lambda^*_d = \frac{n-k}{d},
\end{equation}
where $k \defeq n \bmod d$. Therefore we can write 
\begin{align}\label{eq:werner_J_cont_rectangle}
    \lambda_{\mathrm{min}}(J) = & \sum_{i=0}^{k-1} \of*{ -\of*{\frac{n-k}{d}+1}i + \frac{n-k}{2d}\of*{\frac{n-k}{d}+1} } +\sum_{i=k}^{d-1} \of*{ -\frac{n-k}{d}i + \frac{n-k}{2d}\of*{\frac{n-k}{d}-1} } \nonumber \\
    = & \of*{\frac{n-k}{d}+1} \of*{ \frac{k(n-k)}{2d} - \frac{k(k-1)}{2} } + \frac{n-k}{d} \of*{ \frac{d-k}{2}\of*{\frac{n-k}{d}-1} - \frac{(d-k)(d+k-1)}{2}} \nonumber \\
    = & \frac{k(d-k)(d+1)}{2d} + \frac{n(n-d^2)}{2d}.
\end{align}
The above argument gives a proof for the following theorem.
\begin{theorem} \label{thm:WernerStates}
    The optimization problem \cref{def:p_values_Kn_Werner} has the optimal value
    \begin{equation}
       q_W(n,d) = \frac{d-1}{2d} \frac{(n+k+d)(n-k)}{n(n-1)} + \frac{k(k-1)}{n(n-1)},
    \end{equation}
    where $k = n \bmod d$.
\end{theorem}

\begin{figure}[!ht]
\begin{center}
    \begin{NiceTabular}{c|cccccccc}[columns-width = 2em, cell-space-limits = 0.25em]
        \CodeBefore
            \rectanglecolor{gray!30!white}{2-3}{2-9}
            \rectanglecolor{gray!30!white}{3-8}{3-9}
        \Body
        \diagbox{$d$}{$n$} & 2 & 3 & 4 & 5 & 6 & 7 & 8 & 9 \\ \hline
        2 &	1 &	\nicefrac{1}{2} & \nicefrac{1}{2} & \nicefrac{2}{5} & \nicefrac{2}{5} & \nicefrac{5}{14}  & \nicefrac{5}{14}  & \nicefrac{1}{3} \\
        3 &	1 &	1               & \nicefrac{2}{3} & \nicefrac{3}{5} & \nicefrac{3}{5} & \nicefrac{11}{21} & \nicefrac{1}{2}   & \nicefrac{1}{2} \\
        4 &	1 &	1               & 1               & \nicefrac{3}{4} & \nicefrac{2}{3} & \nicefrac{9}{14}  & \nicefrac{9}{14}  & \nicefrac{7}{12} \\
        5 &	1 &	1               & 1               & 1               & \nicefrac{4}{5} & \nicefrac{5}{7}   & \nicefrac{19}{28} & \nicefrac{2}{3} \\
        6 &	1 &	1               & 1               & 1               & 1               & \nicefrac{5}{6}   & \nicefrac{3}{4}   & \nicefrac{17}{24} \\
        7 & 1 & 1               & 1               & 1               & 1               & 1                 & \nicefrac{6}{7}   & \nicefrac{7}{9} \\
        8 & 1 & 1               & 1               & 1               & 1               & 1                 & 1                 & \nicefrac{7}{8} \\
        9 & 1 & 1               & 1               & 1               & 1               & 1                 & 1                 & 1 \\
    \end{NiceTabular}
    \caption{The first values of $q_W(n,d)$. The values of $q_W(n,d)$, for which the Werner states $\rho_e$ are separable (i.e. $p \leq \nicefrac{1}{2}$), are in grey.}
    \label{fig:firstValuesWernerStates}
\end{center}
\end{figure}

%%%%%%%%%%%%%%%%%%%%%%%%%%%%%%%%%%%%%%%%%%%%%%%%%%%%%%%%%%%%%%%%%%%%%%%%%%%%%%%%%%%%%%%%%%%%%%%%%%%%%%%%%%%%%%%%%%%%%%%%%%%%%%%%%%%%%%%%%%%%%%%%%%%%%%%%%%%%%%%%%%%%%%%%%%%%%%%%%%%%%%%%%%%%%%%%%%%%%%%%%%%%%%%%%%%%%%%%%%
%%%%%%%%%%%%%%%%%%%%%%%%%%%%%%%%%%%%%%%%%%%%%%%%%%%%%%%%%%%%%%%%%%%%%%%%%%%%%%%%%%%%%%%%%%%%%%%%%%%%%%%%%%%%%%%%%%%%%%%%%%%%%%%%%%%%%%%%%%%%%%%%%%%%%%%%%%%%%%%%%%%%%%%%%%%%%%%%%%%%%%%%%%%%%%%%%%%%%%%%%%%%%%%%%%%%%%%%%%
\subsection{Isotropic states}\label{sec:isotropic}

Another important class of bipartite states with symmetries are isotropic states of the form \cref{eq:isotropicStates_proj}, which form a subset of Brauer states. A convenient way to write an isotropic state is using as a linear combination of $\frac{\W}{d}$ and $\frac{\I}{d^2}$:
\begin{equation} \label{eq:isotropicStates}
    p' \cdot \frac{\W}{d} + (1-p') \cdot \frac{\I}{d^2}.
\end{equation}
When $\frac{-1}{d^2 - 1} \leq p' \leq 1$, the state $\rho$ is positive semi-definite and unit trace, and hence a quantum state. The parameters $p$ of \cref{eq:isotropicStates_proj} and $p'$ of \cref{eq:isotropicStates} are related via
\begin{equation} \label{eq:isotropicStates_p_and_p'_relation}
    p' = p + \frac{p-1}{d^2-1} = p \frac{d^2}{d^2-1} - \frac{1}{d^2-1}, \qquad p =  \frac{1}{d^2} + \of*{1- \frac{1}{d^2}} p'
\end{equation}
and an isotropic state $\rho$ becomes separable if and only if $p' \leq \frac{1}{d+1}$.

If we want to formulate a similar optimization problem as in \cref{def:SDPworst} for isotropic states then one should add additional constraints $[\rho_e, U \otimes \bar{U}] = 0$ for every edge $e \in E(K_n)$ and arbitrary unitary $U \in \U_d$, to the optimization problem \eqref{def:SDPworst}. Equivalently, we can formulate the problem \cref{def:sdp_state_primal} for isotropic states as the following SDP:

\begin{equation}\label{eq:SDP}
    \begin{aligned}
        p'_{I}(n,d) \defeq \max_{\rho,p'} \quad & p' \qquad
        \textrm{s.t.} \quad
         \rho_e = \frac{p'}{d} \cdot \W + \frac{1-p'}{d^2} \cdot \I \quad \forall e \in E, \quad
         \Tr [\rho] = 1, \quad \rho \succeq 0.
    \end{aligned}
\end{equation}
This problem depends on the number of quantum systems $n$ and their dimension $d$. The aim is to find a state on a complete graph such that the reduced states between each pair of vertices are as maximally entangled as possible. Note that in the previous SDP \eqref{eq:SDP}, the condition $\Tr [\rho] = 1$ is superfluous, and the optimisation problem reduces to:
\begin{equation}\label{eq:SDPStandard}
    \begin{aligned}
        p'_{I}(n,d) = \max_{\rho,p'} \quad & p' \qquad
        \textrm{s.t.} \quad
         \rho_e = \frac{p'}{d} \cdot \W + \frac{1-p'}{d^2} \cdot \I \quad \forall e \in E, \quad
         \rho \succeq 0.
    \end{aligned}
\end{equation}
It turns out that the dual problem can written as follows (see \cref{app:dualSDPIsotropicStates} for the derivation):
\begin{equation}\label{eq:dualSDPSimplifiedOrthogonalSymmetricSimple_main}
    p'_I(n,d) = \min_{x \in \R} \lambda_{\text{max}} \of*{ \sum_{e \in E} \of*{ \of*{\frac{1}{\abs{E}(1-d)} - x} \of{\I_e - d \cdot \F_e} + x \of{\F_e - \W_e} } }.
\end{equation}

We can find the value \cref{eq:dualSDPSimplifiedOrthogonalSymmetricSimple_main} analytically, which is one of the main results of this paper:

\begin{theorem} \label{thm:isotropicStates}
    The optimization problem \cref{eq:dualSDPSimplifiedOrthogonalSymmetricSimple_main} has the optimal value
    \begin{equation}
        p'_I(n,d) =
        \begin{cases}
            \frac{1}{n + n \bmod 2 - 1} &\text{ if $d > n$ or either $d$ or $n$ is even} \\
            \min \big\{ \frac{2 d + 1}{2 d n + 1}, \frac{1}{n - 1} \big\} &\text{ if $n \geq d$ and both $d$ and $n$ are odd}.
        \end{cases}
    \end{equation}
    The corresponding value $p_I(n,d)$ can be obtained from the relation \cref{eq:isotropicStates_p_and_p'_relation}:
    \begin{equation}
        p_I(n,d) =
        \begin{cases}
            \frac{1}{d^2} + \of*{1- \frac{1}{d^2}} \frac{1}{n + n \bmod 2 - 1} &\text{ if $d > n$ or either $d$ or $n$ is even} \\
            \frac{1}{d^2} + \of*{1- \frac{1}{d^2}} \min \big\{ \frac{2 d + 1}{2 d n + 1}, \frac{1}{n - 1} \big\} &\text{ if $n \geq d$ and both $d$ and $n$ are odd}.
        \end{cases}
    \end{equation}
\end{theorem}
The proof of this theorem is provided in the next sections. Our strategy is to prove the theorem in the two cases separately:
\begin{itemize}
    \item when $d > n$ or $d$ is even or $n$ is even,
    \item when $d \leq n$ and $d$ is odd and $n$ is odd,
\end{itemize}
by combining \cref{thm:largeDimension} and \cref{thm:smallDimension}. We begin by establishing a simple lower bound using the notion of perfect matching of a graph.

\begin{figure}[htbp]
    \centering
    \begin{subfigure}{\textwidth}
        \centering
        \begin{NiceTabular}{c|cccccccc}[columns-width = auto, cell-space-limits = 0.25em]
            \CodeBefore
                \rectanglecolor{gray!30!white}{2-3}{2-9}
                \rectanglecolor{gray!30!white}{3-5}{3-9}
                \rectanglecolor{gray!30!white}{4-5}{4-9}
                \rectanglecolor{gray!30!white}{5-7}{5-9}
                \rectanglecolor{gray!30!white}{6-7}{6-9}
                \rectanglecolor{gray!30!white}{7-9}{7-9}
                \rectanglecolor{gray!30!white}{7-9}{8-9}
            \Body
            \diagbox{$d$}{$n$} & 2 & 3 & 4 & 5 & 6 & 7 & 8 & 9 \\ \hline
            2 &	1 &	\nicefrac{1}{3}  & \nicefrac{1}{3} & \nicefrac{1}{5}   & \nicefrac{1}{5} & \nicefrac{1}{7}   & \nicefrac{1}{7} & \nicefrac{1}{9} \\
            3 &	1 &	\nicefrac{7}{19} & \nicefrac{1}{3} & \nicefrac{7}{31}  & \nicefrac{1}{5} & \nicefrac{7}{43}  & \nicefrac{1}{7}& \nicefrac{1}{8} \\
            4 &	1 &	\nicefrac{1}{3}  & \nicefrac{1}{3} & \nicefrac{1}{5}   & \nicefrac{1}{5} & \nicefrac{1}{7}   & \nicefrac{1}{7}& \nicefrac{1}{9} \\
            5 &	1 &	\nicefrac{1}{3}  & \nicefrac{1}{3} & \nicefrac{11}{51} & \nicefrac{1}{5} & \nicefrac{11}{71} & \nicefrac{1}{7}& \nicefrac{11}{91} \\
            6 &	1 &	\nicefrac{1}{3}  & \nicefrac{1}{3} & \nicefrac{1}{5}   & \nicefrac{1}{5} & \nicefrac{1}{7}   & \nicefrac{1}{7}& \nicefrac{1}{9} \\
            7 & 1 & \nicefrac{1}{3}  & \nicefrac{1}{3} & \nicefrac{1}{5}   & \nicefrac{1}{5} & \nicefrac{5}{33}  & \nicefrac{1}{7}& \nicefrac{15}{127} \\
            8 & 1 & \nicefrac{1}{3}  & \nicefrac{1}{3} & \nicefrac{1}{5}   & \nicefrac{1}{5} & \nicefrac{1}{7}   & \nicefrac{1}{7}& \nicefrac{1}{9} \\
            9 & 1 & \nicefrac{1}{3}  & \nicefrac{1}{3} & \nicefrac{1}{5}   & \nicefrac{1}{5} & \nicefrac{1}{7}   & \nicefrac{1}{7}& \nicefrac{19}{163} \\
        \end{NiceTabular}
        \caption{The first values of $p'_{I}(n,d)$.}
        \label{fig:pPrimeIsotropicStates}
    \end{subfigure} \\[2em]
    \begin{subfigure}{\textwidth}
        \centering
        \begin{NiceTabular}{c|cccccccc}[columns-width = auto, cell-space-limits = 0.25em]
            \CodeBefore
                \rectanglecolor{gray!30!white}{2-3}{2-9}
                \rectanglecolor{gray!30!white}{3-5}{3-9}
                \rectanglecolor{gray!30!white}{4-5}{4-9}
                \rectanglecolor{gray!30!white}{5-7}{5-9}
                \rectanglecolor{gray!30!white}{6-7}{6-9}
                \rectanglecolor{gray!30!white}{7-9}{7-9}
                \rectanglecolor{gray!30!white}{7-9}{8-9}
            \Body
            \diagbox{$d$}{$n$} & 2 & 3 & 4 & 5 & 6 & 7 & 8 & 9 \\ \hline
            2 &	1 &	\nicefrac{1}{2}    & \nicefrac{1}{2}    & \nicefrac{2}{5}    & \nicefrac{2}{5}    & \nicefrac{5}{14}   & \nicefrac{5}{14}   & \nicefrac{1}{3} \\
            3 &	1 &	\nicefrac{25}{57}  & \nicefrac{11}{27}  & \nicefrac{29}{93}  & \nicefrac{13}{45}  & \nicefrac{11}{43}  & \nicefrac{5}{21}   & \nicefrac{2}{9} \\
            4 &	1 &	\nicefrac{3}{8}    & \nicefrac{3}{8}    & \nicefrac{1}{4}    & \nicefrac{1}{4}    & \nicefrac{11}{56}  & \nicefrac{11}{56}  & \nicefrac{1}{6} \\
            5 &	1 &	\nicefrac{9}{25}   & \nicefrac{9}{25}   & \nicefrac{21}{85}  & \nicefrac{29}{125} & \nicefrac{67}{355} & \nicefrac{31}{175} & \nicefrac{71}{455} \\
            6 &	1 &	\nicefrac{19}{54}  & \nicefrac{19}{54}  & \nicefrac{2}{9}    & \nicefrac{2}{9}    & \nicefrac{1}{6}    & \nicefrac{1}{6}    & \nicefrac{11}{81} \\
            7 & 1 & \nicefrac{17}{49}  & \nicefrac{17}{49}  & \nicefrac{53}{245} & \nicefrac{53}{245} & \nicefrac{13}{77}  & \nicefrac{55}{343} & \nicefrac{121}{889} \\
            8 & 1 & \nicefrac{11}{32}  & \nicefrac{11}{32}  & \nicefrac{17}{80}  & \nicefrac{17}{80}  & \nicefrac{5}{32}   & \nicefrac{5}{32}   & \nicefrac{1}{8} \\
            9 & 1 & \nicefrac{83}{243} & \nicefrac{83}{243} & \nicefrac{17}{81}  & \nicefrac{17}{81}  & \nicefrac{29}{189} & \nicefrac{29}{189} & \nicefrac{187}{1467} \\
        \end{NiceTabular}
        \caption{The first values of $p_{I}(n,d)$.}
        \label{fig:pIsotropicStates}
    \end{subfigure}
    \caption{The first values of $p'_{I}(n,d)$ and $p_{I}(n,d)$. The values for which the isotropic states $\rho_e$ are separable (i.e. $p \leq \nicefrac{1}{d}$), are in grey. Note that they decrease with respect to $n$, but are not monotonic with respect to $d$.}
\end{figure}

%%%%%%%%%%%%%%%%%%%%%%%%%%%%%%%%%%%%%%%%%%%%%%%%%%%%%%%%%%%%%%%%%%%%%%%%%%%%%%%%%%%%%%%%%%%%%%%%%%%%%%%%%%%%%%%%%%%%%%%%%%%%%%%%%%%%%%%%%%%%%%%%%%%%%%%%%%%%%%%%%%%%%%%%%%%%%%%%%%%%%%%%%%%%%%%%%%%%%%%%%%%%%%%%%%%%%%%%%%
\subsubsection{Lower bound} \label{sec:lowerBound}

We can get a simple lower bound on $p'_{I}(n,d)$ by looking at a set of perfect matchings on the complete graph. A perfect matching on a graph is a set of edges, such that every vertex is contained in exactly one of those edges.
\begin{proposition} \label{prop:perfectMatchingCount}
    There are $(2n - 1)!!$ perfect matchings on $K_{2n}$, and if $e$ is an edge on $K_{2n}$, then there are $(2n-3)!!$ perfect matchings on $K_{2n}$ containing $e$.
\end{proposition}
\begin{proof}
    Let $a_n$ be the number of perfect matching on $K_{2n}$; clearly $a_1 = 1$. Now assume $n > 1$ and let $v$ be a vertex in $V$. This vertex can be matched with $2n - 1$ other vertices, let $u \in V$ be such other vertex matched with $v$. Remove $u$ and $v$ from $K_{2n}$, the resulting graph $K_{2n} \setminus \{u, v\}$ is the complete graph $K_{2(n-1)}$. Thus, by induction on $n$, the number of perfect matchings on $K_{2n}$ satisfies the recursive relation:
    \begin{equation}
        a_n = (2n - 1) a_{n-1} \quad \implies \quad a_n = (2n - 1)!!.
    \end{equation}
    Assume $e = (u, v)$ in $E$, thus the number of perfect matchings containing $e$ is the number of perfect matchings on $K_{2n} \setminus \{u, v\}$, that is $(2n - 3)!!$.
\end{proof}
\begin{remark*}
    There is no perfect matching on $K_n$ for odd $n$ (e.g., see \cref{fig:perfectMatching}).
\end{remark*}

\begin{figure}[!ht]
    \captionsetup[subfigure]{labelformat=empty}
    \centering
    \begin{subfigure}[b]{0.45\textwidth}
        \centering
        \begin{tikzpicture}[site/.style = {circle,
                                           draw = white,
                                           line width = 2pt,
                                           fill = black!80!white,
                                           inner sep = 3pt}]
        
            \coordinate (1) at (90:4em) {};
            \coordinate (2) at (150:4em) {};
            \coordinate (3) at (210:4em) {};
            \coordinate (4) at (270:4em) {};
            \coordinate (5) at (330:4em) {};
            \coordinate (6) at (30:4em) {};

            \foreach \i in {1,...,6}
            \foreach \j in {\i,...,6} {
                \draw[-, line width = 1.5pt, draw = white] (\j.center) -- (\i.center);
                \draw[-, line width = 1.5pt, draw = black] (\j.center) -- (\i.center);
            }

            \draw[-, line width = 1.5pt, draw = red!70!white] (2.center) -- (1.center);
            \draw[-, line width = 1.5pt, draw = red!70!white] (4.center) -- (3.center);
            \draw[-, line width = 1.5pt, draw = red!70!white] (6.center) -- (5.center);

            \foreach \i in {1,...,6}
                \node[site] (v\i) at (\i.center) {};
        \end{tikzpicture}
        \caption{$K_6$}
     \end{subfigure}
     \hfill
     \begin{subfigure}[b]{0.45\textwidth}
        \centering
        \begin{tikzpicture}[site/.style = {circle,
                                           draw = white,
                                           line width = 2pt,
                                           fill = black!80!white,
                                           inner sep = 3pt}]
        
            \coordinate (1) at (90:4em) {};
            \coordinate (2) at (162:4em) {};
            \coordinate (3) at (234:4em) {};
            \coordinate (4) at (306:4em) {};
            \coordinate (5) at (18:4em) {};

            \foreach \i in {1,...,5}
            \foreach \j in {\i,...,5} {
                \draw[-, line width = 1.5pt, draw = white] (\j.center) -- (\i.center);
                \draw[-, line width = 1.5pt, draw = black] (\j.center) -- (\i.center);
            }

            \draw[-, line width = 1.5pt, draw = red!70!white] (2.center) -- (1.center);
            \draw[-, line width = 1.5pt, draw = red!70!white] (4.center) -- (3.center);

            \foreach \i in {1,...,4}
                \node[site] (v\i) at (\i.center) {};

            \node[circle, draw = white, line width = 2pt, fill = gray, inner sep = 3pt] (v5) at (5.center) {};
        \end{tikzpicture}
         \caption{$K_5$}
    \end{subfigure}
    \caption
        [Perfect matching in red for complete graphs $K_n$, with even $n$. For $K_n$ with odd $n$, some vertices are not matched in grey.]
        {Perfect matching \tikz{\draw[fill = red!70!white] (0,0) rectangle (1.5ex,1.5ex);} for complete graphs $K_n$, with even $n$. For $K_n$ with odd $n$, some vertices are not matched \tikz{\draw[fill = gray] (0,0) rectangle (1.5ex,1.5ex);}.} \label{fig:perfectMatching}
\end{figure}
A lower bound on the optimisation problem in \cref{eq:SDP} can be written as follows. For even $n$, let $E_1, \ldots, E_{(n - 1)!!}$ be all the perfect matchings on $K_n$, and for each perfect matching $E_k$, define the quantum state $\rho^{(k)}$ on $K_n$ by,
\begin{equation}
    \rho^{(k)} \defeq \prod_{e \in E_k} \frac{\W_e}{d}.
\end{equation}
For odd $n$, let $v$ be a vertex in $V$, and let $E_{v, 1}, \ldots, E_{v, {(n - 2)!!}}$ be all the perfect matchings on $K_n \setminus \{ v \}$. Define the quantum state $\rho^{(v, k)}$ on $K_n$ by,
\begin{equation}
    \rho^{(v, k)} \defeq \frac{\I_v}{d^2} \cdot \prod_{e \in E_{v, k}} \frac{\W_e}{d}.
\end{equation}
That is a quantum state maximally entangled on the perfect matching edges, and a maximally mixed state $\tfrac{\I_v}{d^2}$ on the remaining vertex $v$ in the odd case. For example, on $K_6$ the quantum state constructed from the perfect matching $\big\{ (1,2), (3,4), (5,6) \big\}$ is equal to
\begin{equation}
    \frac{\W}{d} \otimes \frac{\W}{d} \otimes \frac{\W}{d},
\end{equation}
and on $K_7$ the quantum state constructed from the perfect matching $\big\{ (2,3), (4,5), (6,7) \big\}$ of $K_7 \setminus \{ 1 \}$ is equal to
\begin{equation}
    \frac{\I}{d} \otimes \frac{\W}{d} \otimes \frac{\W}{d} \otimes \frac{\W}{d}.
\end{equation}

Let $\rho$ be the quantum state defined on $K_n$, as a uniform combination of the previously constructed states $\rho^{(k)}$ and $\rho^{(v,k)}$:
\begin{align*}
    (n~\textrm{even}) \qquad & \rho \defeq \frac{1}{(n - 1)!!} \sum_{1 \leq k \leq (n - 1)!!} \rho^{(k)} \\[1em]
    (n~\textrm{odd}) \qquad & \rho \defeq \frac{1}{n(n - 2)!!} \sum_{\substack{v \in K_n \\ 1 \leq k \leq (n - 2)!!}} \rho^{(v,k)},
\end{align*}
Then for all edges $e$ in $K_n$, the reduced quantum state $\rho_e$ is
\begin{equation}
    \rho_e =
    \begin{cases}
        \frac{1}{n-1} \frac{1}{d} \cdot \W + \frac{n-2}{n-1} \frac{1}{d^2}  \cdot  \I &n~\text{even} \\
        \frac{1}{n} \frac{1}{d} \cdot \W + \frac{n-1}{n} \frac{1}{d^2}  \cdot \I &n~\text{odd},
    \end{cases}
\end{equation}
where the corresponding normalization factors can be found using \cref{prop:perfectMatchingCount}. The lower bound becomes
\begin{equation}
    p'_{I}(n,d) \geq
    \begin{cases}
        \frac{1}{n-1} &n~\text{even} \\
        \frac{1}{n} &n~\text{odd}.
    \end{cases}
\end{equation}
In particular, the lower bound is independent of the dimension $d$.

%%%%%%%%%%%%%%%%%%%%%%%%%%%%%%%%%%%%%%%%%%%%%%%%%%%%%%%%%%%%%%%%%%%%%%%%%%%%%%%%%%%%%%%%%%%%%%%%%%%%%%%%%%%%%%%%%%%%%%%%%%%%%%%%%%%%%%%%%%%%%%%%%%%%%%%%%%%%%%%%%%%%%%%%%%%%%%%%%%%%%%%%%%%%%%%%%%%%%%%%%%%%%%%%%%%%%%%%%%
\subsubsection{Proof of Theorem \ref*{thm:isotropicStates}}

Using \cref{lem:eigenvalueCentralElementsSn,lem:eigenvalueCentralElementsBn}, the decomposition of the vector space ${\big( \mathbb{C}^d \big)}^{\otimes n}$ under the action of the Brauer algebra $\cB^d_n$, and under the decomposition of the restriction of the irreducible representations of the Brauer algebra $\cB^d_n$ to the symmetric group $\S_n$, we can write
\begin{equation}
    f(x) \defeq \lambda_{\text{max}} \of*{ \sum_{e \in E} \of*{ \of*{\frac{1}{\abs{E}(1-d)} - x} \of{\I_e - d \,\F_e} + x \of{\F_e - \W_e} } } = \max_{(\lambda,\mu) \in  \Omega_{n,d}} f_{\mu,\lambda}(x),
\end{equation}
for affine functions $f_{\mu,\lambda}(x)$ defined as
\begin{equation} \label{eq:affineFunction}
    f_{\mu,\lambda}(x) \defeq \frac{1}{d-1}\of*{\frac{d \, \cont(\mu)}{\abs{E}} - 1} + x \of[\big]{\cont(\lambda) + d \, \cont(\mu) - r(d-1) - \abs{E}},
\end{equation}
with $\abs{\lambda} = n - 2r$. The dual optimization problem is then the minimum value of the max over a set of affine functions, i.e.,
\begin{equation} \label{eq:finalDualSDP}
    p'_{I}(n,d) = \min_{x \in \R} \max_{(\lambda,\mu) \in  \Omega_{n,d}} f_{\mu,\lambda}(x).
\end{equation}

\noindent\textbf{Case when $d > n$ or $d$ is even or $n$ is even.}

A feasible solution for our dual problem \cref{eq:finalDualSDP} (i.e., an upper bound for the SDP \cref{eq:SDP}) can be made by setting $x = \frac{1}{\abs{E}(1-d)}$. Note that in this case $x < 0$, since $d \geq 2$. Then \cref{eq:finalDualSDP} becomes

\begin{equation}
    p'_{I}(n,d) \leq \min_{\lambda \in \Irr{\cB^d_n}} \frac{2}{n (n - 1)(1-d)} \big( \cont(\lambda) - r(d - 1) \big).
\end{equation}

\begin{lemma}\label{lem:upperBound}
    If $d > n$, or either $d$ or $n$ is even, then,
    \begin{equation}
        p'_{I}(n,d) \leq
        \begin{cases}
            \frac{1}{n} & \text{ if $n$ is odd} \\
            \frac{1}{n-1} & \text{ if $n$ is even}.
        \end{cases}
    \end{equation}
\end{lemma}
\begin{proof}
    It is enough to prove that
    \begin{equation}
        \min_{\lambda \in \Irr{\cB^d_n}} \cont(\lambda) - r(d - 1) \leq
        \begin{cases}
            \frac{(1 - d) (n - 1)}{2} & \text{ if $n$ is odd} \\[0.5em]
            \frac{(1 - d) n}{2} & \text{ if $n$ is even},
        \end{cases}
    \end{equation}
    
    If $d > n$, the minimization can be restricted to only single-column partitions $\lambda \defeq \big( 1^{(n - 2r)} \big)$, for all $r \in \{ 0, \dots, \lfloor \frac{n}{2} \rfloor \}$, which is always possible when $d > n$. Let $|\lambda| \defeq n - 2r$, then
    \begin{align}
        \min_{\lambda \in \Irr{\cB^d_n}} \cont(\lambda) - r(d - 1) &\leq \min_{r \in \{0, \dots, \lfloor \frac{n}{2} \rfloor \}} \cont \big( 1^{(n - 2r)} \big) - r(d - 1) 
        = \min_{r \in \{0, \dots, \lfloor \frac{n}{2} \rfloor \}} - \frac{|\lambda| \big{(}|\lambda| - 1 \big{)}}{2} - (d-1) \frac{n - |\lambda|}{2} \nonumber \\
        &=
        \begin{cases}
            \frac{(1 - d) (n - 1)}{2} &\text{if $n$ is odd} \\[0.5em]
            \frac{(1 - d) n}{2} &\text{if $n$ is even}.
        \end{cases}
    \end{align}

    Otherwise, if $d$ is even, let $r^* = \lceil \frac{n}{2} \rceil - \frac{d}{2}$. Then the single-column partition $\lambda \defeq \big( 1^{(n - 2(r + r^*))} \big)$ satisfies $\lambda^\prime_1 + \lambda^\prime_2 \leq d$ for all $r \in \{ 0, \dots, \lfloor \frac{n}{2} \rfloor - r^* \}$, and,
    \begin{equation}
        \min_{\lambda \in \Irr{\cB^d_n}} \cont(\lambda) - r(d - 1) \leq \min_{r \in \{ 0, \dots, \lfloor \frac{n}{2} \rfloor - r^* \}} \cont \big( 1^{(n - 2r)} \big) - (r+r^*)(d - 1) =
        \begin{cases}
            \frac{(1 - d) (n - 1)}{2} &\text{if $n$ is odd} \\[0.5em]
            \frac{(1 - d) n}{2} &\text{if $n$ is even}.
        \end{cases}
    \end{equation}
    
    The same result holds if $n$ is even.
\end{proof}

\begin{theorem} \label{thm:largeDimension}
    If $d > n$, or either $d$ or $n$ is even, then,
    \begin{equation}
        p'_{I}(n,d) =
        \begin{cases}
            \frac{1}{n} & \text{ if $n$ is odd} \\
            \frac{1}{n-1} & \text{ if $n$ is even}.
        \end{cases}
    \end{equation}
\end{theorem}
\begin{proof}
    Using construction from \Cref{sec:lowerBound} and \cref{lem:upperBound}, the dual optimization problem is lower and upper bounded by
    $\frac{1}{n}$ if $n$ is odd and $\frac{1}{n-1}$ if $n$ is even.
\end{proof}

\noindent\textbf{Case when $d \leq n$ and $d$ is odd and $n$ is odd.}

Let us evaluate the affine functions $f_{\mu, \lambda}$ of \cref{eq:affineFunction} at the negative coordinate $\tilde{x} \defeq \frac{1}{|E| (1-d)}$:
\begin{align*}
     f_{\mu,\lambda}(\tilde{x}) &= \frac{1}{d-1}\of*{\frac{d \, \cont(\mu)}{\abs{E}} - 1} + \tilde{x} \of*{\cont(\lambda) + d \, \cont(\mu) - r(d-1) - \abs{E}} = \frac{1}{n-1} + \tilde{x} \cdot h(\lambda),
\end{align*}
where $h(\lambda)$ is defined by $h(\lambda) \defeq \frac{1}{2} \sum^{\lambda_1}_{i=0} \lambda^\prime_i (d - \lambda^\prime_i + 2 (i - 1))$. At this coordinate the affine functions do not depend on $\mu$ anymore, so we define
\begin{equation}
    g(\lambda) \defeq f_{\mu, \lambda} (\tilde{x}).
\end{equation}
The offsets of the affine functions do not depend on $\lambda$ either, therefore we define
\begin{equation}
    a(\mu) \defeq \frac{1}{d-1}\of*{\frac{d \, \cont(\mu)}{\abs{E}} - 1}.
\end{equation}

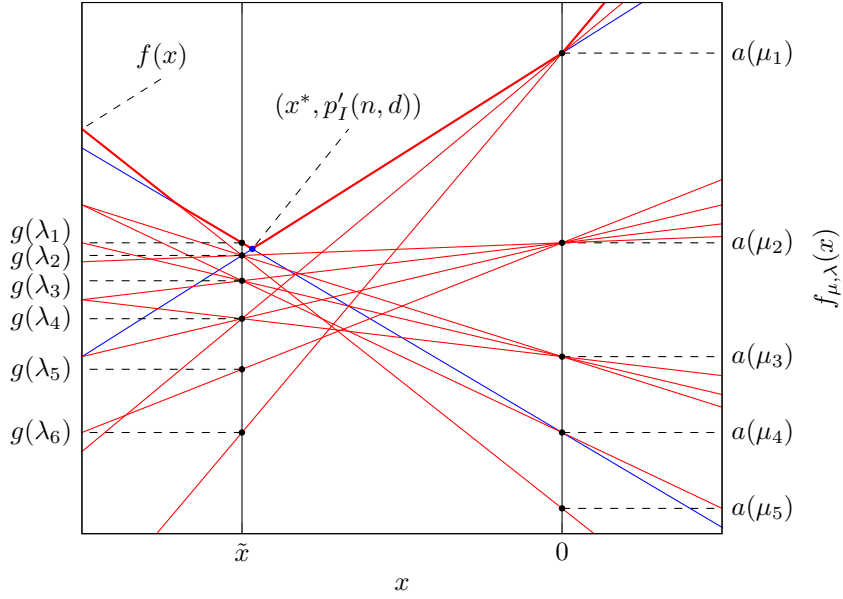
\begin{figure}[!ht]
\centering
\begin{tikzpicture}
\newcommand\xtilde{-0.05}
\newcommand\xstar{-3/62}
\newcommand\fxstar{7/31}

\newcommand{\horLineFromPointRight}[1]{
  \draw[dashed,opacity=0.4] 
  (#1) -- (#1-|{rel axis cs:1,0})
}
\newcommand{\horLineFromPointLeft}[1]{
  \draw[dashed,opacity=0.4] 
  (#1) -- (#1-|{rel axis cs:0,0})
}
\begin{axis}[
    axis lines* = box,
    ytick=\empty,
    xlabel = \(x\),
    ylabel = {\(f_{\mu,\lambda}(x)\)},
    ytick pos = right,
    ylabel near ticks, 
    yticklabel pos=right,
    xmin=1.5*\xtilde, xmax=-0.5*\xtilde,
    ymin=-0.9, ymax=1.2,
    xtick={\xtilde,0}, 
    xticklabels = {$\tilde{x}$,$0$},
    scaled x ticks = false,
    extra y tick style={major y tick style={draw=none},grid=none},
    extra y ticks={1,0.25,-0.2,-0.5,-0.8},
    extra y tick labels={$a(\mu_1)$,$a(\mu_2)$,$a(\mu_3)$,$a(\mu_4)$,$a(\mu_5)$}
]

%vertical lines at x^* and 0
\addplot[samples=100, domain=0:6] coordinates {(-0.05,2)(-0.05,-2)};
\addplot[samples=100, domain=0:6] coordinates {(0,2)(0,-2)};

%projection of mu's to the right Y axis
\pgfplotsinvokeforeach{1,0.25,-0.2,-0.5,-0.8}{
    \addplot[mark=*,mark size=1pt] coordinates {(0,#1)};
    \horLineFromPointRight{0,#1};
}

%intersection of two lines for optimal point
\addplot[samples=100, smooth, domain=-0.1:0.05, color=blue]{-1/2 - 15*x};
\addplot[samples=100, smooth, domain=-0.1:0.05, color=blue]{1 + 16*x};

%intersection highlight
\addplot[mark=*,mark size=1pt,blue] coordinates {(\xstar,\fxstar)};
\draw[ultra thin,shorten <=2pt, opacity=0.5, text opacity=1,dashed] (\xstar,\fxstar) -- (\xtilde/1.5,0.7) node[above] {$(x^*,p'_I(n,d))$};

%function f(x)
\addplot [domain=-0.1:0.05, samples=100, color = red, thick] {max(1+16*x,1+30*x,-1/2-15*x,-4/5-20*x)};
\draw[ultra thin,shorten <=2pt, opacity=0.5, text opacity=1,dashed] (1.5*\xtilde,-4/5-20*1.5*\xtilde) -- (\xtilde*1.25,0.9) node[above] {$f(x)$};

%16 skipped
\foreach \b in {21,30}{
    \addplot[domain=-0.1:0.05,color=red,draw opacity=0.3]{1 + \b*x};
    }

\foreach \b in {1,3,6,10}{
    \addplot[domain=-0.1:0.05,color=red,draw opacity=0.3]{1/4 + \b*x};
    }

\foreach \b in {-8,-6,-3}{
    \addplot[domain=-0.1:0.05,color=red,draw opacity=0.3]{-1/5 + \b*x};
    }

%-15 skipped
\foreach \b in {-12}{
    \addplot[domain=-0.1:0.05,color=red,draw opacity=0.3]{-1/2 + \b*x};
    }

% \foreach \b in {-20,-18}{ % Previous wrong Okada set
\foreach \b in {-20}{
    \addplot[domain=-0.1:0.05,color=red,draw opacity=0.3]{-4/5 + \b*x};
    }

\end{axis}

%projection of lambdas's to the left Y axis
\begin{axis}[
    axis lines* = box,
    ylabel near ticks,
    ytick=\empty,
    xtick=\empty,
    ytick pos = left,
    yticklabel pos=left,
    xmin=1.5*\xtilde, xmax=-0.5*\xtilde,
    ymin=-0.9, ymax=1.2,
    scaled x ticks = false,
    extra y tick style={major y tick style={draw=none},grid=none},
    extra y ticks={0.25,0.2,0.1,-0.05,-0.25,-0.5},
    extra y tick labels={\raisebox{4mm}{$g(\lambda_1)$},\raisebox{-3mm}{$g(\lambda_2)$},\raisebox{-4mm}{$g(\lambda_3)$},$g(\lambda_4)$,$g(\lambda_5)$,$g(\lambda_6)$},
]
\pgfplotsinvokeforeach{0.25,0.2,0.1,-0.05,-0.25,-0.5}{
    \addplot[mark=*,mark size=1pt] coordinates {(\xtilde,#1)};
    \horLineFromPointLeft{\xtilde,#1};
}
\end{axis}
\end{tikzpicture}

\caption{A typical behavior of the spectrum $f_{\mu,\lambda}(x)$ (thin red lines; $f(x)$ is in bold red) for all $(\lambda,\mu) \in { \Omega_{n,d}}$ when $d$ is odd, $n$ is odd and $d \leq n \leq 2d+1$. The coordinate $x=x^*$ corresponds to the optimal value $f(x^*) = p'_I(n,d)$. The plot corresponds to the parameters $n=5$ and $d=3$. The partitions $\lambda$ corresponding to the points $(\tilde{x},g(\lambda))$ are $\lambda_1 = (1^3)$, $\lambda_2 = (1)$, $\lambda_3 = (2,1)$, $\lambda_4 = (3)$, $\lambda_5 = (4,1)$, $\lambda_6 = (5)$. The partitions $\mu$ characterising the offsets $a(\mu)$ for the functions $f_{\mu,\lambda}(x)$ are $\mu_1 = (5)$, $\mu_2 = (4,1)$, $\mu_3 = (3,2)$, $\mu_4 = (3,1,1)$, $\mu_5 = (2,2,1)$. Some other values in that case are $\tilde{x} = 1/20$, $g(\lambda_1) = 1/4$, $x^* = -3/62$, $p'(n,d) = 7/31$.}
\label{fig:plot_n=5_d=3}
\end{figure}

Let $n \geq d$ and $k \defeq \lfloor \frac{n-d}{2} / d \rfloor$ and $m \defeq \frac{n-d}{2} \bmod d$. Then we can define two partitions $\lambda_1, \lambda_2$ in $\Irr{\cB^d_n}$ and three partitions $\mu_1, \mu_2, \widetilde{\mu}_2$\footnote{In the definition above, $\widetilde{\mu}_2$ is given using the column notation $\widetilde{\mu}^\prime_2$. Using the row notation it becomes $\widetilde{\mu}_2 = \left( (2k+3)^m, (2k+1)^{d-m} \right)$: $m$ rows of size $(2k+3)$ and $d-m$ rows of size $(2k+1)$. For example, see \cref{fig:plot_n=5_d=3}, where $\widetilde{\mu}_2$ of the current proof corresponds to $\mu_4$ on the plot.} in $\Irr{\cS^d_n}$:
\begin{equation} \label{eq:5partitions}
    \begin{aligned}
        \lambda_1 &\defeq (1^{d}) \quad & \quad \mu_1 &\defeq (n)\\
        \lambda_2 &\defeq (1) &\mu_2 &\defeq (n-d+1, 1^{d-1}) \\
        & &\widetilde{\mu}_2^\prime &\defeq (d^{2k+1}, m^2).
    \end{aligned}
\end{equation}

\begin{lemma} \label{lem:TwoRegimes}
    Let $d$ and $n$ odd, $n \geq d$, and $\lambda_1, \mu_2$ from \cref{eq:5partitions}, then
    \begin{equation}
        g(\lambda_1) - a(\mu_2) = \frac{2 d + 2 - n}{n - 1}.
    \end{equation}
    In particular $g(\lambda_1) < a(\mu_2)$ if $n \geq 2d + 3$, and $g(\lambda_1) > a(\mu_2)$ if $n \leq 2d + 1$.
\end{lemma}
\begin{proof}
    The content of $\mu_2$ is $\frac{(n-d+1)(n-d)}{2} - \frac{d(d-1)}{2}$. Also $h(\lambda_1) = 0$, so $g(\lambda_1) = \frac{1}{n-1}$. Then
    \begin{align}
        g(\lambda_1) - a(\mu_2) &= \frac{1}{n-1} - \frac{1}{d-1} \left( \frac{d \, \cont(\mu_2)}{|E|} -1 \right) =-d \frac{(n-d+1)(n-d)-d(d-1)}{n(n-1)(d-1)} + \frac{1}{d-1} + \frac{1}{n-1} \nonumber \\
        &=\frac{2 d + 2 - n}{n - 1}.
    \end{align}
\end{proof}

\begin{lemma} \label{lem:LambdaMuOrders}
    Let $d$ and $n$ odd, $n \geq d$, and the partitions defined in \cref{eq:5partitions}, then for all $i$ and $j$ we have $(\lambda_i, \mu_j) \in S$, and the relations
    \begin{equation}
        g(\lambda) \leq g(\lambda_2) \leq g(\lambda_1) \quad \text{ and } \quad a(\mu) \leq a(\mu_1),
    \end{equation}
    for all $\lambda \neq \lambda_1$ in $\Irr{\cB^d_n}$ and all $\mu \neq \mu_1$ in $\Irr{\cS^d_n}$. Moreover for all $\left( \lambda_1, \mu \right) \in  \Omega_{n,d}$ we have
    \begin{equation}
        a(\mu) \leq a(\mu_2).
    \end{equation}
\end{lemma}
\begin{proof}
    By definition of $S$, we have $(\lambda_i, \mu_j) \in S$ for all $i$ and $j$.
    Let $\lambda$ in $\Irr{\cB^d_n}$ then
    \begin{equation}
        g(\lambda) = \frac{1}{n-1} + \tilde{x} h(\lambda),
    \end{equation}
    with $\tilde{x} < 0$ and $h(\lambda) = \frac{1}{2} \sum^{\lambda_1}_{i=0} \lambda^\prime_i (d - \lambda^\prime_i + 2 (i - 1))$. But since $\lambda^\prime_1 \leq d$ holds for all $\lambda$ in $\Irr{\cB^d_n}$, then $h(\lambda) \geq 0$. In particular, $h(\lambda_2) = \frac{d-1}{2}$, and $h(\lambda) = 0$ iff $\lambda = \lambda_1$. Assume there exists $\lambda$ in $\Irr{\cB^d_n}$ such that
    \begin{equation}
        g(\lambda_2) \leq g(\lambda) \leq g(\lambda_1),
    \end{equation}
    then necessarily the first term of $h(\lambda)$ is either $h(\lambda_1)$ or $h(\lambda_2)$, since it minimized for the columns $(1)$ and $(1^d)$. But since all the terms in $h(\lambda)$ are positive, then either $g(\lambda) = g(\lambda_1)$ or $g(\lambda) = g(\lambda_2)$.

    Because $\mu_1$ is the $n$-box Young diagram that maximizes the content function, then
    \begin{equation}
        a(\mu) \leq a(\mu_1),
    \end{equation}
    for all $\mu \neq \mu_1$ in $\Irr{\cS^d_n}$.
    
    Assume there exists $\mu$ such that $\left( \lambda_1, \mu \right) \in  \Omega_{n,d}$ and
    \begin{equation}
        a(\mu_2) \leq a(\mu) \leq a(\mu_1).
    \end{equation}
    Since $\left( (1^d), \mu \right) \in S$ implies $\left( (1^d), \mu \right) \in  \Omega_{n,d}$, then $\left( \lambda_1, \mu \right) \in S$, and by definition $r(\mu) = d$. Thus necessarily $\cont(\mu_2) \leq \cont(\mu)$, which implies that the first row of $\mu_2$ is of size at most $n-d+1$. But the content of a Young diagram is a non-decreasing function of the first row's size, i.e.~for all Young diagrams $\nu \pt n$ and $\mu \pt n$ with the same number of boxes $n$ such that $\nu_1 \leq \mu_1$ we have $\cont(\nu) \leq \cont(\mu)$. Then $\mu_2$ and $\mu$ share the same first row, and $\mu_2 = \mu$.
\end{proof}

\begin{theorem} \label{thm:smallDimension}
   When $d$ is odd, $n$ is odd and $n \geq d$, the value $p'_I(n, d)$ is
    \begin{equation}
        p'_I(n, d) = \min \left( \frac{2 d + 1}{2 d n + 1}, \frac{1}{n-1} \right).
    \end{equation}
\end{theorem}
\begin{proof}
    Let $\lambda_1, \lambda_2$ and $\mu_1, \mu_2, \widetilde{\mu}_2$ as in \cref{eq:5partitions}, and let us prove that the optimal solution of the dual problem is $p'_I(n,d) = g(\lambda_1)$ when $n \geq 2d + 3$, and lies at intersection of the affine functions $f_{\mu_1, \lambda_2}$ and $f_{\mu_2, \lambda_1}$, when $n \leq 2d + 1$. Now since 
    \begin{align}
        \cont(\widetilde{\mu}_2) &= \sum_{i=1}^{2k+1} \of*{-\frac{d(d-1)}{2}+(i-1)d} + \sum_{i=2k+2}^{2k+3} \of*{-\frac{m(m-1)}{2}+(i-1)m} \nonumber \\
        &= \frac{d(2k+1)(2k+1-d)}{2}+m(4k+4-m),
    \end{align} 
    and $n-d=2kd+2m$ with $m \in \set{0,\dotsc,d-1}$, we can write
    \begin{align}
        g(\lambda_1) - a(\widetilde{\mu}_2) &= \frac{1}{n-1} - \frac{1}{d-1} \left( \frac{d \, \cont(\widetilde{\mu}_2)}{|E|} -1 \right) = \frac{n(d+n-2) - 2 d \, \cont(\widetilde{\mu}_2)}{n(n-1)(d-1)} \nonumber \\
        &= \frac{(d + 2)(2m^2 - 2dm - n + dn)}{n(n-1)(d-1)} \geq \frac{(d + 2)\of[\big]{-(d-1)(d+1)/2 + d(d-1)}}{n(n-1)(d-1)} \nonumber \\
        &= \frac{(d + 2)(d - 1)}{2n(n-1)} > 0,
    \end{align}
    where in the first inequality we used $n \geq d$ and that the minimum of $2m^2 - 2dm$ on the domain $m \in \set{0,\dotsc,d-1}$ is achieved for $m=(d-1)/2$. Geometrically this means that the point $(0,a(\widetilde{\mu}_2))$ is always lower than $(\tilde{x},g(\lambda_1))$.
    
    Suppose that $n \geq 2d + 3$, then the relation
    \begin{equation}
        g(\lambda_1) < a(\mu_2),
    \end{equation}
    holds by \cref{lem:TwoRegimes}. Therefore $p'_I(n,d) \geq g(\lambda_1)$ since the optimal point should be above the intersection of the affine functions $f_{\mu_2, \lambda_1}$ and $f_{\widetilde{\mu}_2, \lambda_1}$ that is, above $g(\lambda_1)$. But since $g(\lambda) \leq g(\lambda_1)$ for all $\lambda$ in $\Irr{\cB^d_n}$ by \cref{lem:LambdaMuOrders}, it must be $p'_I(n,d) = g(\lambda_1)$.
    
    Suppose that $n \leq 2d + 1$, then the relation
    \begin{equation}
        a(\mu) \leq a(\mu_2),
    \end{equation}
    holds for all $\left( \lambda_1, \mu \right) \in  \Omega_{n,d}$, by \cref{lem:LambdaMuOrders}. Then $p'_I(n,d)$ lies above the affine function $f_{\mu_2, \lambda_1}$. But $g(\lambda) \leq g(\lambda_1)$ for all $\lambda$ in $\Irr{\cB^d_n}$, by \cref{lem:LambdaMuOrders}, then $p'_I(n,d)$ must lie on the affine function $f_{\mu_2, \lambda_1}$, at the intersection with another affine function $f_{\mu, \lambda}$ with $g(\lambda) \leq a(\mu)$. Among all such affine functions there are no functions with $\lambda = \lambda_1$ due to \cref{lem:LambdaMuOrders}. Because $g(\lambda) \leq g(\lambda_2)$ for all $\lambda \neq \lambda_1$ in $\Irr{\cB^d_n}$ and $a(\mu)\leq a(\mu_1)$ for all $\mu \in \Irr{\cS^d_n}$, by \cref{lem:LambdaMuOrders}, it must be that this function is $f_{\mu_1, \lambda_2}$. Therefore $p'_I(n,d)$ lies at intersection of the affine functions $f_{\mu_1, \lambda_2}$ and $f_{\mu_2, \lambda_1}$.  
    
    In order to find the intersection of $f_{\mu_1, \lambda_2}$ and $f_{\mu_2, \lambda_1}$ we need to solve $p'_I(n,d) \defeq f_{\mu_1, \lambda_2}(x^*) = f_{\mu_2, \lambda_1}(x^*)$, which gives $x^* = \frac{4d}{(d-1)(n-1)(2dn+1)}$ and $p'_I(n,d) = \frac{2d+1}{2dn+1}$.
    
    In conclusion, when $n \geq 2d + 3$ then $p'_I(n,d) = \frac{1}{n-1}$, and when $n \leq 2 d + 1$ then $p'(n,d) = \frac{2d+1}{2dn+1}$, which is equivalent to the statement of the theorem.
\end{proof}

%%%%%%%%%%%%%%%%%%%%%%%%%%%%%%%%%%%%%%%%%%%%%%%%%%%%%%%%%%%%%%%%%%%%%%%%%%%%%%%%%%%%%%%%%%%%%%%%%%%%%%%%%%%%%%%%%%%%%%%%%%%%%%%%%%%%%%%%%%%%%%%%%%%%%%%%%%%%%%%%%%%%%%%%%%%%%%%%%%%%%%%%%%%%%%%%%%%%%%%%%%%%%%%%%%%%%%%%%%
%%%%%%%%%%%%%%%%%%%%%%%%%%%%%%%%%%%%%%%%%%%%%%%%%%%%%%%%%%%%%%%%%%%%%%%%%%%%%%%%%%%%%%%%%%%%%%%%%%%%%%%%%%%%%%%%%%%%%%%%%%%%%%%%%%%%%%%%%%%%%%%%%%%%%%%%%%%%%%%%%%%%%%%%%%%%%%%%%%%%%%%%%%%%%%%%%%%%%%%%%%%%%%%%%%%%%%%%%%
\subsection{Brauer states}

Understanding the $K_n$-extendibility in full generality in the case of Brauer states
\begin{equation}
    p \cdot \Pi_{\0} + q \cdot \tfrac{\Pi_{\ydsm{1,1}}}{\Tr \Pi_{\ydsm{1,1}}} + (1 - p - q) \cdot \tfrac{\Pi_{\ydsm{2}}}{\Tr \Pi_{\ydsm{2}}}
\end{equation}
means to find the full 2D region of allowed $(p,q)$ values for arbitrary $n$ and $d$. 
As a first step to solve this general problem, this section aims to determine the maximum values of $q$ and $p$, denoted $q_B(n,d)$ and $p_B(n,d)$, for the $K_n$-extendibility of Brauer states, as well as the complete $K_n$-extendibility polytope of qubit Brauer states. 

%%%%%%%%%%%%%%%%%%%%%%%%%%%%%%%%%%%%%%%%%%%%%%%%%%%%%%%%%%%%%%%%%%%%%%%%%%%%%%%%%%%%%%%%%%%%%%%%%%%%%%%%%%%%%%%%%%%%%%%%%%%%%%%%%%%%%%%%%%%%%%%%%%%%%%%%%%%%%%%%%%%%%%%%%%%%%%%%%%%%%%%%%%%%%%%%%%%%%%%%%%%%%%%%%%%%%%%%%%
\subsubsection{Maximizing the $q_B(n,d)$}

We define $q_B(n,d)$ formally as follows:
\begin{equation} \label{def:qB}
    q_B(n,d) \defeq  \max_{\rho,q,p} \, q \quad
        \textrm{s.t.} \quad
         \rho_e = p \cdot \Pi_{\0} + q \cdot \tfrac{\Pi_{\ydsm{1,1}}}{\Tr \Pi_{\ydsm{1,1}}} + (1 - p - q) \cdot \tfrac{\Pi_{\ydsm{2}}}{\Tr \Pi_{\ydsm{2}}} \quad \forall e \in E, \quad \Tr [\rho] = 1, \quad \rho \succeq 0.
\end{equation}

It turns out that this value is the same as the corresponding value for Werner state case.

\begin{lemma} \label{lem:qB=qW}
    For every $n$ and $d$ the following relation holds
    \begin{equation}
        q_B(n,d) = q_W(n,d).
    \end{equation}
\end{lemma}
\begin{proof}
Given any solution $\rho$ to the optimization problem (\ref{def:qB}), we can twirl it over unitary group action $U\xp{n}$ without changing the value of the optimization problem since $\Pi_{\ydsm{1,1}}$ is invariant under $U\xp{2}$ action. It means that 2-body marginals of the twirled $\rho$ are Werner states, and we reduced the problem of calculating $q_B(n,d)$ to calculating $q_W(n,d)$.
\end{proof}

Now, we move to the more complicated case of understanding the value $p_B(n,d)$.

%%%%%%%%%%%%%%%%%%%%%%%%%%%%%%%%%%%%%%%%%%%%%%%%%%%%%%%%%%%%%%%%%%%%%%%%%%%%%%%%%%%%%%%%%%%%%%%%%%%%%%%%%%%%%%%%%%%%%%%%%%%%%%%%%%%%%%%%%%%%%%%%%%%%%%%%%%%%%%%%%%%%%%%%%%%%%%%%%%%%%%%%%%%%%%%%%%%%%%%%%%%%%%%%%%%%%%%%%%
\subsubsection{Maximizing the $p_B(n,d)$} \label{sec:Brauer_max_p}

\noindent Consider now \cref{def:SDPworst} for the projector onto maximally entangled state $\Pi \defeq \frac{\W}{d}$ on the complete graph $K_n$. We have due to \cref{def:sdp_dual_worst_edgetransitive}:
\begin{equation}
    p_{B}(n,d) = \frac{1}{d \cdot \abs{E}} \lambda_{\mathrm{max}}\of[\bigg]{\sum_{e \in E} \W_e},
\end{equation}
and with the help of Jucys--Murphy elements we can get the spectrum of the Hamiltonian $H_B \defeq \sum_{e \in E} \W_e$. Just note that the sum of Jucys--Murphy elements for the symmetric group algebra and the Brauer algebra commute. It means that we can subtract the corresponding spectra to get 
\begin{equation}
    \spec(H_B) = \set[\big]{ g(\mu,\lambda) \, \big| \, \mu \in \Irr{\S_n^d}, \, \lambda \in \Irr{\cB_n^d}, \, (\lambda,\mu) \in  \Omega_{n,d}},
\end{equation}
where $ \Omega_{n,d}$ is defined in \cref{def:okada_set} and
\begin{equation}
    g(\mu,\lambda) \defeq \cont(\mu) - \cont(\lambda) + \frac{(n-\abs{\lambda})(d-1)}{2}.
\end{equation}
Therefore
\begin{equation}
    p_{B}(n,d) = \frac{1}{d \cdot \abs{E}} \max_{(\mu,\lambda) \in  \Omega_{n,d}} g(\mu,\lambda).
\end{equation}

\begin{theorem} \label{thm:BrauerStates}
    The optimization problem \cref{def:p_values_Kn_Brauer} has the optimal value
    \begin{equation}
        p_{B}(n,d) = \frac{1}{d} + \of*{1-\frac{1}{d}} \frac{1}{n + n \bmod 2 - 1}.
    \end{equation}
\end{theorem}

\begin{figure}[!ht]
    \centering
    \begin{NiceTabular}{c|cccccccc}[columns-width = 2em, cell-space-limits = 0.25em]
        \CodeBefore
        \Body
        \diagbox{$d$}{$n$} & 2 & 3 & 4 & 5 & 6 & 7 & 8 & 9 \\ \hline
        2 & 1 & \nicefrac{2}{3} & \nicefrac{2}{3} & \nicefrac{3}{5} & \nicefrac{3}{5} & \nicefrac{4}{7} & \nicefrac{4}{7} & \nicefrac{5}{9} \\ 
        3 & 1 & \nicefrac{5}{9} & \nicefrac{5}{9} & \nicefrac{7}{15} & \nicefrac{7}{15} & \nicefrac{3}{7} & \nicefrac{3}{7} & \nicefrac{11}{27} \\ 
        4 & 1 & \nicefrac{1}{2} & \nicefrac{1}{2} & \nicefrac{2}{5} & \nicefrac{2}{5} & \nicefrac{5}{14} & \nicefrac{5}{14} & \nicefrac{1}{3} \\ 
        5 & 1 & \nicefrac{7}{15} & \nicefrac{7}{15} & \nicefrac{9}{25} & \nicefrac{9}{25} & \nicefrac{11}{35} & \nicefrac{11}{35} & \nicefrac{13}{45} \\ 
        6 & 1 & \nicefrac{4}{9} & \nicefrac{4}{9} & \nicefrac{1}{3} & \nicefrac{1}{3} & \nicefrac{2}{7} & \nicefrac{2}{7} & \nicefrac{7}{27} \\ 
        7 & 1 & \nicefrac{3}{7} & \nicefrac{3}{7} & \nicefrac{11}{35} & \nicefrac{11}{35} & \nicefrac{13}{49} & \nicefrac{13}{49} & \nicefrac{5}{21} \\ 
        8 & 1 & \nicefrac{5}{12} & \nicefrac{5}{12} & \nicefrac{3}{10} & \nicefrac{3}{10} & \nicefrac{1}{4} & \nicefrac{1}{4} & \nicefrac{2}{9} \\ 
        9 & 1 & \nicefrac{11}{27} & \nicefrac{11}{27} & \nicefrac{13}{45} & \nicefrac{13}{45} & \nicefrac{5}{21} & \nicefrac{5}{21} & \nicefrac{17}{81} \\ 
    \end{NiceTabular}
    \caption{The first largest values of $p_B(n,d)$ of the $K_n$-extendible Brauer states. All 2-body Brauer states corresponding to these values are entangled due to \cref{cor:entangled Brauer}, in contrast to \cref{fig:firstValuesWernerStates}.}
\end{figure}

\begin{proof}
    Let $\lambda \vdash n - 2r, \, \lambda \in \Irr{\cB_n^d}$ for some fixed $r$ and define 
    \begin{equation}
       f(\lambda) \defeq - \cont(\lambda) + \frac{(n-\abs{\lambda})(d-1)}{2}.
    \end{equation}
    Since the content of a Young diagram $\lambda$ is a non-decreasing function of the first row's size (when the number of boxes is fixed), then in order to maximize $f(\lambda)$ we can assume that $\lambda$ takes the most rectangular shape possible in order to minimize $\cont(\lambda)$. Namely, $f(\lambda^*) \geq f(\lambda)$ where $\lambda^* \pt n - 2r$
    \begin{equation}
        \lambda^*_1 = \dotsc = \lambda^*_k = \frac{n - 2r - k}{d} + 1, \quad \lambda^*_{k+1} = \dotsc = \lambda^*_d = \frac{n - 2r - k}{d},
    \end{equation}
    and $k \defeq n-2r \bmod d$. Then, in particular, we can use our calculation from \cref{eq:werner_J_cont_rectangle} to write 
    \begin{equation}\label{eq:thm_brauer_f}
        f(\lambda^*) = \frac{1}{2d} \of*{nd(d-1) - \abs{\lambda^*}^2 + d\abs{\lambda^*} - k(d-k)(d+1)}
    \end{equation}
    Assume $\abs{\lambda^*} > d$ and consider three cases:
    \begin{enumerate}
        \item $n-2r-k$ is even. Define $\tilde{r} \defeq \tfrac{n-2r-k}{2}$. Then $n - 2r - 2\tilde{r} = k$ and we can set $\tilde{\lambda} \pt k$ as a vertical one column Young diagram $\tilde{\lambda} \defeq (1^k)$. Since $\tilde{k} \defeq \abs{\tilde{\lambda}} \bmod d = k$. We see from \cref{eq:thm_brauer_f} that $f(\tilde{\lambda}) > f(\lambda^*)$.
        \item $n-2r-k$ is odd and $k=0$. Then also $d$ is odd, meaning that we can set $\tilde{\lambda} \pt d$, $\tilde{\lambda} \defeq (1^{d})$ preserving $\tilde{k} = k = 0$. Again we deduce from \cref{eq:thm_brauer_f} that $f(\tilde{\lambda}) > f(\lambda^*)$.
        \item $n-2r-k$ is odd and $k>0$. In this case we define $\tilde{\lambda} \pt k - 1$, $\tilde{\lambda} \defeq (1^{k-1})$. Using $\abs{\lambda^*} \geq d + k$, we can estimate the difference $f(\tilde{\lambda}) - f(\lambda^*)$ as
        \begin{align}
            2d\of{f(\tilde{\lambda}) - f(\lambda^*)} &= \abs{\lambda^*}^2 - \abs{\tilde{\lambda}}^2 - d\of{\abs{\lambda^*} - \abs{\tilde{\lambda}}} - (k-1)(d-k+1)(d+1) + k(d-k)(d+1) \nonumber \\
            &= \abs{\lambda^*} \of{\abs{\lambda^*} - d} - (k-1)^2 + d(k-1) + (d+1) \of{1 + d - 2k}\nonumber \\
            &= \abs{\lambda^*} \of{\abs{\lambda^*} - d} + d^2 - k^2 + d - kd \nonumber \\
            &\geq (d+k)k + d^2 -k^2 + d - kd = d(d+1) > 0,
        \end{align}
        meaning that again $f(\tilde{\lambda}) > f(\lambda^*)$.
    \end{enumerate}
    The above analysis means that we can assume without loss of generality that the maximizer of the function $f(\lambda)$ over $\lambda \in \Irr{\cB_n^d}$ is a Young diagram $\lambda$ with one column only, i.e.~$\lambda = (1^k)$ for some $k \leq d$. Therefore 
    \begin{equation}
        \max_{(\mu,\lambda) \in  \Omega_{n,d}} g(\mu,\lambda) = \max_{(\mu,\lambda) \in \Gamma_{n,d}} g(\mu,\lambda) = \begin{cases}
            g((n),\0) &\text{ $n$ even, }\\
            g((n),(1)) &\text{ $n$ odd, }
        \end{cases}
    \end{equation}
    which gives us 
    \begin{equation}
        p_{B}(n,d) = \frac{2}{dn(n-1)} \of*{\frac{n(n-1)}{2} + \frac{(n - n \bmod 2)(d-1)}{2}} = \frac{1}{d} + \frac{(n - n \bmod 2)(d-1)}{dn(n-1)}.
    \end{equation}
\end{proof}

\begin{corollary}\label{cor:entangled Brauer}
    For all $n$ and $d$, any feasible solution $\rho$ for the optimal value $p_B(n,d)$, has Brauer states $\rho_e$ entangled.
\end{corollary}
\begin{proof}
    Using the \textsc{PPT} criterion from \cref{app:BrauerStates}, an element from the two-parameter $(p,q)$ family of Brauer states for any fixed $d \geq 2$, is separable if and only if it lies in the region specified by
    \begin{equation}
        \begin{cases}
            0 \leq p \leq \frac{1}{d} \\
            0 \leq q \leq \frac{1}{2} 
        \end{cases}
    \end{equation}
    But from \cref{thm:BrauerStates}, we always have $p_{B}(n,d) > \frac{1}{d}$ for any finite $n$.
\end{proof}

\begin{theorem}
    For all $n$ and $d$, there exists a feasible solution $\rho$ for the optimal value $p_B(n,d)$ such that for all edges $e$:
    \begin{equation*}
        \rho_e = p_B(n,d) \cdot \Pi_\0 + \big( 1 - p_B(n,d) \big) \cdot \tfrac{\Pi_{\ydsm{2}}}{\Tr \Pi_{\ydsm{2}}}.
    \end{equation*}
\end{theorem}
\begin{proof}
    We are going to solve the following optimization problem:
    \begin{equation}
    \begin{aligned}
        p^*(n,d) = \max_{\rho,p} \quad & p \qquad
        \textrm{s.t.} \quad
         \rho_e = p \cdot \Pi_\0 + \of*{ 1 - p} \cdot \tfrac{\Pi_{\ydsm{2}}}{\Tr \Pi_{\ydsm{2}}} \quad \forall e \in E, \quad \Tr [\rho] = 1, \quad \rho \succeq 0.
    \end{aligned}
    \end{equation}
    and show that, in fact, $p^*(n,d) = p_B(n,d)$. This suffices to prove the claim. In \cref{app:brauer_q=0} we show that actually 
    \begin{equation}
    \begin{aligned}
        p^*(n,d)  = \min_{x \in \mathbb{R}} \max_{ (\lambda,\mu) \in \Omega_{n,d} } \quad & \frac{1}{d \cdot \abs{E}} \of[\bigg]{ \cont(\mu) - \cont(\lambda) + \frac{(n - \abs{\lambda})(d-1)}{2}} + \frac{x}{d} \of[\bigg]{1 - \frac{\cont(\mu)}{\abs{E}}}.
    \end{aligned}
    \end{equation}
    Equivalently, we want to minimize over $x \in \mathbb{R}$ a piecewise linear function $h(x) \defeq \max_{ (\lambda,\mu) \in \Omega_{n,d} } h_{\lambda, \mu}(x)$, where we define affine functions $h_{\lambda, \mu}(x)$ as
    \begin{equation}
        h_{\lambda, \mu}(x) \defeq \frac{1}{d \cdot \abs{E}} \of[\bigg]{ \cont(\mu) - \cont(\lambda) + \frac{(n - \abs{\lambda})(d-1)}{2}} + \frac{x}{d} \of[\bigg]{1 - \frac{\cont(\mu)}{\abs{E}}}.
    \end{equation}
    Note that the finite optimum value $p^*(n,d)$ is always achieved at an intersection of at least two different affine functions $h_{\lambda, \mu}(x)$, and there must be at least one function among them with non-positive slope and at least one with non-negative slope. However, $\cont(\mu) \leq \abs{E}$ so all functions $h_{\lambda, \mu}(x)$ have non-negative slopes. Moreover, the functions $h_{\lambda, \mu}(x)$ with $\mu = (n)$ are the only ones which have zero slopes (because this is the only $\mu$ which achieves $\cont(\mu) = \abs{E}$). So it means that the optimum value is achieved for $\mu = (n)$ and some $\lambda$ such that $(\lambda,\mu) \in \Omega_{n,d}$ with the formula 
    \begin{equation}
        p^*(n,d) = \max_{ (\lambda,(n)) \in \Omega_{n,d} } \frac{1}{d \cdot \abs{E}} \of[\bigg]{ \abs{E} - \cont(\lambda) + \frac{(n - \abs{\lambda})(d-1)}{2}},
    \end{equation}
    where we used the fact that $\cont(\mu) = \abs{E}$. But we have already calculated this quantity in \cref{thm:BrauerStates}, so 
    \begin{equation}
        p^*(n,d) = \frac{1}{d \cdot \abs{E}} \max_{(\mu,\lambda) \in  \Omega_{n,d}} g(\mu,\lambda) = p_B(n,d),
    \end{equation}
    which proves the claim.
\end{proof}

%%%%%%%%%%%%%%%%%%%%%%%%%%%%%%%%%%%%%%%%%%%%%%%%%%%%%%%%%%%%%%%%%%%%%%%%%%%%%%%%%%%%%%%%%%%%%%%%%%%%%%%%%%%%%%%%%%%%%%%%%%%%%%%%%%%%%%%%%%%%%%%%%%%%%%%%%%%%%%%%%%%%%%%%%%%%%%%%%%%%%%%%%%%%%%%%%%%%%%%%%%%%%%%%%%%%%%%%%%
\subsubsection{$K_n$-extendibility polytope for qubits}

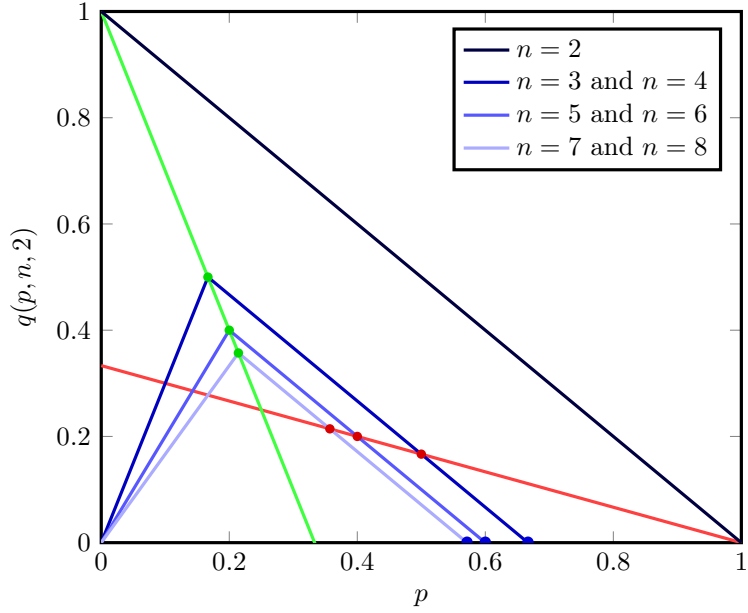
\begin{figure}[H]
    \centering
    \begin{tikzpicture}
        \begin{axis}[
            xlabel={$p$},
            ylabel={$q(p, n, 2)$},
            xmin=0, xmax=1,
            ymin=0, ymax=1,
            legend pos=north east,
            legend cell align=left,
            no marks,
            samples=100,
            very thick
        ]

            % Isotropic states
            \addplot+[domain=0:1, red!75!white, forget plot] {(1 - x) / 3};

            % Werner states
            \addplot+[domain=0:1, green!75!white, forget plot] {1 - 3 * x)};
        
            % n=2
            \addplot+[domain=0:1, blue!25!black] {1 - x};
            \addlegendentry{$n=2$}
            
            % n=3 and n=4
            \addplot+[domain=0:1/6, blue!75!black] {3 * x};
            \addplot+[domain=1/6:2/3, blue!75!black, forget plot] {(2 - 3 * x) / 3};
            \addlegendentry{$n=3$ and $n=4$}
            
            % n=5 and n=6
            \addplot+[domain=0:1/5, blue!66!white] {2 * x};
            \addplot+[domain=1/5:3/5, blue!66!white, forget plot] {(3 - 5 * x) / 5};
            \addlegendentry{$n=5$ and $n=6$}
    
            % n=7 and n=8
            \addplot+[domain=0:3/14, blue!33!white] {(5 * x) / 3};
            \addplot+[domain=3/14:4/7, blue!33!white, forget plot] {(4 - 7 * x) / 7};
            \addlegendentry{$n=7$ and $n=8$}

            % Isotropic states
            \node[circle, fill, inner sep=1.25pt, red!85!black] at (1/2,1/6) {};
            \node[circle, fill, inner sep=1.25pt, red!85!black] at (2/5,1/5) {};
            \node[circle, fill, inner sep=1.25pt, red!85!black] at (5/14,3/14) {};

            % Werner states
            \node[circle, fill, inner sep=1.25pt, green!85!black] at (1/6,1/2) {};
            \node[circle, fill, inner sep=1.25pt, green!85!black] at (1/5,2/5) {};
            \node[circle, fill, inner sep=1.25pt, green!85!black] at (3/14,5/14) {};

            % Brauer states
            \node[circle, fill, inner sep=1.5pt, blue!85!black] at (2/3,0.001) {};
            \node[circle, fill, inner sep=1.5pt, blue!85!black] at (3/5,0.001) {};
            \node[circle, fill, inner sep=1.5pt, blue!85!black] at (4/7,0.001) {};
        \end{axis}
    \end{tikzpicture}
    \caption{In blue, the optimal values of $q(p, n, 2)$ for $n \in \{2, \ldots, 8\}$. The $K_n$-extendibility for qubits Brauer states is the polytope given by the extremal points of those piecewice functions. The blues dots correspond to the optimal values of \cref{thm:BrauerStates}. The green line corresponds to the parameters $(p,q)$ of Werner states, and the green dots to the optimal values of \cref{thm:WernerStates}.
    The red line corresponds to the parameters $(p,q)$ of isotropic states, and the red dots to the optimal values of \cref{thm:isotropicStates}.
    }
    \label{fig:BrauerQubitPolytope}
\end{figure}

To understand the complete $K_n$-extendibility for qubit Brauer states, we will solve the following optimization problem
\begin{equation}
    \begin{aligned}
        q(p, n, d) = \max_{\rho,q} \quad & q \\
        \textrm{s.t.} \quad & \rho_e = p \cdot \Pi_{\0} + q \cdot \tfrac{\Pi_{\ydsm{1,1}}}{\Tr \Pi_{\ydsm{1,1}}} + (1 - p - q) \cdot \tfrac{\Pi_{\ydsm{2}}}{\Tr \Pi_{\ydsm{2}}}, \quad \forall e \in E\\
        &\rho \succeq 0,
\end{aligned}
\end{equation}
for all $p \in [0, 1]$, $n \geq 2$ and $d = 2$, i.e. the largest values of $q$, given a fixed $p$ such that the Brauer state with parameter $(p,q)$ is $K_n$-extendible. Note that from \cref{sec:Brauer_max_p}, the optimal solution $q(p, n, d)$ is zero when $p$ is larger than $p_{B}(n,d)$.

Recall that from \cref{rem:okada_set_qubit}, the complete set $\Omega_{n,2}$ is known: $(\lambda, \mu) \in \Omega_{n,2}$ if and only if $\lambda_1 \leq \mu_1 - \mu_2$, with the exceptions of $\lambda = \varnothing$, in which case both rows of $\mu$ must be even, and $\lambda = (1,1)$, in which case both rows of $\mu$ must be odd.

In \cref{app:brauer_p} we show that $q(p, n, 2)$ is equal to the following optimization problem:
\begin{equation}
    \begin{aligned}
         q(p,n,2) &= \min_{x \in \mathbb{R}} \max_{ (\lambda,\mu) \in \Omega_{n,2} } \; f_{\mu,\lambda}(p, x) \\
         f_{\mu,\lambda}(p, x) &= \frac{x}{2 \cdot \abs{E}} \of[\bigg]{ \cont(\mu) - \cont(\lambda) + \frac{(n - \abs{\lambda})}{2}} + \frac{1}{2} \of[\bigg]{1 - \frac{\cont(\mu)}{\abs{E}}} - p \cdot x.
    \end{aligned}
\end{equation}
Similarly to the proof of the $K_n$-extendibility of isotropic states \cref{sec:isotropic}, let $\tilde{x} \coloneqq 1$, and define the two functions:
\begin{equation*}
    g(p, \lambda) \coloneqq f_{\mu,\lambda}(p, \tilde{x}) \quad \text{ and } \quad a(\mu) \coloneqq f_{\mu,\lambda}(p, 0),
\end{equation*}
where $f_{\mu,\lambda}(p, \tilde{x})$ does not depend on $\mu$, and $f_{\mu,\lambda}(p, 0)$ does not depend on $\lambda$ nor $p$.

Let $\mu$ in $\Omega_{n,2}$, then $\mu = (n-k,k)$ for some $k \in \{0, \ldots, \lfloor \tfrac{n}{2} \rfloor\}$, and hence there are $\lfloor \tfrac{n}{2} \rfloor + 1$ such possible $\mu$, namely:
\begin{equation*}
    (n,0), \qquad (n-1,1), \qquad (n-2,2), \quad \cdots \quad \big( n - \lfloor \tfrac{n}{2} \rfloor , \lfloor \tfrac{n}{2} \rfloor \big).
\end{equation*}
In the following Lemma we will prove that for all all $\mu = (n-k,k)$ there exist indeed $\lambda_1$ and $\lambda_2$ defined by
\begin{equation} \label{eq:lambdasBrauer}
    \lambda_1 \coloneqq
    \begin{cases}
        (1) &\text{if $n$ is odd} \\
        (1,1) &\text{if $n$ is even}
    \end{cases}
    \quad \text{ and } \quad \lambda_2 \coloneqq
    \begin{cases}
        (1) &\text{if $n$ is odd} \\
        \0 &\text{if $n$ is even},
    \end{cases}
\end{equation}
such that either $(\lambda_1, \mu) \in \Omega_{n,2}$ or $(\lambda_2, \mu) \in \Omega_{n,2}$
\begin{lemma}\label{lem:lambdasBrauer}
    Let $\lambda_1$, $\lambda_2$ be the partitions defined by \cref{eq:lambdasBrauer}, and let $\mu = (n-k,k)$ for some $k \in \{0, \ldots, \lfloor \tfrac{n}{2} \rfloor\}$, then:
    \begin{itemize}
        \item if $n$ is odd, both $(\lambda_1, \mu)$ and $(\lambda_2, \mu)$ are in $\Omega_{n,2}$,
        \item if $n$ is even and $k$ is odd, $(\lambda_1, \mu) \in \Omega_{n,2}$,
        \item if $n$ is even and $k$ is even, $(\lambda_2, \mu) \in \Omega_{n,2}$.
    \end{itemize}
    Moreover for all $\lambda$ in $\Omega_{n,2}$, then
    \begin{equation*}
        g(p, \lambda_1) = g(p, \lambda_2) \geq g(p, \lambda).
    \end{equation*}
\end{lemma}
\begin{proof}
    The first part of the Lemma is a direct consequence of the \cref{rem:okada_set_qubit}: if $n$ is odd then $\mu_1 - \mu_2 \in \{1, \ldots, n\}$ is always larger or equal to $1$, otherwise both rows of $\mu$ have the same parity than $k$, and the result holds by the two exception rules of \cref{rem:okada_set_qubit}.
    
    Note that given any $\lambda, \lambda'$ in $\Omega_{n,2}$, such that $g(0, \lambda) \geq g(0, \lambda')$, the inequality $g(p, \lambda) \geq g(p, \lambda')$ holds for any $p \in [0,1]$. Hence we will prove the second part of the Lemma for $p=0$. We have
    \begin{equation*}
        g(0, \lambda_1) = g(0, \lambda_2) = \frac{n}{2 (n - 1)}.
    \end{equation*}
    Using that $\lambda = \{1,1\}$ is the only possible vertical Young diagram (i.e. with negative content) of $\Irr{\cB^2_n}$, since
    \begin{align}
        \Irr{\cB^2_{2n}} &= \{1,1\} \cup \set[\Big]{(n - 2r) \:\Big\lvert\: r \in \{ 0, \ldots, n\} } \\
        \Irr{\cB^2_{2n+1}} &= \set[\Big]{(n - 2r) \:\Big\lvert\: r \in \big\{ 0, \ldots, n \big\} },
    \end{align}
    and given that
    \begin{equation}
        g \big( 0, (n-2r) \big) = \frac{- \cont \big( (n-2r) \big) + r}{2 \cdot |E|} + \frac{1}{2},
    \end{equation}
    is an increasing function of $r \in \big\{ 0, \ldots, \lfloor \tfrac{n}{2} \rfloor \big\}$ with maximum at $r = \lfloor \tfrac{n}{2} \rfloor$, i.e. $g \big( 0, \0 \big)$ if $n$ is even, and $g \big( 0, \{1\} \big)$ if $n$ is odd; we conclude that
    \begin{equation}
        g(0, \lambda_1) = g(0, \lambda_2) \geq g \big( 0, (n-2r) \big),
    \end{equation}
    for all $r \in \big\{ 0, \ldots, \lfloor \tfrac{n}{2} \rfloor \big\}$, and thus in particular for all $\lambda$ in $\Omega_{n,2}$.
\end{proof}

Let $\lambda$ in $\Omega_{n,2}$, then either $\lambda$ equals $\{1,1\}$ or $\0$, or there exists a $r \in \big\{ 0, \ldots, \lfloor \tfrac{n-1}{2} \rfloor \big\}$ such that $\lambda = (n-2r)$. In the next Lemma we will see which pair $(\mu, \lambda)$ are in $\Omega_{n,2}$ in the later case.
\begin{lemma}\label{lem:musBrauer}
    Let $\lambda = (n-2r)$ for some $r \in \big\{ 0, \ldots, \lfloor \tfrac{n-1}{2} \rfloor \big\}$. Then the pair $(\mu, \lambda)$ with $\mu = (n-k, k)$ is in $\Omega_{n,2}$ if and only if $k \in \{0, \ldots, r\}$.
\end{lemma}
\begin{proof}
    Let $\mu = (n-k, k)$ for some $k \in \{0, \ldots, \lfloor \tfrac{n}{2} \rfloor\}$. Since $\lambda$ is neither $\0$ nor $\{1,1\}$, the two exceptions of \cref{rem:okada_set_qubit}, we know that $(\mu, \lambda) \in \Omega_{n,2}$ if and only if
    \begin{equation}
        n-2r \leq n-2k,
    \end{equation}
    that is, if $r \geq k$.
\end{proof}

Let $\mu = (n-k, k)$ with $k \in \{0, \ldots, \lfloor \tfrac{n}{2} \rfloor\}$, then the function
\begin{equation*}
    k \longmapsto a \big( (n-k,k) \big) = \frac{k(n+1-k)}{n(n-1)},
\end{equation*}
is a positive increasing function on the interval $[0, \tfrac{n+1}{2}]$.

Let $\lambda = (n-2r)$ for some $r \in \big\{ 0, \ldots, \lfloor \tfrac{n-1}{2} \rfloor \big\}$, and $\mu = (n-r, r)$. Then the function
\begin{equation*}
    p \longmapsto g(p,\lambda) - a(\mu) = \frac{r (n-r-1)}{n(n-1)} - p,
\end{equation*}
is positive for all $p \leq \frac{r (n-r-1)}{n(n-1)}$, i.e. the affine function $f_{\mu,\lambda}$ has positive slope. In particular, the affine function $f_{(n),(n)}$ has always a non-positive slope for $p \geq 0$.

\begin{figure}
    \centering
    \begin{minipage}{.45\textwidth}
        \centering
        \begin{tikzpicture}
            \newcommand\xtilde{1}
            \newcommand\xstar{-2}
            \newcommand\fxstar{3/10}
            
            \newcommand{\horLineFromPointRight}[1]{
                \draw[dashed,opacity=0.4] (#1) -- (#1-|{rel axis cs:1,0})
            }
            \newcommand{\horLineFromPointLeft}[1]{
              \draw[dashed,opacity=0.4] (#1) -- (#1-|{rel axis cs:0,0})
            }
            
            \begin{axis}[
                width=0.9\textwidth,
                axis lines* = box,
                title = {$f_{\mu,\lambda}(x)$},
                ytick=\empty,
                xlabel = \empty,
                ylabel = \empty,
                ytick pos = right,
                ylabel near ticks, 
                yticklabel pos=right,
                xmin=-0.3+\xstar, xmax=0.3+\xtilde,
                ymin=-0.25, ymax=0.8,
                xtick={\xtilde,0}, 
                xticklabels = {$\tilde{x}$,$0$},
                scaled x ticks = false,
                extra y tick style={major y tick style={draw=none},grid=none},
                extra y ticks={-3/20,0,1/4,9/20},
                extra y tick labels={${g(p,\lambda_4)}$,$0$,${g(p,\lambda_3)}$,${\displaystyle g(p,\lambda_1) \atop \displaystyle g(p,\lambda_2)}$}
            ]
            \end{axis}

            \begin{axis}[
                width=0.9\textwidth,
                axis lines* = box,
                ytick=\empty,
                xlabel = \empty,
                ylabel = \empty,
                ytick pos = left,
                ylabel near ticks, 
                yticklabel pos=left,
                xmin=-0.3+\xstar, xmax=0.3+\xtilde,
                ymin=-0.25, ymax=0.8,
                xtick={\xtilde,0}, 
                xticklabels = {$\tilde{x}$,$0$},
                scaled x ticks = false,
                extra y tick style={major y tick style={draw=none},grid=none},
                extra y ticks={0,1/4,2/5},
                extra y tick labels={$a(\mu_3)$,$a(\mu_2)$,$a(\mu_1)$}
            ]

            % f_{\mu,\lambda}(x)
            \foreach \f in {(-3*x)/20, x/4, (9*x)/20, 1/4, 1/4 + x/5, (8 + x)/20} {
                \addplot[color=red, draw opacity=0.3] {\f};
            }
            
            %vertical lines at 0 and \tilde{x}
            \addplot[samples=100] coordinates {(0,-0.25)(0,0.8)};
            \addplot[samples=100] coordinates {(\xtilde,-0.25)(\xtilde,0.8)};

            %horizontal lines at 0
            \addplot[samples=100] coordinates {(-0.3+\xstar,0)(0.3+\xtilde,0)};
            
            %projection of mu's to the left Y axis
            \pgfplotsinvokeforeach{0,1/4,2/5}{
                \addplot[mark=*,mark size=1pt] coordinates {(0,#1)};
                \horLineFromPointLeft{0,#1};
            }

            %projection of lambda's to the right Y axis
            \pgfplotsinvokeforeach{-3/20,1/4,9/20}{
                \addplot[mark=*,mark size=1pt] coordinates {(\xtilde,#1)};
                \horLineFromPointRight{\xtilde,#1};
            }
            
            %intersection of two lines for optimal point
            \addplot[samples=100, smooth, color=blue]{(8+x)/20};
            \addplot[samples=100, smooth, color=blue]{-3*x/20};
            
            %intersection highlight
            \addplot[mark=*,mark size=1pt,blue] coordinates {(\xstar,\fxstar)};
            \draw[ultra thin, shorten <= 2pt, opacity=0.5, text opacity=1, dashed] (\xstar,\fxstar) -- (\xstar+0.75,\fxstar+0.2) node[above] {$\big( x^*,q(p, n, 2) \big)$};
            
            %function f(x)
            \addplot [samples=100, color = red, thick] {max((-3*x)/20, (9*x)/20, (8 + x)/20)};
            \draw[ultra thin, shorten <= 2pt, opacity=0.5, text opacity=1, dashed] (0.2,0.41) -- (0.5,0.55) node[above] {$f(x)$};
            \end{axis}
        \end{tikzpicture}
    \end{minipage}%
    \hfill%
    \begin{minipage}{.45\textwidth}
        \centering
        \begin{tikzpicture}
            \newcommand\xtilde{1}
            \newcommand\xstar{1}
            \newcommand\fxstar{4/15}
            
            \newcommand{\horLineFromPointRight}[1]{
                \draw[dashed,opacity=0.4] (#1) -- (#1-|{rel axis cs:1,0})
            }
            \newcommand{\horLineFromPointLeft}[1]{
              \draw[dashed,opacity=0.4] (#1) -- (#1-|{rel axis cs:0,0})
            }
            
            \begin{axis}[
                width=0.9\textwidth,
                axis lines* = box,
                title = {$f_{\mu,\lambda}(x)$},
                ytick=\empty,
                xlabel = \empty,
                ylabel = \empty,
                ytick pos = right,
                ylabel near ticks, 
                yticklabel pos=right,
                xmin=-0.3+0, xmax=0.3+\xtilde,
                ymin=-0.45, ymax=0.75,
                xtick={\xtilde,0}, 
                xticklabels = {$\tilde{x}$,$0$},
                scaled x ticks = false,
                extra y tick style={major y tick style={draw=none},grid=none},
                extra y ticks={-1/3,-0.03,1/15,4/15},
                extra y tick labels={${g(p,\lambda_4)}$,$0$,${g(p,\lambda_3)}$,${\displaystyle g(p,\lambda_1) \atop \displaystyle g(p,\lambda_2)}$}
            ]
            \end{axis}

            \begin{axis}[
                width=0.9\textwidth,
                axis lines* = box,
                ytick=\empty,
                xlabel = \empty,
                ylabel = \empty,
                ytick pos = left,
                ylabel near ticks, 
                yticklabel pos=left,
                xmin=-0.3+0, xmax=0.3+\xtilde,
                ymin=-0.45, ymax=0.75,
                xtick={\xtilde,0}, 
                xticklabels = {$\tilde{x}$,$0$},
                scaled x ticks = false,
                extra y tick style={major y tick style={draw=none},grid=none},
                extra y ticks={0,1/4,2/5},
                extra y tick labels={$a(\mu_3)$,$a(\mu_2)$,$a(\mu_1)$}
            ]

            % f_{\mu,\lambda}(x)
            \foreach \f in {-1/3*x, x/15, (4*x)/15, 1/4 - (11*x)/60, (15 + x)/60, (-2*(-3 + x))/15} {
                \addplot[color=red, draw opacity=0.3] {\f};
            }
            
            %vertical lines at 0 and \tilde{x}
            \addplot[samples=100] coordinates {(0,-0.45)(0,0.75)};
            \addplot[samples=100] coordinates {(\xtilde,-0.45)(\xtilde,0.75)};

            %horizontal lines at 0
            \addplot[samples=100] coordinates {(-0.3+0,0)(0.3+\xtilde,0)};
            
            %projection of mu's to the left Y axis
            \pgfplotsinvokeforeach{0,1/4,2/5}{
                \addplot[mark=*,mark size=1pt] coordinates {(0,#1)};
                \horLineFromPointLeft{0,#1};
            }

            %projection of lambda's to the right Y axis
            \pgfplotsinvokeforeach{-1/3,1/15,4/15}{
                \addplot[mark=*,mark size=1pt] coordinates {(\xtilde,#1)};
                \horLineFromPointRight{\xtilde,#1};
            }
            
            %intersection of two lines for optimal point
            \addplot[samples=100, smooth, color=blue]{-(2/15)*(-3+x)};
            \addplot[samples=100, smooth, color=blue]{(4*x)/15};
            
            %intersection highlight
            \addplot[mark=*,mark size=1pt,blue] coordinates {(\xstar,\fxstar)};
            \draw[ultra thin, shorten <= 2pt, opacity=0.5, text opacity=1, dashed] (\xstar,\fxstar) -- (\xstar-0.4,\fxstar+0.2) node[above] {$\big( x^*,q(p, n, 2) \big)$};
            
            %function f(x)
            \addplot [samples=100, color = red, thick] {max((-2*(-3 + x))/15, -1/3*x, (4*x)/15)};
            \draw[ultra thin, shorten <= 2pt, opacity=0.5, text opacity=1, dashed] (-0.2,0.426667) -- (-0.15,0.526667) node[above] {$f(x)$};
            \end{axis}
        \end{tikzpicture}
    \end{minipage}
    \caption{Two typical behaviors of the spectrum $f_{\mu,\lambda}(x)$ (thin red lines; $f(x)$ is in bold red) for all $(\lambda,\mu) \in { \Omega_{n,2}}$, for small $p$ on the left and large $p$ on the right. The coordinate $x=x^*$ corresponds to the optimal value $f(x^*) = q(p,n,2)$. The plot corresponds to the parameters $n=5$ with $p=\tfrac{3}{20}$ on the left and $p=\tfrac{1}{3}$ on the right. The partitions $\lambda$ corresponding to the points $(\tilde{x},g(p,\lambda))$ are $\lambda_1 = \lambda_2 = (1)$, $\lambda_3 = (3)$, $\lambda_4 = (5)$. The partitions $\mu$ characterising the offsets $a(\mu)$ for the functions $f_{\mu,\lambda}(x)$ are $\mu_1 = (3,2)$, $\mu_2 = (4,1)$, $\mu_3 = (3,2)$, $\mu_4 = (5)$.}
    \label{fig:BrauerFocalPoints}
\end{figure}
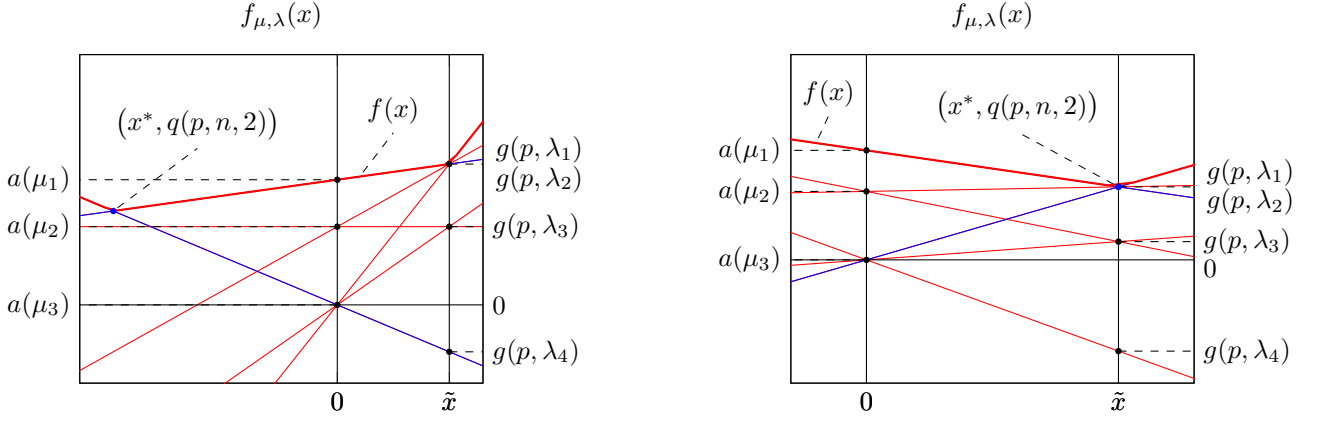

Now we are ready to state that, the $K_n$-extendibility for qubit Brauer states is given by the polytope in \cref{fig:BrauerQubitPolytope}. Its boundary is described in the following theorem:
\begin{theorem} \label{thm:d=2_brauer_region}
    For all $n$, and all $p \in [0,1]$ the optimal value $q(p, n, 2)$ is equal to
    \begin{equation}
        q(p, n, 2) =
        \begin{cases}
            \frac{\ceil{n/2} + 1}{\ceil{n/2} - 1} p &\text{if $p < \frac{\ceil{n/2}-1}{2(2\ceil{n/2}-1)}$}, \\
            \frac{\ceil{n/2}}{2\ceil{n/2}-1} - p &\text{if $\frac{\ceil{n/2}}{2\ceil{n/2}-1} \geq p \geq \frac{\ceil{n/2}-1}{2(2\ceil{n/2}-1)}$}, \\
            0 &\text{otherwise}.
        \end{cases}
    \end{equation}
\end{theorem}
\begin{proof}
    From \cref{lem:lambdasBrauer}, we know that for all $\mu \in \Irr{\cS^d_n}$, either $(\lambda_1, \mu)$ or $(\lambda_2, \mu)$ is in $\Omega_{n,2}$, and the function $\lambda \mapsto g(p, \lambda)$ is maximized for those two $\lambda$'s. Since $a(\mu)$ is maximized for $\mu_1 \coloneqq \big( n - \lfloor \tfrac{n}{2} \rfloor, \lfloor \tfrac{n}{2} \rfloor \big)$, we know that the optimal value $q(p, n, 2)$ lies at the intersection of the affine function $f_{\mu_1,\lambda_1}$ or $f_{\mu_1,\lambda_2}$.

    If $n$ is even, the affine functions $f_{\mu_1,\lambda_1}$ and $f_{\mu_1,\lambda_2}$ have slopes equal to
    \begin{equation}
        g \big( p, (1,1) \big) - a \big( \tfrac{n}{2} , \tfrac{n}{2} \big) = g \big( p, \0 \big) - a \big( \tfrac{n}{2} , \tfrac{n}{2} \big) = \frac{n-2}{4 (n-1)} - p.
    \end{equation}
    Not that in this case, only one of the two affine functions $f_{\mu_1,\lambda_1}$ and $f_{\mu_1,\lambda_2}$ are included in the optimization problem $q(p, n, 2)$. When $n=2$ and $p=0$, both $f_{\mu_1,\lambda_1}$ and $f_{\mu_1,\lambda_2}$ are constant functions and optimal value $q(p, n, 2)$ is equals to $f_{\mu_1,\lambda_1}(0) = f_{\mu_1,\lambda_2}(0) = 1$, as expected.

    Let $n$ even and $p \leq \frac{n-2}{4 (n-1)}$, then $f_{\mu_1,\lambda_1}$ and $f_{\mu_1,\lambda_2}$ have non-negative slope, and the optimal value $q(p, n, 2)$ must lie at the intersection of another affine function $f_{\mu,\lambda}$ with non-positive slope. From \cref{lem:musBrauer} we know that for all $r \in \big\{ 0, \ldots, \lfloor \tfrac{n-1}{2} \rfloor \big\}$ and all $k \in \{0, \ldots, r\}$, the pair $(\lambda_r, \mu_k)$ with $\lambda_r = (n-2r)$ and $\mu_k = (n-k, k)$ is in $\Omega_{n,2}$. Thus if $f_{\mu_k,\lambda_r}$ has non-positive slope for some $r$ and $k$, then given $k' \geq k$, the affine function $f_{\mu_{k'},\lambda_r}$ has also non-positive slope. Moreover the intersection between $f_{\mu_{k'},\lambda_r}$ and $f_{\mu_1,\lambda_1}$ or $f_{\mu_1,\lambda_2}$ is always higher than the intersection between $f_{\mu_k,\lambda_r}$ and $f_{\mu_1,\lambda_1}$ or $f_{\mu_1,\lambda_2}$. Thus we can restrict the intersection problem to affine functions $f_{\mu_r,\lambda_r}$ with $r \in \big\{ 0, \ldots, \lfloor \tfrac{n-1}{2} \rfloor \big\}$. The affine functions $f_{\mu_r,\lambda_r}$ intersect $f_{\mu_1,\lambda_1}$ or $f_{\mu_1,\lambda_2}$ at
    \begin{equation}
        x_r = \frac{4}{2r - n + 2} - 1 \quad \text{ and } \quad y_r = \frac{4p(1 - n) + n - 2}{(n-1)(2r - n + 2)} + \frac{1}{n-1}+p.
    \end{equation}
    The largest value of $y_r$ is reached at $r=0$ and is equal to $y_0 = \frac{p(n+2)}{n-2}$. Not that since $g \big( p, (1,1) \big) - a \big( \tfrac{n}{2} , \tfrac{n}{2} \big)$ and $g \big( p, \0 \big) - a \big( \tfrac{n}{2} , \tfrac{n}{2} \big)$ are both non-negative, and from \cref{lem:lambdasBrauer}:
    \begin{equation*}
        g(p, \lambda_1) = g(p, \lambda_2) \geq g(p, \lambda),
    \end{equation*}
    for all $\lambda$ in $\Omega_{n,2}$, then there is no $\mu$ in $\Omega_{n,2}$ such that $f_{\mu,\lambda_1}$ or $f_{\mu,\lambda_2}$ has non-positive slope. Thus if $n$ even and $p \leq \frac{n-2}{4 (n-1)}$, then $q(p, n, 2) = \frac{p(n+2)}{n-2}$.

    Let $n$ even and $p > \frac{n-2}{4 (n-1)}$, then $f_{\mu_1,\lambda_1}$ and $f_{\mu_1,\lambda_2}$ have negative slope, and the optimal value $q(p, n, 2)$ must lie at the intersection of another affine function $f_{\mu,\lambda}$ with non-negative slope. But from \cref{lem:lambdasBrauer}:
    \begin{equation*}
        g(p, \lambda_1) = g(p, \lambda_2) \geq g(p, \lambda),
    \end{equation*} for all $\lambda$ in $\Omega_{n,2}$. Thus no intersection can occurs at a point higher than $g(p, \lambda_1) = g(p, \lambda_2)$. But from \cref{lem:lambdasBrauer}, we know that for any $\mu_k = (n-k,k)$ for some $k \in \{0, \ldots, \lfloor \tfrac{n}{2} \rfloor\}$, either $(\lambda_1, \mu_k)$ or $(\lambda_2, \mu_k)$ is in $\Omega_{n,2}$. In particular the affine functions $f_{\mu_0,\lambda_1}$ or $f_{\mu_0,\lambda_2}$ have non-negative slope when $g(p, \lambda_1) = g(p, \lambda_2) \geq 0$, that is $p \leq \frac{n}{2(n-1)}$,  and intersect $f_{\mu_1,\lambda_1}$ or $f_{\mu_1,\lambda_2}$ identically at
    \begin{equation}
        x_0 = 1 \quad \text{ and } \quad y_0 = \frac{n}{2 (n-1)} - p.
    \end{equation}
    Thus if $n$ even and $\frac{n}{2(n-1)} \geq p > \frac{n-2}{4 (n-1)}$, then $q(p, n, 2) = \frac{n}{2 (n-1)} - p$.
    
    If $n$ is odd, both $f_{\mu_1,\lambda_1}$ and $f_{\mu_1,\lambda_2}$ are included in the optimization problem $q(p, n, 2)$, and the two affine functions $f_{\mu_1,\lambda_1}$, $f_{\mu_1,\lambda_2}$ coincide with slopes equal to
    \begin{equation}
        g \big( p, (1) \big) - a \big( \lceil \tfrac{n}{2} \rceil, \lfloor \tfrac{n}{2} \rfloor \big) = \frac{n-1}{4n} - p.
    \end{equation}

     Let $n$ be odd and $p \leq \frac{n-1}{4n}$, then $f_{\mu_1,\lambda_1}$ and $f_{\mu_1,\lambda_2}$ have a non-negative slope, and the optimal value $q(p, n, 2)$ must be at the intersection of another affine function $f_{\mu,\lambda}$ with a non-positive slope. Following the same argument as for even $n$: the affine functions $f_{\mu_r,\lambda_r}$ intersect $f_{\mu_1,\lambda_1}$ or $f_{\mu_1,\lambda_2}$ at
    \begin{equation}
        x_r = \frac{4}{2r - n + 1} - 1 \quad \text{ and } \quad y_r = \frac{n-4pn-1}{n (2r - n + 1)}+\frac{1}{n} + p.
    \end{equation}
    The largest value of $y_r$ is reached at $r=0$, is equal to $y_0 = \frac{p(n+3)}{n-1}$, and no other intersection occurs higher. Thus if $n$ odd and $p \leq \frac{n-1}{4n}$, then $q(p, n, 2) = \frac{p(n+3)}{n-1}$.

    Let $n$ odd and $p > \frac{n-1}{4n}$, then $f_{\mu_1,\lambda_1}$ and $f_{\mu_1,\lambda_2}$ have negative slope, and the optimal value $q(p, n, 2)$ must lie at the intersection of another affine function $f_{\mu,\lambda}$ with non-negative slope. Following the same argument as for even $n$: the affine functions $f_{\mu_0,\lambda_1}$ or $f_{\mu_0,\lambda_2}$ have non-negative slope when $g(p, \lambda_1) = g(p, \lambda_2) \geq 0$, that is $p \leq \frac{n+1}{2 n}$, and intersect $f_{\mu_1,\lambda_1}$ or $f_{\mu_1,\lambda_2}$ identically at
    \begin{equation}
        x_0 = 1 \quad \text{ and } \quad y_0 = \frac{1}{2} \of*{\frac{1}{n} - 2p + 1}.
    \end{equation}
    Thus if $n$ odd and $\frac{n+1}{2 n} \geq p > \frac{n-1}{4n}$, then $q(p, n, 2) = \frac{1}{2} \of*{\frac{1}{n} - 2p + 1} = \frac{n+1}{2n} - p$. Collecting everything together we get the following answer:
    \begin{equation}
        q(p, n, 2) =
        \begin{cases}
            \frac{n+2}{n-2}p &\text{if $n$ even and $p < \frac{n-2}{4(n-1)}$}, \\
            \frac{n}{2(n-1)} - p &\text{if $n$ even and $\frac{n}{2(n-1)} \geq p \geq \frac{n-2}{4(n-1)}$}, \\
            \frac{n+3}{n-1}p &\text{if $n$ odd and $p < \frac{n-1}{4n}$}, \\
            \frac{n+1}{2n} - p &\text{if $n$ odd and $\frac{n+1}{2 n} \geq p \geq \frac{n-1}{4n}$}, \\
            0 &\text{otherwise},
        \end{cases}
        =
        \begin{cases}
            \frac{\ceil{n/2} + 1}{\ceil{n/2} - 1} p &\text{if $p < \frac{\ceil{n/2}-1}{2(2\ceil{n/2}-1)}$}, \\
            \frac{\ceil{n/2}}{2\ceil{n/2}-1} - p &\text{if $\frac{\ceil{n/2}}{2\ceil{n/2}-1} \geq p \geq \frac{\ceil{n/2}-1}{2(2\ceil{n/2}-1)}$}, \\
            0 &\text{otherwise}.
        \end{cases}
    \end{equation}
\end{proof}

%%%%%%%%%%%%%%%%%%%%%%%%%%%%%%%%%%%%%%%%%%%%%%%%%%%%%%%%%%%%%%%%%%%%%%%%%%%%%%%%%%%%%%%%%%%%%%%%%%%%%%%%%%%%%%%%%%%%%%%%%%%%%%%%%%%%%%%%%%%%%%%%%%%%%%%%%%%%%%%%%%%%%%%%%%%%%%%%%%%%%%%%%%%%%%%%%%%%%%%%%%%%%%%%%%%%%%%%%%
%%%%%%%%%%%%%%%%%%%%%%%%%%%%%%%%%%%%%%%%%%%%%%%%%%%%%%%%%%%%%%%%%%%%%%%%%%%%%%%%%%%%%%%%%%%%%%%%%%%%%%%%%%%%%%%%%%%%%%%%%%%%%%%%%%%%%%%%%%%%%%%%%%%%%%%%%%%%%%%%%%%%%%%%%%%%%%%%%%%%%%%%%%%%%%%%%%%%%%%%%%%%%%%%%%%%%%%%%%
\subsection{Discussion}

%%%%%%%%%%%%%%%%%%%%%%%%%%%%%%%%%%%%%%%%%%%%%%%%%%%%%%%%%%%%%%%%%%%%%%%%%%%%%%%%%%%%%%%%%%%%%%%%%%%%%%%%%%%%%%%%%%%%%%%%%%%%%%%%%%%%%%%%%%%%%%%%%%%%%%%%%%%%%%%%%%%%%%%%%%%%%%%%%%%%%%%%%%%%%%%%%%%%%%%%%%%%%%%%%%%%%%%%%%
\subsubsection{Asymptotic limits}

In this section, we analyze the asymptotic behavior of $q_W(n,d)$, $p_B(n,d)$, and $p_I(n,d)$ as the dimension $d$ becomes large. Deriving from \cref{thm:WernerStates}, the asymptotic behavior of the Werner case $q_W(n,d)$ is:
\begin{align}
    \lim_{d \to \infty} q_W(n,d) &= \lim_{d \to \infty} \of*{\frac{d-1}{2d} \cdot \frac{(n+k+d)(n-k)}{n(n-1)} + \frac{k(k-1)}{n(n-1)}} = 1,
\end{align}
with $k = n \bmod d$. Correspondingly, for the Brauer case $p_B(n,d)$, \cref{thm:BrauerStates} yields:
\begin{align}
    \lim_{d \to \infty} p_B(n,d) &= \lim_{d \to \infty} \of*{\frac{1}{d} + \frac{1}{n + n \bmod 2 - 1} - \frac{1}{d(n + n \bmod 2 - 1)}} = \frac{1}{n + n \bmod 2 - 1}.
\end{align}
Lastly, the isotropic case $p'_I(n,d)$, from \cref{thm:isotropicStates}, takes the form::
\begin{align}
    \lim_{d \to \infty} p_I(n,d) = \lim_{d \to \infty} p'_I(n,d) &= \lim_{d \to \infty}
    \begin{cases}
        \frac{1}{n + n \bmod 2 - 1} &\text{ if $d > n$ or either $d$ or $n$ is even} \\
        \min \big\{ \frac{2 d + 1}{2 d n + 1}, \frac{1}{n - 1} \big\} &\text{ if $n \geq d$ and both $d$ and $n$ are odd}
    \end{cases} \nonumber \\
    &= \frac{1}{n + n \bmod 2 - 1}.
\end{align}
It can be observed that in large dimensions, $p_B(n,d)$ and $p_I(n,d)$ become equal.

As the number of vertices $n$ grows, $p'_I(n,d)$ converges to zero. This implies that an isotropic state $\rho$ is $K_n$-extendible for all $n$ if and only if $\rho = \frac{\I}{d^2}$. For Werner and Brauer states, the limits $\lim_{n \to \infty} q_W(n,d) = \frac{d-1}{2d}$ and $\lim_{n \to \infty} p_B(n,d) = \frac{1}{d}$ hold, which makes non-trivial $K_n$-extendible Werner and Brauer states possible.

Moreover, due to \cref{app:BrauerStates}, it is interesting to note that for a fixed $d$ and for every $n \geq 2$, there exists entangled Brauer state which is $K_n$-extendible. However, this is not true neither for Werner nor for isotropic states.

We could not solve a general Brauer $K_n$-extendibility region problem. However, we are tempted to formulate the following conjecture:
\begin{conjecture} \label{conj}
     Brauer $K_n$-extendibility region is a polytope for all $n$ and $d$. However, in the limit $n \to \infty$ it is not a polytope.
\end{conjecture}
For example, one could see how the limiting shape of the Brauer $K_n$-extendibility region for $d = 3$ looks numerically in \cref{fig:BrauerNonPolytope}. Related to that example, we expect that results and methods developed in \cite{Ryan_2022,jakabBilinearbiquadraticModelComplete2018,jakabInterplayUnitaryPermutation2022} could be helpful to tackle specifically the $d=3$ case analytically for arbitrary $n$.

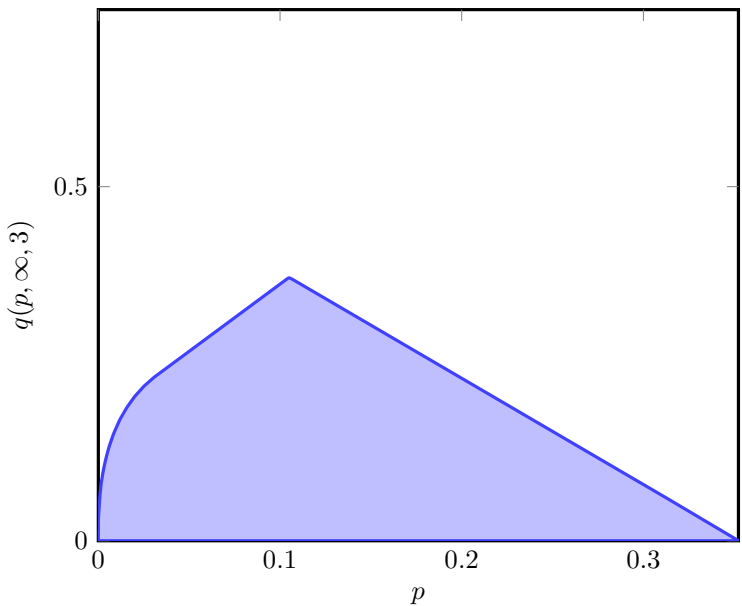
\begin{figure}
    \centering
    \begin{tikzpicture}
        \begin{axis}[
            xlabel={$p$},
            ylabel={$q(p, \infty, 3)$},
            xtick=\empty,
            ytick=\empty,
            extra x ticks={0,0.1,0.2,0.3},
            extra x tick labels={$0$,$0.1$,$0.2$,$0.3$},
            extra y ticks={0,0.5},
            extra y tick labels={$0$,$0.5$},
            xmin=0, xmax=0.3523,
            ymin=0, ymax=0.75,
            legend pos=north east,
            legend cell align=left,
            no marks,
            samples=1,
            very thick
        ]

        \addplot[smooth, draw=blue!75!white, fill=blue!25!white] table [col sep=comma, x index=0, y index=1] {nonPolytope.csv} \closedcycle ;
        
        \end{axis}
    \end{tikzpicture}
    \caption{Asymptotic $K_n$-extendibility for qutrit Brauer states is not a polytope.}
    \label{fig:BrauerNonPolytope}
\end{figure}

%%%%%%%%%%%%%%%%%%%%%%%%%%%%%%%%%%%%%%%%%%%%%%%%%%%%%%%%%%%%%%%%%%%%%%%%%%%%%%%%%%%%%%%%%%%%%%%%%%%%%%%%%%%%%%%%%%%%%%%%%%%%%%%%%%%%%%%%%%%%%%%%%%%%%%%%%%%%%%%%%%%%%%%%%%%%%%%%%%%%%%%%%%%%%%%%%%%%%%%%%%%%%%%%%%%%%%%%%%
\subsubsection{Optimal states}

The optimal state for the Werner case which achieves $q_W(n,d)$ can be easily obtained from the proof in \cref{app:werner_primal}. Namely, an optimal state $\rho$ is the normalized projector onto the isotypic component $\lambda$:
\begin{equation}
    \rho = \frac{\varepsilon_\lambda}{\Tr \varepsilon_\lambda},
\end{equation}
such that
\begin{equation}
    \lambda_1 = \dotsc = \lambda_k = \frac{n-k}{d}+1, \quad \lambda_{k+1} = \dotsc = \lambda_d = \frac{n-k}{d},
\end{equation}
where $k \defeq n \bmod d$.

In general, it is not known which quantum state $\rho$ gives the optimal value $p_B(n,d),$ or $p'_I(n,d)$, but using \cref{sec:Schur_Weyl_duality_orthogonal}, it must be in the algebra generated by the action of the Brauer algebra into the tensor space ${\big{(} \C^d \big{)}}^{\otimes n}$.

For example when $n=3$ and $d=3$ an optimal quantum state $\rho$ for the Brauer problem $p_B(3,3)$ is:
\begin{align}
    \rho = \frac{1}{45} \bigg[
    &\begin{tikzpicture}[baseline = (baseline.center),
                        site/.style = {circle,
                                       draw = white,
                                       outer sep = 0.5pt,
                                       fill = black!80!white,
                                       inner sep = 1pt}]
        \node (m1) {};
        \node (m2) [yshift = -0.5em] at (m1.center) {};
        \node (m3) [yshift = -0.5em] at (m2.center) {};
% do not remove
        \node (baseline) [yshift = -0.25em] at (m2.center) {};
% do not remove
        \node[site] (l1) [xshift = -0.5em] at (m1.center) {};
        \node[site] (r1) [xshift = 0.5em] at (m1.center) {};
% do not remove
        \node[site] (l2) [yshift = -0.5em] at (l1.center) {};
        \node[site] (r2) [yshift = -0.5em] at (r1.center) {};
% do not remove
        \node[site] (l3) [yshift = -0.5em] at (l2.center) {};
        \node[site] (r3) [yshift = -0.5em] at (r2.center) {};
% do not remove
        \draw[-, line width = 2.5pt, draw = white] (r1) .. controls (m1) and (m2) .. (r2);
        \draw[-, line width = 1pt, draw = black] (r1) .. controls (m1) and (m2) .. (r2);
% do not remove
        \draw[-, line width = 2.5pt, draw = white] (l1) .. controls (m1) and (m2) .. (l2);
        \draw[-, line width = 1pt, draw = black] (l1) .. controls (m1) and (m2) .. (l2);
% do not remove
        \draw[-, line width = 2.5pt, draw = white] (l3) to (r3);
        \draw[-, line width = 1pt, draw = black] (l3) to (r3);
    \end{tikzpicture}
    +
    \begin{tikzpicture}[baseline = (baseline.center),
                        site/.style = {circle,
                                       draw = white,
                                       outer sep = 0.5pt,
                                       fill = black!80!white,
                                       inner sep = 1pt}]
        \node (m1) {};
        \node (m2) [yshift = -0.5em] at (m1.center) {};
        \node (m3) [yshift = -0.5em] at (m2.center) {};
% do not remove
        \node (baseline) [yshift = -0.25em] at (m2.center) {};
% do not remove
        \node[site] (l1) [xshift = -0.5em] at (m1.center) {};
        \node[site] (r1) [xshift = 0.5em] at (m1.center) {};
% do not remove
        \node[site] (l2) [yshift = -0.5em] at (l1.center) {};
        \node[site] (r2) [yshift = -0.5em] at (r1.center) {};
% do not remove
        \node[site] (l3) [yshift = -0.5em] at (l2.center) {};
        \node[site] (r3) [yshift = -0.5em] at (r2.center) {};
% do not remove
        \draw[-, line width = 2.5pt, draw = white] (r1) .. controls (m1) and (m3) .. (r3);
        \draw[-, line width = 1pt, draw = black] (r1) .. controls (m1) and (m3) .. (r3);
% do not remove
        \draw[-, line width = 2.5pt, draw = white] (l1) .. controls (m1) and (m3) .. (l3);
        \draw[-, line width = 1pt, draw = black]    (l1) .. controls (m1) and (m3) .. (l3);
% do not remove
        \draw[-, line width = 2.5pt, draw = white] (l2) to (r2);
        \draw[-, line width = 1pt, draw = black] (l2) to (r2);
    \end{tikzpicture}
    +
    \begin{tikzpicture}[baseline = (baseline.center),
                        site/.style = {circle,
                                       draw = white,
                                       outer sep = 0.5pt,
                                       fill = black!80!white,
                                       inner sep = 1pt}]
        \node (m1) {};
        \node (m2) [yshift = -0.5em] at (m1.center) {};
        \node (m3) [yshift = -0.5em] at (m2.center) {};
% do not remove
        \node (baseline) [yshift = -0.25em] at (m2.center) {};
% do not remove
        \node[site] (l1) [xshift = -0.5em] at (m1.center) {};
        \node[site] (r1) [xshift = 0.5em] at (m1.center) {};
% do not remove
        \node[site] (l2) [yshift = -0.5em] at (l1.center) {};
        \node[site] (r2) [yshift = -0.5em] at (r1.center) {};
% do not remove
        \node[site] (l3) [yshift = -0.5em] at (l2.center) {};
        \node[site] (r3) [yshift = -0.5em] at (r2.center) {};
% do not remove
        \draw[-, line width = 2.5pt, draw = white] (r2) .. controls (m2) and (m3) .. (r3);
        \draw[-, line width = 1pt, draw = black] (r2) .. controls (m2) and (m3) .. (r3);
% do not remove
        \draw[-, line width = 2.5pt, draw = white] (l2) .. controls (m2) and (m3) .. (l3);
        \draw[-, line width = 1pt, draw = black] (l2) .. controls (m2) and (m3) .. (l3);
% do not remove
        \draw[-, line width = 2.5pt, draw = white] (l1) to (r1);
        \draw[-, line width = 1pt, draw = black] (l1) to (r1);
    \end{tikzpicture}
    +
    \begin{tikzpicture}[baseline = (baseline.center),
                        site/.style = {circle,
                                       draw = white,
                                       outer sep = 0.5pt,
                                       fill = black!80!white,
                                       inner sep = 1pt}]
        \node (m1) {};
        \node (m2) [yshift = -0.5em] at (m1.center) {};
        \node (m3) [yshift = -0.5em] at (m2.center) {};
% do not remove
        \node (baseline) [yshift = -0.25em] at (m2.center) {};
% do not remove
        \node[site] (l1) [xshift = -0.5em] at (m1.center) {};
        \node[site] (r1) [xshift = 0.5em] at (m1.center) {};
% do not remove
        \node[site] (l2) [yshift = -0.5em] at (l1.center) {};
        \node[site] (r2) [yshift = -0.5em] at (r1.center) {};
% do not remove
        \node[site] (l3) [yshift = -0.5em] at (l2.center) {};
        \node[site] (r3) [yshift = -0.5em] at (r2.center) {};
% do not remove
        \draw[-, line width = 2.5pt, draw = white] (r1) .. controls (m1) and (m2) .. (r2);
        \draw[-, line width = 1pt, draw = black] (r1) .. controls (m1) and (m2) .. (r2);
% do not remove
        \draw[-, line width = 2.5pt, draw = white] (l1) .. controls (m1) and (m3) .. (l3);
        \draw[-, line width = 1pt, draw = black] (l1) .. controls (m1) and (m3) .. (l3);
% do not remove
        \draw[-, line width = 2.5pt, draw = white] (l2) to (r3);
        \draw[-, line width = 1pt, draw = black] (l2) to (r3);
    \end{tikzpicture}
    +
    \begin{tikzpicture}[baseline = (baseline.center),
                        site/.style = {circle,
                                       draw = white,
                                       outer sep = 0.5pt,
                                       fill = black!80!white,
                                       inner sep = 1pt}]
        \node (m1) {};
        \node (m2) [yshift = -0.5em] at (m1.center) {};
        \node (m3) [yshift = -0.5em] at (m2.center) {};
% do not remove
        \node (baseline) [yshift = -0.25em] at (m2.center) {};
% do not remove
        \node[site] (l1) [xshift = -0.5em] at (m1.center) {};
        \node[site] (r1) [xshift = 0.5em] at (m1.center) {};
% do not remove
        \node[site] (l2) [yshift = -0.5em] at (l1.center) {};
        \node[site] (r2) [yshift = -0.5em] at (r1.center) {};
% do not remove
        \node[site] (l3) [yshift = -0.5em] at (l2.center) {};
        \node[site] (r3) [yshift = -0.5em] at (r2.center) {};
% do not remove
        \draw[-, line width = 2.5pt, draw = white] (r1) .. controls (m1) and (m3) .. (r3);
        \draw[-, line width = 1pt, draw = black] (r1) .. controls (m1) and (m3) .. (r3);
% do not remove
        \draw[-, line width = 2.5pt, draw = white] (l1) .. controls (m1) and (m2) .. (l2);
        \draw[-, line width = 1pt, draw = black] (l1) .. controls (m1) and (m2) .. (l2);
% do not remove
        \draw[-, line width = 2.5pt, draw = white] (l3) to (r2);
        \draw[-, line width = 1pt, draw = black] (l3) to (r2);
    \end{tikzpicture}
    +
    \begin{tikzpicture}[baseline = (baseline.center),
                        site/.style = {circle,
                                       draw = white,
                                       outer sep = 0.5pt,
                                       fill = black!80!white,
                                       inner sep = 1pt}]
        \node (m1) {};
        \node (m2) [yshift = -0.5em] at (m1.center) {};
        \node (m3) [yshift = -0.5em] at (m2.center) {};
% do not remove
        \node (baseline) [yshift = -0.25em] at (m2.center) {};
% do not remove
        \node[site] (l1) [xshift = -0.5em] at (m1.center) {};
        \node[site] (r1) [xshift = 0.5em] at (m1.center) {};
% do not remove
        \node[site] (l2) [yshift = -0.5em] at (l1.center) {};
        \node[site] (r2) [yshift = -0.5em] at (r1.center) {};
% do not remove
        \node[site] (l3) [yshift = -0.5em] at (l2.center) {};
        \node[site] (r3) [yshift = -0.5em] at (r2.center) {};
% do not remove
        \draw[-, line width = 2.5pt, draw = white] (r2) .. controls (m2) and (m3) .. (r3);
        \draw[-, line width = 1pt, draw = black] (r2) .. controls (m2) and (m3) .. (r3);
% do not remove
        \draw[-, line width = 2.5pt, draw = white] (l1) .. controls (m1) and (m3) .. (l3);
        \draw[-, line width = 1pt, draw = black] (l1) .. controls (m1) and (m3) .. (l3);
% do not remove
        \draw[-, line width = 2.5pt, draw = white] (l2) to (r1);
        \draw[-, line width = 1pt, draw = black] (l2) to (r1);
    \end{tikzpicture}
    +
    \begin{tikzpicture}[baseline = (baseline.center),
                        site/.style = {circle,
                                       draw = white,
                                       outer sep = 0.5pt,
                                       fill = black!80!white,
                                       inner sep = 1pt}]
        \node (m1) {};
        \node (m2) [yshift = -0.5em] at (m1.center) {};
        \node (m3) [yshift = -0.5em] at (m2.center) {};
% do not remove
        \node (baseline) [yshift = -0.25em] at (m2.center) {};
% do not remove
        \node[site] (l1) [xshift = -0.5em] at (m1.center) {};
        \node[site] (r1) [xshift = 0.5em] at (m1.center) {};
% do not remove
        \node[site] (l2) [yshift = -0.5em] at (l1.center) {};
        \node[site] (r2) [yshift = -0.5em] at (r1.center) {};
% do not remove
        \node[site] (l3) [yshift = -0.5em] at (l2.center) {};
        \node[site] (r3) [yshift = -0.5em] at (r2.center) {};
% do not remove
        \draw[-, line width = 2.5pt, draw = white] (r1) .. controls (m1) and (m3) .. (r3);
        \draw[-, line width = 1pt, draw = black] (r1) .. controls (m1) and (m3) .. (r3);
% do not remove
        \draw[-, line width = 2.5pt, draw = white] (l2) .. controls (m2) and (m3) .. (l3);
        \draw[-, line width = 1pt, draw = black] (l2) .. controls (m2) and (m3) .. (l3);
% do not remove
        \draw[-, line width = 2.5pt, draw = white] (l1) to (r2);
        \draw[-, line width = 1pt, draw = black] (l1) to (r2);
    \end{tikzpicture}
    +
    \begin{tikzpicture}[baseline = (baseline.center),
                        site/.style = {circle,
                                       draw = white,
                                       outer sep = 0.5pt,
                                       fill = black!80!white,
                                       inner sep = 1pt}]
        \node (m1) {};
        \node (m2) [yshift = -0.5em] at (m1.center) {};
        \node (m3) [yshift = -0.5em] at (m2.center) {};
% do not remove
        \node (baseline) [yshift = -0.25em] at (m2.center) {};
% do not remove
        \node[site] (l1) [xshift = -0.5em] at (m1.center) {};
        \node[site] (r1) [xshift = 0.5em] at (m1.center) {};
% do not remove
        \node[site] (l2) [yshift = -0.5em] at (l1.center) {};
        \node[site] (r2) [yshift = -0.5em] at (r1.center) {};
% do not remove
        \node[site] (l3) [yshift = -0.5em] at (l2.center) {};
        \node[site] (r3) [yshift = -0.5em] at (r2.center) {};
% do not remove
        \draw[-, line width = 2.5pt, draw = white] (r2) .. controls (m2) and (m3) .. (r3);
        \draw[-, line width = 1pt, draw = black] (r2) .. controls (m2) and (m3) .. (r3);
% do not remove
        \draw[-, line width = 2.5pt, draw = white] (l1) .. controls (m1) and (m2) .. (l2);
        \draw[-, line width = 1pt, draw = black] (l1) .. controls (m1) and (m2) .. (l2);
% do not remove
        \draw[-, line width = 2.5pt, draw = white] (l3) to (r1);
        \draw[-, line width = 1pt, draw = black] (l3) to (r1);
    \end{tikzpicture}
    +
    \begin{tikzpicture}[baseline = (baseline.center),
                        site/.style = {circle,
                                       draw = white,
                                       outer sep = 0.5pt,
                                       fill = black!80!white,
                                       inner sep = 1pt}]
        \node (m1) {};
        \node (m2) [yshift = -0.5em] at (m1.center) {};
        \node (m3) [yshift = -0.5em] at (m2.center) {};
% do not remove
        \node (baseline) [yshift = -0.25em] at (m2.center) {};
% do not remove
        \node[site] (l1) [xshift = -0.5em] at (m1.center) {};
        \node[site] (r1) [xshift = 0.5em] at (m1.center) {};
% do not remove
        \node[site] (l2) [yshift = -0.5em] at (l1.center) {};
        \node[site] (r2) [yshift = -0.5em] at (r1.center) {};
% do not remove
        \node[site] (l3) [yshift = -0.5em] at (l2.center) {};
        \node[site] (r3) [yshift = -0.5em] at (r2.center) {};
% do not remove
        \draw[-, line width = 2.5pt, draw = white] (r1) .. controls (m1) and (m2) .. (r2);
        \draw[-, line width = 1pt, draw = black] (r1) .. controls (m1) and (m2) .. (r2);
% do not remove
        \draw[-, line width = 2.5pt, draw = white] (l2) .. controls (m2) and (m3) .. (l3);
        \draw[-, line width = 1pt, draw = black] (l2) .. controls (m2) and (m3) .. (l3);
% do not remove
        \draw[-, line width = 2.5pt, draw = white] (l1) to (r3);
        \draw[-, line width = 1pt, draw = black] (l1) to (r3);
    \end{tikzpicture} \bigg].
\end{align}
and an optimal quantum state $\rho$ for the isotropic problem $p'_I(3,3)$ is:
\begin{align}
    \rho = \frac{1}{57} \bigg[
    &\begin{tikzpicture}[baseline = (baseline.center),
                        site/.style = {circle,
                                       draw = white,
                                       outer sep = 0.5pt,
                                       fill = black!80!white,
                                       inner sep = 1pt}]
        \node (m1) {};
        \node (m2) [yshift = -0.5em] at (m1.center) {};
        \node (m3) [yshift = -0.5em] at (m2.center) {};
% do not remove
        \node (baseline) [yshift = -0.25em] at (m2.center) {};
% do not remove
        \node[site] (l1) [xshift = -0.5em] at (m1.center) {};
        \node[site] (r1) [xshift = 0.5em] at (m1.center) {};
% do not remove
        \node[site] (l2) [yshift = -0.5em] at (l1.center) {};
        \node[site] (r2) [yshift = -0.5em] at (r1.center) {};
% do not remove
        \node[site] (l3) [yshift = -0.5em] at (l2.center) {};
        \node[site] (r3) [yshift = -0.5em] at (r2.center) {};
% do not remove
        \draw[-, line width = 2.5pt, draw = white] (r1) .. controls (m1) and (m2) .. (r2);
        \draw[-, line width = 1pt, draw = black] (r1) .. controls (m1) and (m2) .. (r2);
% do not remove
        \draw[-, line width = 2.5pt, draw = white] (l1) .. controls (m1) and (m2) .. (l2);
        \draw[-, line width = 1pt, draw = black] (l1) .. controls (m1) and (m2) .. (l2);
% do not remove
        \draw[-, line width = 2.5pt, draw = white] (l3) to (r3);
        \draw[-, line width = 1pt, draw = black] (l3) to (r3);
    \end{tikzpicture}
    +
    \begin{tikzpicture}[baseline = (baseline.center),
                        site/.style = {circle,
                                       draw = white,
                                       outer sep = 0.5pt,
                                       fill = black!80!white,
                                       inner sep = 1pt}]
        \node (m1) {};
        \node (m2) [yshift = -0.5em] at (m1.center) {};
        \node (m3) [yshift = -0.5em] at (m2.center) {};
% do not remove
        \node (baseline) [yshift = -0.25em] at (m2.center) {};
% do not remove
        \node[site] (l1) [xshift = -0.5em] at (m1.center) {};
        \node[site] (r1) [xshift = 0.5em] at (m1.center) {};
% do not remove
        \node[site] (l2) [yshift = -0.5em] at (l1.center) {};
        \node[site] (r2) [yshift = -0.5em] at (r1.center) {};
% do not remove
        \node[site] (l3) [yshift = -0.5em] at (l2.center) {};
        \node[site] (r3) [yshift = -0.5em] at (r2.center) {};
% do not remove
        \draw[-, line width = 2.5pt, draw = white] (r1) .. controls (m1) and (m3) .. (r3);
        \draw[-, line width = 1pt, draw = black] (r1) .. controls (m1) and (m3) .. (r3);
% do not remove
        \draw[-, line width = 2.5pt, draw = white] (l1) .. controls (m1) and (m3) .. (l3);
        \draw[-, line width = 1pt, draw = black]    (l1) .. controls (m1) and (m3) .. (l3);
% do not remove
        \draw[-, line width = 2.5pt, draw = white] (l2) to (r2);
        \draw[-, line width = 1pt, draw = black] (l2) to (r2);
    \end{tikzpicture}
    +
    \begin{tikzpicture}[baseline = (baseline.center),
                        site/.style = {circle,
                                       draw = white,
                                       outer sep = 0.5pt,
                                       fill = black!80!white,
                                       inner sep = 1pt}]
        \node (m1) {};
        \node (m2) [yshift = -0.5em] at (m1.center) {};
        \node (m3) [yshift = -0.5em] at (m2.center) {};
% do not remove
        \node (baseline) [yshift = -0.25em] at (m2.center) {};
% do not remove
        \node[site] (l1) [xshift = -0.5em] at (m1.center) {};
        \node[site] (r1) [xshift = 0.5em] at (m1.center) {};
% do not remove
        \node[site] (l2) [yshift = -0.5em] at (l1.center) {};
        \node[site] (r2) [yshift = -0.5em] at (r1.center) {};
% do not remove
        \node[site] (l3) [yshift = -0.5em] at (l2.center) {};
        \node[site] (r3) [yshift = -0.5em] at (r2.center) {};
% do not remove
        \draw[-, line width = 2.5pt, draw = white] (r2) .. controls (m2) and (m3) .. (r3);
        \draw[-, line width = 1pt, draw = black] (r2) .. controls (m2) and (m3) .. (r3);
% do not remove
        \draw[-, line width = 2.5pt, draw = white] (l2) .. controls (m2) and (m3) .. (l3);
        \draw[-, line width = 1pt, draw = black] (l2) .. controls (m2) and (m3) .. (l3);
% do not remove
        \draw[-, line width = 2.5pt, draw = white] (l1) to (r1);
        \draw[-, line width = 1pt, draw = black] (l1) to (r1);
    \end{tikzpicture}
    +
    \begin{tikzpicture}[baseline = (baseline.center),
                        site/.style = {circle,
                                       draw = white,
                                       outer sep = 0.5pt,
                                       fill = black!80!white,
                                       inner sep = 1pt}]
        \node (m1) {};
        \node (m2) [yshift = -0.5em] at (m1.center) {};
        \node (m3) [yshift = -0.5em] at (m2.center) {};
% do not remove
        \node (baseline) [yshift = -0.25em] at (m2.center) {};
% do not remove
        \node[site] (l1) [xshift = -0.5em] at (m1.center) {};
        \node[site] (r1) [xshift = 0.5em] at (m1.center) {};
% do not remove
        \node[site] (l2) [yshift = -0.5em] at (l1.center) {};
        \node[site] (r2) [yshift = -0.5em] at (r1.center) {};
% do not remove
        \node[site] (l3) [yshift = -0.5em] at (l2.center) {};
        \node[site] (r3) [yshift = -0.5em] at (r2.center) {};
% do not remove
        \draw[-, line width = 2.5pt, draw = white] (r1) .. controls (m1) and (m2) .. (r2);
        \draw[-, line width = 1pt, draw = black] (r1) .. controls (m1) and (m2) .. (r2);
% do not remove
        \draw[-, line width = 2.5pt, draw = white] (l1) .. controls (m1) and (m3) .. (l3);
        \draw[-, line width = 1pt, draw = black] (l1) .. controls (m1) and (m3) .. (l3);
% do not remove
        \draw[-, line width = 2.5pt, draw = white] (l2) to (r3);
        \draw[-, line width = 1pt, draw = black] (l2) to (r3);
    \end{tikzpicture}
    +
    \begin{tikzpicture}[baseline = (baseline.center),
                        site/.style = {circle,
                                       draw = white,
                                       outer sep = 0.5pt,
                                       fill = black!80!white,
                                       inner sep = 1pt}]
        \node (m1) {};
        \node (m2) [yshift = -0.5em] at (m1.center) {};
        \node (m3) [yshift = -0.5em] at (m2.center) {};
% do not remove
        \node (baseline) [yshift = -0.25em] at (m2.center) {};
% do not remove
        \node[site] (l1) [xshift = -0.5em] at (m1.center) {};
        \node[site] (r1) [xshift = 0.5em] at (m1.center) {};
% do not remove
        \node[site] (l2) [yshift = -0.5em] at (l1.center) {};
        \node[site] (r2) [yshift = -0.5em] at (r1.center) {};
% do not remove
        \node[site] (l3) [yshift = -0.5em] at (l2.center) {};
        \node[site] (r3) [yshift = -0.5em] at (r2.center) {};
% do not remove
        \draw[-, line width = 2.5pt, draw = white] (r1) .. controls (m1) and (m3) .. (r3);
        \draw[-, line width = 1pt, draw = black] (r1) .. controls (m1) and (m3) .. (r3);
% do not remove
        \draw[-, line width = 2.5pt, draw = white] (l1) .. controls (m1) and (m2) .. (l2);
        \draw[-, line width = 1pt, draw = black] (l1) .. controls (m1) and (m2) .. (l2);
% do not remove
        \draw[-, line width = 2.5pt, draw = white] (l3) to (r2);
        \draw[-, line width = 1pt, draw = black] (l3) to (r2);
    \end{tikzpicture}
    +
    \begin{tikzpicture}[baseline = (baseline.center),
                        site/.style = {circle,
                                       draw = white,
                                       outer sep = 0.5pt,
                                       fill = black!80!white,
                                       inner sep = 1pt}]
        \node (m1) {};
        \node (m2) [yshift = -0.5em] at (m1.center) {};
        \node (m3) [yshift = -0.5em] at (m2.center) {};
% do not remove
        \node (baseline) [yshift = -0.25em] at (m2.center) {};
% do not remove
        \node[site] (l1) [xshift = -0.5em] at (m1.center) {};
        \node[site] (r1) [xshift = 0.5em] at (m1.center) {};
% do not remove
        \node[site] (l2) [yshift = -0.5em] at (l1.center) {};
        \node[site] (r2) [yshift = -0.5em] at (r1.center) {};
% do not remove
        \node[site] (l3) [yshift = -0.5em] at (l2.center) {};
        \node[site] (r3) [yshift = -0.5em] at (r2.center) {};
% do not remove
        \draw[-, line width = 2.5pt, draw = white] (r2) .. controls (m2) and (m3) .. (r3);
        \draw[-, line width = 1pt, draw = black] (r2) .. controls (m2) and (m3) .. (r3);
% do not remove
        \draw[-, line width = 2.5pt, draw = white] (l1) .. controls (m1) and (m3) .. (l3);
        \draw[-, line width = 1pt, draw = black] (l1) .. controls (m1) and (m3) .. (l3);
% do not remove
        \draw[-, line width = 2.5pt, draw = white] (l2) to (r1);
        \draw[-, line width = 1pt, draw = black] (l2) to (r1);
    \end{tikzpicture}
    +
    \begin{tikzpicture}[baseline = (baseline.center),
                        site/.style = {circle,
                                       draw = white,
                                       outer sep = 0.5pt,
                                       fill = black!80!white,
                                       inner sep = 1pt}]
        \node (m1) {};
        \node (m2) [yshift = -0.5em] at (m1.center) {};
        \node (m3) [yshift = -0.5em] at (m2.center) {};
% do not remove
        \node (baseline) [yshift = -0.25em] at (m2.center) {};
% do not remove
        \node[site] (l1) [xshift = -0.5em] at (m1.center) {};
        \node[site] (r1) [xshift = 0.5em] at (m1.center) {};
% do not remove
        \node[site] (l2) [yshift = -0.5em] at (l1.center) {};
        \node[site] (r2) [yshift = -0.5em] at (r1.center) {};
% do not remove
        \node[site] (l3) [yshift = -0.5em] at (l2.center) {};
        \node[site] (r3) [yshift = -0.5em] at (r2.center) {};
% do not remove
        \draw[-, line width = 2.5pt, draw = white] (r1) .. controls (m1) and (m3) .. (r3);
        \draw[-, line width = 1pt, draw = black] (r1) .. controls (m1) and (m3) .. (r3);
% do not remove
        \draw[-, line width = 2.5pt, draw = white] (l2) .. controls (m2) and (m3) .. (l3);
        \draw[-, line width = 1pt, draw = black] (l2) .. controls (m2) and (m3) .. (l3);
% do not remove
        \draw[-, line width = 2.5pt, draw = white] (l1) to (r2);
        \draw[-, line width = 1pt, draw = black] (l1) to (r2);
    \end{tikzpicture}
    +
    \begin{tikzpicture}[baseline = (baseline.center),
                        site/.style = {circle,
                                       draw = white,
                                       outer sep = 0.5pt,
                                       fill = black!80!white,
                                       inner sep = 1pt}]
        \node (m1) {};
        \node (m2) [yshift = -0.5em] at (m1.center) {};
        \node (m3) [yshift = -0.5em] at (m2.center) {};
% do not remove
        \node (baseline) [yshift = -0.25em] at (m2.center) {};
% do not remove
        \node[site] (l1) [xshift = -0.5em] at (m1.center) {};
        \node[site] (r1) [xshift = 0.5em] at (m1.center) {};
% do not remove
        \node[site] (l2) [yshift = -0.5em] at (l1.center) {};
        \node[site] (r2) [yshift = -0.5em] at (r1.center) {};
% do not remove
        \node[site] (l3) [yshift = -0.5em] at (l2.center) {};
        \node[site] (r3) [yshift = -0.5em] at (r2.center) {};
% do not remove
        \draw[-, line width = 2.5pt, draw = white] (r2) .. controls (m2) and (m3) .. (r3);
        \draw[-, line width = 1pt, draw = black] (r2) .. controls (m2) and (m3) .. (r3);
% do not remove
        \draw[-, line width = 2.5pt, draw = white] (l1) .. controls (m1) and (m2) .. (l2);
        \draw[-, line width = 1pt, draw = black] (l1) .. controls (m1) and (m2) .. (l2);
% do not remove
        \draw[-, line width = 2.5pt, draw = white] (l3) to (r1);
        \draw[-, line width = 1pt, draw = black] (l3) to (r1);
    \end{tikzpicture}
    +
    \begin{tikzpicture}[baseline = (baseline.center),
                        site/.style = {circle,
                                       draw = white,
                                       outer sep = 0.5pt,
                                       fill = black!80!white,
                                       inner sep = 1pt}]
        \node (m1) {};
        \node (m2) [yshift = -0.5em] at (m1.center) {};
        \node (m3) [yshift = -0.5em] at (m2.center) {};
% do not remove
        \node (baseline) [yshift = -0.25em] at (m2.center) {};
% do not remove
        \node[site] (l1) [xshift = -0.5em] at (m1.center) {};
        \node[site] (r1) [xshift = 0.5em] at (m1.center) {};
% do not remove
        \node[site] (l2) [yshift = -0.5em] at (l1.center) {};
        \node[site] (r2) [yshift = -0.5em] at (r1.center) {};
% do not remove
        \node[site] (l3) [yshift = -0.5em] at (l2.center) {};
        \node[site] (r3) [yshift = -0.5em] at (r2.center) {};
% do not remove
        \draw[-, line width = 2.5pt, draw = white] (r1) .. controls (m1) and (m2) .. (r2);
        \draw[-, line width = 1pt, draw = black] (r1) .. controls (m1) and (m2) .. (r2);
% do not remove
        \draw[-, line width = 2.5pt, draw = white] (l2) .. controls (m2) and (m3) .. (l3);
        \draw[-, line width = 1pt, draw = black] (l2) .. controls (m2) and (m3) .. (l3);
% do not remove
        \draw[-, line width = 2.5pt, draw = white] (l1) to (r3);
        \draw[-, line width = 1pt, draw = black] (l1) to (r3);
    \end{tikzpicture} 
    + 2 \Big(
    \begin{tikzpicture}[baseline = (baseline.center),
                        site/.style = {circle,
                                       draw = white,
                                       outer sep = 0.5pt,
                                       fill = black!80!white,
                                       inner sep = 1pt}]
        \node (m1) {};
        \node (m2) [yshift = -0.5em] at (m1.center) {};
        \node (m3) [yshift = -0.5em] at (m2.center) {};
% do not remove
        \node (baseline) [yshift = -0.25em] at (m2.center) {};
% do not remove
        \node[site] (l1) [xshift = -0.5em] at (m1.center) {};
        \node[site] (r1) [xshift = 0.5em] at (m1.center) {};
% do not remove
        \node[site] (l2) [yshift = -0.5em] at (l1.center) {};
        \node[site] (r2) [yshift = -0.5em] at (r1.center) {};
% do not remove
        \node[site] (l3) [yshift = -0.5em] at (l2.center) {};
        \node[site] (r3) [yshift = -0.5em] at (r2.center) {};
% do not remove
        \draw[-, line width = 2.5pt, draw = white] (l1) to (r1);
        \draw[-, line width = 1pt, draw = black] (l1) to (r1);
% do not remove
        \draw[-, line width = 2.5pt, draw = white] (l2) to (r2);
        \draw[-, line width = 1pt, draw = black] (l2) to (r2);
% do not remove
        \draw[-, line width = 2.5pt, draw = white] (l3) to (r3);
        \draw[-, line width = 1pt, draw = black] (l3) to (r3);
    \end{tikzpicture}
    +
    \begin{tikzpicture}[baseline = (baseline.center),
                        site/.style = {circle,
                                       draw = white,
                                       outer sep = 0.5pt,
                                       fill = black!80!white,
                                       inner sep = 1pt}]
        \node (m1) {};
        \node (m2) [yshift = -0.5em] at (m1.center) {};
        \node (m3) [yshift = -0.5em] at (m2.center) {};
% do not remove
        \node (baseline) [yshift = -0.25em] at (m2.center) {};
% do not remove
        \node[site] (l1) [xshift = -0.5em] at (m1.center) {};
        \node[site] (r1) [xshift = 0.5em] at (m1.center) {};
% do not remove
        \node[site] (l2) [yshift = -0.5em] at (l1.center) {};
        \node[site] (r2) [yshift = -0.5em] at (r1.center) {};
% do not remove
        \node[site] (l3) [yshift = -0.5em] at (l2.center) {};
        \node[site] (r3) [yshift = -0.5em] at (r2.center) {};
% do not remove
        \draw[-, line width = 2.5pt, draw = white] (l3) to (r1);
        \draw[-, line width = 1pt, draw = black] (l3) to (r1);
% do not remove
        \draw[-, line width = 2.5pt, draw = white] (l1) to (r2);
        \draw[-, line width = 1pt, draw = black] (l1) to (r2);
% do not remove
        \draw[-, line width = 2.5pt, draw = white] (l2) to (r3);
        \draw[-, line width = 1pt, draw = black] (l2) to (r3);
    \end{tikzpicture}
    +
    \begin{tikzpicture}[baseline = (baseline.center),
                        site/.style = {circle,
                                       draw = white,
                                       outer sep = 0.5pt,
                                       fill = black!80!white,
                                       inner sep = 1pt}]
        \node (m1) {};
        \node (m2) [yshift = -0.5em] at (m1.center) {};
        \node (m3) [yshift = -0.5em] at (m2.center) {};
% do not remove
        \node (baseline) [yshift = -0.25em] at (m2.center) {};
% do not remove
        \node[site] (l1) [xshift = -0.5em] at (m1.center) {};
        \node[site] (r1) [xshift = 0.5em] at (m1.center) {};
% do not remove
        \node[site] (l2) [yshift = -0.5em] at (l1.center) {};
        \node[site] (r2) [yshift = -0.5em] at (r1.center) {};
% do not remove
        \node[site] (l3) [yshift = -0.5em] at (l2.center) {};
        \node[site] (r3) [yshift = -0.5em] at (r2.center) {};
% do not remove
        \draw[-, line width = 2.5pt, draw = white] (l2) to (r1);
        \draw[-, line width = 1pt, draw = black] (l2) to (r1);
% do not remove
        \draw[-, line width = 2.5pt, draw = white] (l3) to (r2);
        \draw[-, line width = 1pt, draw = black] (l3) to (r2);
% do not remove
        \draw[-, line width = 2.5pt, draw = white] (l1) to (r3);
        \draw[-, line width = 1pt, draw = black] (l1) to (r3);
    \end{tikzpicture}
    -
    \begin{tikzpicture}[baseline = (baseline.center),
                        site/.style = {circle,
                                       draw = white,
                                       outer sep = 0.5pt,
                                       fill = black!80!white,
                                       inner sep = 1pt}]
        \node (m1) {};
        \node (m2) [yshift = -0.5em] at (m1.center) {};
        \node (m3) [yshift = -0.5em] at (m2.center) {};
% do not remove
        \node (baseline) [yshift = -0.25em] at (m2.center) {};
% do not remove
        \node[site] (l1) [xshift = -0.5em] at (m1.center) {};
        \node[site] (r1) [xshift = 0.5em] at (m1.center) {};
% do not remove
        \node[site] (l2) [yshift = -0.5em] at (l1.center) {};
        \node[site] (r2) [yshift = -0.5em] at (r1.center) {};
% do not remove
        \node[site] (l3) [yshift = -0.5em] at (l2.center) {};
        \node[site] (r3) [yshift = -0.5em] at (r2.center) {};
% do not remove
        \draw[-, line width = 2.5pt, draw = white] (l2) to (r1);
        \draw[-, line width = 1pt, draw = black] (l2) to (r1);
% do not remove
        \draw[-, line width = 2.5pt, draw = white] (l1) to (r2);
        \draw[-, line width = 1pt, draw = black] (l1) to (r2);
% do not remove
        \draw[-, line width = 2.5pt, draw = white] (l3) to (r3);
        \draw[-, line width = 1pt, draw = black] (l3) to (r3);
    \end{tikzpicture}
    -
    \begin{tikzpicture}[baseline = (baseline.center),
                        site/.style = {circle,
                                       draw = white,
                                       outer sep = 0.5pt,
                                       fill = black!80!white,
                                       inner sep = 1pt}]
        \node (m1) {};
        \node (m2) [yshift = -0.5em] at (m1.center) {};
        \node (m3) [yshift = -0.5em] at (m2.center) {};
% do not remove
        \node (baseline) [yshift = -0.25em] at (m2.center) {};
% do not remove
        \node[site] (l1) [xshift = -0.5em] at (m1.center) {};
        \node[site] (r1) [xshift = 0.5em] at (m1.center) {};
% do not remove
        \node[site] (l2) [yshift = -0.5em] at (l1.center) {};
        \node[site] (r2) [yshift = -0.5em] at (r1.center) {};
% do not remove
        \node[site] (l3) [yshift = -0.5em] at (l2.center) {};
        \node[site] (r3) [yshift = -0.5em] at (r2.center) {};
% do not remove
        \draw[-, line width = 2.5pt, draw = white] (l3) to (r1);
        \draw[-, line width = 1pt, draw = black] (l3) to (r1);
% do not remove
        \draw[-, line width = 2.5pt, draw = white] (l1) to (r3);
        \draw[-, line width = 1pt, draw = black] (l1) to (r3);
% do not remove
        \draw[-, line width = 2.5pt, draw = white] (l2) to (r2);
        \draw[-, line width = 1pt, draw = black] (l2) to (r2);
    \end{tikzpicture}
    -
    \begin{tikzpicture}[baseline = (baseline.center),
                        site/.style = {circle,
                                       draw = white,
                                       outer sep = 0.5pt,
                                       fill = black!80!white,
                                       inner sep = 1pt}]
        \node (m1) {};
        \node (m2) [yshift = -0.5em] at (m1.center) {};
        \node (m3) [yshift = -0.5em] at (m2.center) {};
% do not remove
        \node (baseline) [yshift = -0.25em] at (m2.center) {};
% do not remove
        \node[site] (l1) [xshift = -0.5em] at (m1.center) {};
        \node[site] (r1) [xshift = 0.5em] at (m1.center) {};
% do not remove
        \node[site] (l2) [yshift = -0.5em] at (l1.center) {};
        \node[site] (r2) [yshift = -0.5em] at (r1.center) {};
% do not remove
        \node[site] (l3) [yshift = -0.5em] at (l2.center) {};
        \node[site] (r3) [yshift = -0.5em] at (r2.center) {};
% do not remove
        \draw[-, line width = 2.5pt, draw = white] (l1) to (r1);
        \draw[-, line width = 1pt, draw = black] (l1) to (r1);
% do not remove
        \draw[-, line width = 2.5pt, draw = white] (l3) to (r2);
        \draw[-, line width = 1pt, draw = black] (l3) to (r2);
% do not remove
        \draw[-, line width = 2.5pt, draw = white] (l2) to (r3);
        \draw[-, line width = 1pt, draw = black] (l2) to (r3);
    \end{tikzpicture}
    \Big) \bigg].
\end{align}

We leave it as an open problem to analyse and understand in detail the structure of the optimal states in terms of the Brauer diagrams.

%%%%%%%%%%%%%%%%%%%%%%%%%%%%%%%%%%%%%%%%%%%%%%%%%%%%%%%%%%%%%%%%%%%%%%%%%%%%%%%%%%%%%%%%%%%%%%%%%%%%%%%%%%%%%%%%%%%%%%%%%%%%%%%%%%%%%%%%%%%%%%%%%%%%%%%%%%%%%%%%%%%%%%%%%%%%%%%%%%%%%%%%%%%%%%%%%%%%%%%%%%%%%%%%%%%%%%%%%%
\subsubsection{$K_{n,m}$-extendibility}
One can formulate optimization problems \cref{def:sdp_dual_worst_edgetransitive,def:sdp_state_primal} in the setting of other edge-transitive graphs $G$. In particular, if $G$ is a complete bipartite graph $K_{n,m}$, then the relevant optimization problems were solved in \cite{jakab2022extendibility} for Werner and isotropic states. For example, the optimal values $p_I(K_{n,m},d)$ and $p_B(K_{n,m},d)$ are:
\begin{equation}
    p_I(K_{n,m},d) = p_B(K_{n,m},d) = \frac{1}{d} + \of*{1-\frac{1}{d}} \frac{1}{\max (n,m)}.
\end{equation}

Note, that in the case of a complete bipartite graph $K_{n,m}$ the value $p_B(K_{n,m},d)$ of the Brauer state optimization problem \cref{def:sdp_dual_worst_edgetransitive} coincides with the value $p_I(K_{n,m},d)$ above due to symmetries of the graph $K_{n,m}$: one can twirl the solution of the Brauer state optimization problem with $U \otimes \bar{U}$ to get the isotropic solution without changing the optimal value $p_B(K_{n,m},d)$.

Notably, $q_W(K_{n,m},d)$ does not have a nice expression for general $d$ \cite{jakab2022extendibility}, so it is currently an open problem to find one. In the special case $m=1$, the formula for $q_W(K_{n,1},d)$ was found in \cite{johnson2013compatible}:
\begin{equation}
    q_W(K_{n,1},d) = \min \of*{\frac{n+d-1}{2n},1}.
\end{equation}

In general, it is interesting to understand the values $q_W(G,d),\, p_I(G,d),\, p_B(G,d)$ as well as the full $(p,q)$ extendibility region for arbitrary edge-transitive graphs. In the following section, we provide an application of the Werner state extendibility on cycle graphs $C_n$.

%%%%%%%%%%%%%%%%%%%%%%%%%%%%%%%%%%%%%%%%%%%%%%%%%%%%%%%%%%%%%%%%%%%%%%%%%%%%%%%%%%%%%%%%%%%%%%%%%%%%%%%%%%%%%%%%%%%%%%%%%%%%%%%%%%%%%%%%%%%%%%%%%%%%%%%%%%%%%%%%%%%%%%%%%%%%%%%%%%%%%%%%%%%%%%%%%%%%%%%%%%%%%%%%%%%%%%%%%%
%%%%%%%%%%%%%%%%%%%%%%%%%%%%%%%%%%%%%%%%%%%%%%%%%%%%%%%%%%%%%%%%%%%%%%%%%%%%%%%%%%%%%%%%%%%%%%%%%%%%%%%%%%%%%%%%%%%%%%%%%%%%%%%%%%%%%%%%%%%%%%%%%%%%%%%%%%%%%%%%%%%%%%%%%%%%%%%%%%%%%%%%%%%%%%%%%%%%%%%%%%%%%%%%%%%%%%%%%%
%%%%%%%%%%%%%%%%%%%%%%%%%%%%%%%%%%%%%%%%%%%%%%%%%%%%%%%%%%%%%%%%%%%%%%%%%%%%%%%%%%%%%%%%%%%%%%%%%%%%%%%%%%%%%%%%%%%%%%%%%%%%%%%%%%%%%%%%%%%%%%%%%%%%%%%%%%%%%%%%%%%%%%%%%%%%%%%%%%%%%%%%%%%%%%%%%%%%%%%%%%%%%%%%%%%%%%%%%%
\section{Cyclic graph (\texorpdfstring{$n \to \infty$}{n to infinity} limit) and Bell state discrimination}\label{sec:Cyclic}

\noindent It is also interesting to look at other edge-transitive graphs for our SDP in \cref{def:sdp_dual_worst_edgetransitive}. In particular, a family of cyclic graphs is immediately interesting. Consider the cyclic graph with $n$ vertices $G = C_n$ and local dimension $d=2$, then
\begin{equation}
    p_W(C_n,2) = \lambda_{\mathrm{max}}(H_{C_n}),
\end{equation}
where $H_{C_n}$ takes the simple form $H_{C_n} = \frac{1}{n} \sum_{e \in E(C_n)} \frac{\I_e - \F_e}{2}$. Here we encounter an interesting connection to the study of integrable one-dimensional spin chains. The spectrum of the Hamiltonian $H_{C_n}$ has previously been analysed in that context, namely in the antiferromagnetic Heisenberg XXX$_{1/2}$ model \cite{heisenberg1928,bethe1931theorie} whose Hamiltonian can be reformulated as $H_{\mathrm{XXX}} = - H_{C_n}$. In particular, the thermodynamic limit of $n \to \infty$ spins (vertices in our setting) has been solved analytically using the algebraic Bethe ansatz \cite{faddeev1996algebraic} with the result $\lambda_{\mathrm{min}}(H_{\mathrm{XXX}})/n \to -\ln(2)$ as $n \to \infty$. This implies
\begin{align} \label{eq:wernervalueinfinitecircle}
    p_W(C_n,2) \to \ln(2) \quad \text{ as } n \to \infty.
\end{align}
We can use this value to get new bounds on a state discrimination problem described below.

\subsection{Bell state discrimination in the simultaneous communication setting}
We find an application of the above result in the setting of time-constrained local state discrimination. The setting is as follows: Consider two verifiers $V_A, V_B$ who send their respective register of a bipartite quantum states $\rho_k^{AB}$, drawn from some set of orthogonal states $\{\rho_k^{AB}\}_k$, to two spatially separated players $A,B$, who do not pre-share any entanglement. The task of $A,B$ is to recover the index $k$ using only local operations and one round of simultaneous communication (see Fig. \ref{fig:spacetimediagram}), and send their answer back to the verifiers on time. Operations within this model can also be seen as attacks on a cryptographic primitive called \textit{Quantum Position Verification}.

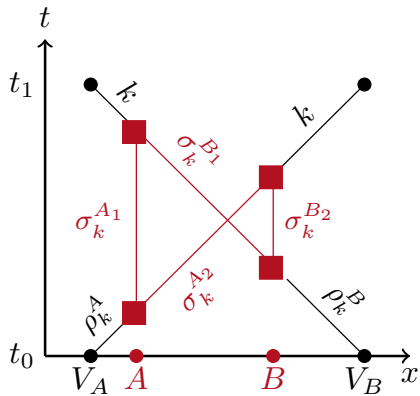
\begin{figure}[h]
\centering
\scalebox{1.2}{
\begin{tikzpicture} 
    %Axes
    \draw[->,black,thick] (-0.5,-3) -- (-0.5,0.5) node[above]{$t$};
    \draw[->,black,thick] (-0.5,-3) -- (3.5,-3) node[below]{$x$};

    \draw[black]  (0,0) -- (0.5,-0.5) node[midway, above, sloped]{$k$};
    \draw[darkred]  (0.5,-0.5) -- (2,-2) node[midway, above left, sloped]{\small$\sigma_k^{B_1}$};
    \draw[darkred]  (0.5,-0.5) -- (0.5,-2.5) node[midway, left]{\small$\sigma_k^{A_1}$};
    \draw[black]  (0.5,-2.5) -- (0,-3) node[midway, above, sloped]{\small $\rho^A_k$};
    \filldraw [black] (0,-3) circle (2pt) node[below]{$V_A$};

    \draw[black]  (3,0) -- (2,-1) node[midway, above, sloped]{$k$};
    \draw[darkred]  (2,-1) -- (0.5,-2.5) node[midway, below left, sloped]{\small$\sigma_k^{A_2}$};
    \draw[darkred]  (2,-1) -- (2,-2) node[midway, right]{\small$\sigma_k^{B_2}$};
    \draw[black]  (2.15,-2.15) -- (3,-3) node[midway, above, sloped]{\small $\rho^B_k$};
    \filldraw [black] (3,-3) circle (2pt) node[below]{$V_B$};

    \filldraw[darkred] (0.35,-0.65) rectangle ++(0.25,0.25);
    \filldraw[darkred] (1.85,-1.15) rectangle ++(0.25,0.25);
    \filldraw[darkred] (0.35,-2.65) rectangle ++(0.25,0.25);
    \filldraw[darkred] (1.85,-2.15) rectangle ++(0.25,0.25);

    \filldraw [black] (0,0) circle (2pt);
    \filldraw [black] (3,0) circle (2pt);
    
    \filldraw [darkred] (0.5,-3) circle (2pt) node[below]{$A$};
    \filldraw [darkred] (2,-3) circle (2pt) node[below]{$B$};

    \draw[black] (-0.75,-3) node[] (2pt) {$t_{0}$};
    \draw[black] (-0.75,0) node[] (2pt) {$t_{1}$};
\end{tikzpicture}
}
\caption{Spacetime diagram of time constrained local state discrimination. Two verifiers $V_A, V_B$ simultaneously send $\rho_A^k, \rho_B^k$ at time $t_0$ to two spatially separated players $A,B$, who have to distinguish the index $k$ using local operations and one round of simultaneous communication. The registers that get send over, $A_2, B_1$, are classical in the LOSCC setting, but can be quantum in the LOSQC setting. The registers that both players keep, $A_1, B_2$ can be quantum in both settings.}
\label{fig:spacetimediagram}
\end{figure}

We consider two models of time-constrained local state discrimination. The round of simultaneous communication can either be \textit{classical} or \textit{quantum}. These models are denoted by LOSCC and LOSQC, respectively. It was shown that there exist ensembles of entangled states as well as product states that can be discriminated perfectly under LOSQC operations, but not under LOSCC \cite{allerstorfer2022role, george2023time}, showing that there is a separation between the two models. 

In what follows, we focus on the specific task of time-constrained Bell state discrimination, where the input state ensemble of the verifiers is the four Bell states $\{ \ket{\Phi_+}, \ket{\Phi_-}, \ket{\Psi_+}, \ket{\Psi_-} \}$, which are sent uniformly at random. We allow for strategies that are approximately correct, in the sense that a strategy succeeds with probability $p$ if both answers from $A,B$ equal $k$ with probability $p$. That is, we are interested in how well we can approximate a Bell state measurement under LOSCC and LOSQC. 

When $A,B$ are restricted to LOSCC operations, we have an optimal probability of $1/2$ to distinguish the Bell states correctly. If $A,B$ both measure in the computational basis and check whether their measurements outcomes are equal or not, this probability can be achieved. It is an open problem whether there exists a strategy in LOSQC that can outperform LOSCC. Until now, the best known upper bound on this task was $3/4$ \cite{allerstorfer2022role}. With our results, we can get a better upper bound of $\ln(2) \approx 0.69$. 

\subsubsection{Connection to entanglement swapping}

We will connect the existence of good strategies for Bell state discrimination in LOSQC to the existence of a line graph state whose value $p_W$ depends on the probability $p_{\mathrm{succ}}$ of correctly distinguishing the Bell state. To make this connection, we consider an equivalent setting in which we purify the local inputs as in \cite{allerstorfer2022role}.

In the purified setting, the verifiers locally generate an EPR pair, keep half of the pair, and send the other half as the input to the players. The task for $A,B$ does not change, they need to perform a Bell State measurement under LOSQC and answer their outcome to the verifiers. Then, note the following entanglement swapping identity:
\begin{align}\label{equ:eswapping}
\begin{split}
\ket{\Phi_+}_{V_AA} \ket{\Phi_+}_{V_BB} = &\frac{1}{2} \Big[ \ket{\Phi_+}_{V_AV_B} \ket{\Phi_+}_{AB} + \ket{\Phi_-}_{V_AV_B} \ket{\Phi_-}_{AB} + \ket{\Psi_+}_{V_AV_B} \ket{\Psi_+}_{AB} + \ket{\Psi_-}_{V_AV_B} \ket{\Psi_-}_{AB}\Big],
\end{split}
\end{align}
which shows that the measurement outcome of a Bell state measurement on the verifiers qubits is the same as on the qubits $A,B$ receive. Thus, if the verifiers can later check whether $A,B$ applied a Bell state measurement by comparing the answer they receive to their own measurement result. From the point of view of $A,B$ nothing changes if the verifiers measure their local qubits before receiving an answer, or even before they sent the qubits to $A,B$ (which was the regular setting). Therefore, the probability that the verifiers have the same measurement outcome as the answer they receive in this purified setting is exactly the probability $p_{\mathrm{succ}}$ with which $A,B$ can distinguish the Bell states under LOSQC in the non-purified protocol. The settings are in this sense equivalent. 

Consider the state the verifiers hold in the purified picture after they receive the answers from $A,B$, but \textit{before} they apply a Bell state measurement to check if the answer is correct. We know that this state has overlap $p_{\mathrm{succ}}$ with the Bell state that corresponds to the answer they received. We can now do the following:
\begin{enumerate}
    \item Apply a local Pauli operation that maps the answer they received to the antisymmetric Bell state.
    \item Apply the same Haar random single qubit unitary to both qubits, i.e. a \textit{Werner twirling channel}. 
\end{enumerate}
Then, the state the verifiers get is exactly a Werner state as in $\eqref{eq:WernerState}$, with parameter $p$ corresponding to the probability $p_{\mathrm{succ}}$ of successfully distinguishing the Bell states under LOSQC.

\subsubsection{Reducing the success probability to the cyclic graph}

In its most general form, two players $A,B$ will apply some map on their local inputs with two output registers, one that they keep and one that they send on to the other player (see Fig.~\ref{fig:spacetimediagram}). In the purified picture the registers $A_1, A_2$ can also be entangled to $V_A$, and the registers $B_1, B_2$ to $V_B$. As we have established, the task of players $A,B$ is to swap entanglement into the registers of $V_A, V_B$. Crucially, in the task of Bell state discrimination it is clear that the answer of only a single player suffices to distinguish the correct state since both players have to answer correctly. Conversely, this implies that only a single final operation at one of the players swaps an entangled state with overlap $p$ to some Bell state into the registers of the verifiers.

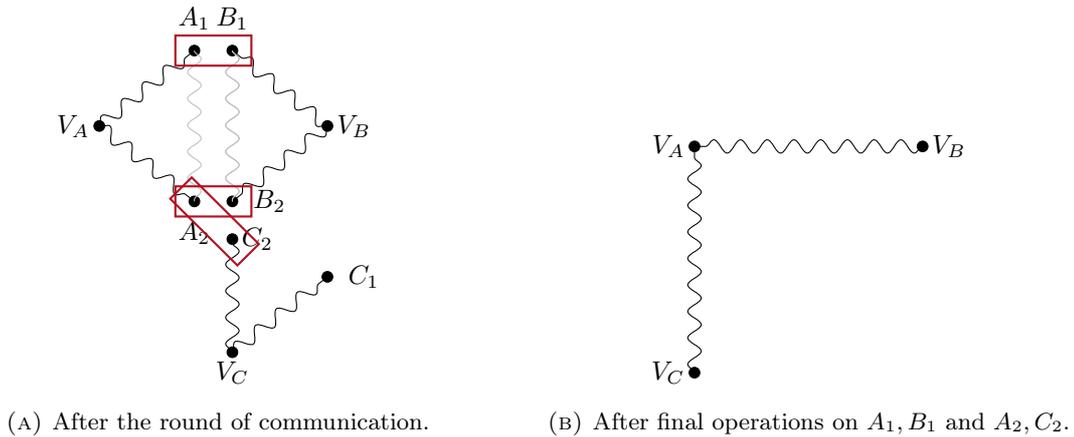
\begin{figure}[ht]
    \begin{subfigure}{0.45\textwidth}
    \centering
    \begin{tikzpicture}
        \filldraw [black] (0,0) circle (2pt) node[left] {${V_A}$};
        \draw[black] [,decorate,decoration=snake] (0,0) -- (1.25,1);
        \draw[black] [,decorate,decoration=snake] (0,0) -- (1.25,-1);
        \filldraw [black] (1.25,1) circle (2pt) node[above=1.5mm] {${A_1}$};
    
        \filldraw [black] (3,0) circle (2pt) node[right] {${V_B}$};
        \draw[black] [,decorate,decoration=snake] (3,0) -- (1.75,1);
        \draw[black] [,decorate,decoration=snake] (3,0) -- (1.75,-1);
        \filldraw [black] (1.75,1) circle (2pt) node[above=1.5mm] {${B_1}$};
        
        \filldraw [black] (1.75,-3) circle (2pt) node[below] {${V_C}$};
        \draw[black] [,decorate,decoration=snake] (1.75,-3) -- (1.75,-1.5);
        \draw[black] [,decorate,decoration=snake] (1.75,-3) -- (3,-2);
        \filldraw [black] (1.75,-1.5) circle (2pt) node[right] {${C_2}$};
    
        \filldraw [black] (1.25,-1) circle (2pt) node[below=1.5mm] {${A_2}$};
        \filldraw [black] (1.75,-1) circle (2pt) node[right=1.5mm] {${B_2}$};
        \filldraw [black] (3,-2) circle (2pt) node[right=1.5mm] {${C_1}$};
        
        \draw[lightgray] [,decorate,decoration=snake] (1.25,1) -- (1.25,-1);
        \draw[lightgray] [,decorate,decoration=snake] (1.75,1) -- (1.75,-1);
        \draw[lightgray] [,decorate,decoration=snake] (1.75,1) -- (1.75,-1);

        \draw[darkred, thick] (1,0.8) rectangle (2,1.2);
        \draw[darkred, thick, rotate around={-45:(1.25,-1)}] (1,-0.8) rectangle (2.25,-1.2);
        \draw[darkred, thick, opacity=0.33] (1,-0.8) rectangle (2,-1.2);
    \end{tikzpicture}
    \caption{After the round of communication.}
    \label{fig:TwoStateBeforeMeas}
    \end{subfigure}
    \begin{subfigure}{0.45\textwidth}
    \centering
    \begin{tikzpicture}
        %Upper verifier and their entanglement
        \filldraw [black] (0,0) circle (2pt) node[left] {${V_A}$};   
        \filldraw [black] (3,0) circle (2pt) node[right] {${V_B}$};
        \draw[black] [,decorate,decoration=snake] (3,0) -- (0,0);

        %Lower verifier and their entanglement
        \filldraw [black] (0,-3) circle (2pt) node[left] {${V_{C}}$};   

        %Left entanglement
        \draw[black] [,decorate,decoration=snake] (0,-3) -- (0,0);
    \end{tikzpicture}
    \caption{After final operations on $A_1,B_1$ and $A_2,C_2$.}
    \label{fig:TwoStateAfterMeas}
\end{subfigure}
\caption{Entanglement structure including a third hypothetical verifier ${V_C}$ and player ${C}$ who applies the same operation as ${B}$ to his input. After the final measurement operation entanglement will be swapped to the shared state of the verifiers. We have omitted the entanglement structure with the individual registers here.}
\label{fig:entanglementdiagrams}
\end{figure}

Now consider the situation in which there is a third verifier $V_C$, who behaves exactly like $V_B$, and we consider a third player $C$ who applies the same quantum map as $B$ on the input he receives from $V_C$. As a thought experiment, we now apply the same final operation on the $A_2,C_2$ register as would have been applied on the $A_2,B_2$ register. In Fig. \ref{fig:TwoStateBeforeMeas} we see the entanglement structure visualized after the communication round, but before the final operations. After the final operations, we get the structure as in Fig. \ref{fig:TwoStateAfterMeas}. By the previous argument, the reduced states on $V_A,V_B$ and $V_A,V_C$ now both have overlap $p_{\mathrm{succ}}$ with some Bell state that corresponds to the answer the players get. Note that these Bell states do not have to be equal. By applying a local Pauli gate on the qubits that $V_B$ and $V_C$ we can bring these states to the antisymmetric Bell state $\ket{\Psi_-}$. We can apply a $\textit{Werner Twirling Channel}$ on the joint state of the verifiers to end up with a line graph state with 3 vertices as in $\eqref{eq:WernerState}$, where the parameter $p$ for the reduced states on adjacent pairs of qubits is equal to $p_{\mathrm{succ}}$. Thus, LOSQC strategies imply the existence of certain Werner graph states.

We can extend the above argument by introducing more hypothetical verifiers, and combining them in the same way, to get a line graph state with any amount of vertices. This leads to the following proposition that relates LOSQC strategies to the existence of states:

\begin{proposition}
    If there exists a LOSQC strategy for the task of Bell state discrimination, where the Bell states are picked uniformly at random, that succeeds with probability $p_{\mathrm{succ}}$, then there exists a line graph state of any length $n$ of qubits, such that all neighboring pairs of qubits are in Werner states with parameter $p$ = $p_{\mathrm{succ}}$.
\end{proposition}

The value of the parameter $p$ can be computed for the line graph and decreases for larger values of $n$. If we take equal weights for the line in the dual program, which is a feasible solution to the dual program, we get an upper bound for the value of $p$. Taking the limit of $n \to \infty$ for this input line coincides with the value of the infinite circle where all the weights are equal \cite[equ. (2.50)]{Alcaraz_1987}. Thus we get $p \leq \ln (2)$ as in equation \eqref{eq:wernervalueinfinitecircle}. This gives us the following corollary.

\begin{cor}
    The LOSQC success probability of discriminating Bell states $p_\mathrm{succ}$ is upper bounded by the value $\lim_{n \to \infty} p_W(C_n, 2) = \ln(2) \approx 0.69$. In particular, the attack success probability of \emph{QPV}$_\emph{Bell}$ with no pre-shared entanglement is bounded by the same value.
\end{cor}

\bigskip

\section{Acknowledgements} 

We thank the Centre International de Rencontres Mathématiques (CIRM) and the Random Tensors conference, where this work was initiated, for their hospitality. R.A.~was supported by the Dutch Research Council (NWO/OCW), as part of the Quantum Software Consortium programme (project number 024.003.037). P.VL.~was supported by the Dutch Research Council (NWO/OCW), as part of the NWO Gravitation Programme Networks (project number 024.002.003). M.C.~acknowledges financial support from the European Research Council (ERC Grant Agreement No. 818761), VILLUM FONDEN via the QMATH Centre of Excellence (Grant No.10059) and the Novo Nordisk Foundation (grant NNF20OC0059939 ‘Quantum for Life’). I.N.~and D.R.~were supported by the ANR project \href{https://esquisses.math.cnrs.fr/}{ESQuisses}, grant number ANR-20-CE47-0014-01, and by the PHC program \emph{Star} (Applications of random matrix theory and abstract harmonic analysis to quantum information theory). I.N.~ also received support from the ANR project \href{https://www.math.univ-toulouse.fr/~gcebron/STARS.php}{STARS}, grant number ANR-20-CE40-0008. D.G.~and M.O.~were supported by NWO Vidi grant (Project No. VI.Vidi.192.109).

%%%%%%%%%%%%%%%%%%%%%%%%%%%%%%%%%%%%%%%%%%%%%%%%%%%%%%%%%%%%%%%%%%%%%%%%%%%%%%%%%%%%%%%%%%%%%%%%%%%%%%%%%%%%%%%%%%%%%%%%%%%%%%%%%%%%%%%%%%%%%%%%%%%%%%%%%%%%%%%%%%%%%%%%%%%%%%%%%%%%%%%%%%%%%%%%%%%%%%%%%%%%%%%%%%%%%%%%%%
%%%%%%%%%%%%%%%%%%%%%%%%%%%%%%%%%%%%%%%%%%%%%%%%%%%%%%%%%%%%%%%%%%%%%%%%%%%%%%%%%%%%%%%%%%%%%%%%%%%%%%%%%%%%%%%%%%%%%%%%%%%%%%%%%%%%%%%%%%%%%%%%%%%%%%%%%%%%%%%%%%%%%%%%%%%%%%%%%%%%%%%%%%%%%%%%%%%%%%%%%%%%%%%%%%%%%%%%%%

\printbibliography

@article{allerstorfer2022role,
    archiveprefix = {arXiv},
    author = {Allerstorfer, Rene and Buhrman, Harry and Speelman, Florian and Lunel, Philip Verduyn},
    eprint = {2208.04341},
    title = {On the Role of Quantum Communication and Loss in Attacks on Quantum Position Verification},
    year = {2022}
}

@article{Alcaraz_1987,
doi = {10.1088/0305-4470/20/18/038},
year = {1987},
publisher = {},
volume = {20},
number = {18},
pages = {6397},
author = {F. C. Alcaraz and  M. N. Barber and  M. T. Batchelor and  R. J. Baxter and  G. R. W. Quispel},
title = {Surface exponents of the quantum XXZ, Ashkin-Teller and Potts models},
journal = {Journal of Physics A: Mathematical and General},
}

@article{george2023time,
  title={Time-Constrained Local Quantum State Discrimination},
  author={George, Ian and Allerstorfer, Rene and Lunel, Philip Verduyn and Chitambar, Eric},
  journal={arXiv preprint arXiv:2311.00677},
  year={2023}
}

@article{andersen1992edge,
    author = {Andersen, Lars D{\o}vling and Ding, Songkang and Sabidussi, Gert and Vestergaard, Preben Dahl},
    doi = {10.1007/BF01271706},
    journal = {Graphs and Combinatorics},
    pages = {31--44},
    publisher = {Springer},
    title = {Edge orbits and edge-deleted subgraphs},
    volume = {8},
    year = {1992}
}

@article{andersen2017semisimplicity,
    author = {Andersen, Henning Haahr and Stroppel, Catharina and Tubbenhauer, Daniel},
    doi = {10.1017/S1446788716000392},
    journal = {Journal of the Australian Mathematical Society},
    number = {1},
    pages = {1--44},
    publisher = {Cambridge University Press},
    title = {Semisimplicity of Hecke and (walled) Brauer algebras},
    volume = {103},
    year = {2017}
}

@article{anshu2020beyond,
    archiveprefix = {arXiv},
    author = {Anshu, Anurag and Gosset, David and Morenz, Karen},
    eprint = {2003.14394},
    title = {Beyond product state approximations for a quantum analogue of Max Cut},
    year = {2020}
}

@article{bethe1931theorie,
    author = {Bethe, Hans},
    doi = {10.1007/BF01341708},
    journal = {Zeitschrift f{\"u}r Physik},
    number = {3-4},
    pages = {205--226},
    publisher = {Springer},
    title = {Zur theorie der metalle: I. Eigenwerte und eigenfunktionen der linearen atomkette},
    volume = {71},
    year = {1931}
}

@book{biggs1993algebraic,
    author = {Biggs, Norman},
    doi = {10.1017/CBO9780511608704},
    number = {67},
    publisher = {Cambridge university press},
    title = {Algebraic graph theory},
    year = {1993}
}

@book{boyd2004convex,
    author = {Boyd, Stephen P and Vandenberghe, Lieven},
    doi = {10.1145/2020408.2020410},
    publisher = {Cambridge university press},
    title = {Convex optimization},
    year = {2004}
}

@article{brauer1937algebras,
    author = {Brauer, Richard},
    doi = {10.2307/1968843},
    journal = {Annals of Mathematics},
    pages = {857--872},
    publisher = {JSTOR},
    title = {On algebras which are connected with the semisimple continuous groups},
    year = {1937}
}

@article{bravyi2022generating,
    archiveprefix = {arXiv},
    author = {Bravyi, Sergey and Sharma, Yash and Szegedy, Mario and De Wolf, Ronald},
    eprint = {2211.06497},
    title = {Generating $ k $ EPR-pairs from an n-party resource state},
    year = {2022}
}

@article{buhrman2014position,
    author = {Buhrman, Harry and Chandran, Nishanth and Fehr, Serge and Gelles, Ran and Goyal, Vipul and Ostrovsky, Rafail and Schaffner, Christian},
    doi = {10.1007/978-3-642-22792-9_24},
    journal = {SIAM Journal on Computing},
    number = {1},
    pages = {150--178},
    publisher = {SIAM},
    title = {Position-based quantum cryptography: Impossibility and constructions},
    volume = {43},
    year = {2014}
}

@article{caves2002unknown,
    author = {Caves, Carlton M and Fuchs, Christopher A and Schack, R{\"u}diger},
    doi = {10.1063/1.1494475},
    journal = {Journal of Mathematical Physics},
    number = {9},
    pages = {4537--4559},
    publisher = {American Institute of Physics},
    title = {Unknown quantum states: the quantum de Finetti representation},
    volume = {43},
    year = {2002}
}

@article{Christandl2007,
    author = {Christandl, Matthias and König, Robert and Mitchison, Graeme and Renner, Renato},
    doi = {10.1007/s00220-007-0189-3},
    eprint = {quant-ph/0602130},
    issn = {1432-0916},
    journal = {Communications in Mathematical Physics},
    number = {2},
    pages = {473--498},
    publisher = {Springer},
    title = {One-and-a-Half Quantum de Finetti Theorems},
    volume = {273},
    year = {2007}
}

@article{doherty2004complete,
    author = {Doherty, Andrew C and Parrilo, Pablo A and Spedalieri, Federico M},
    doi = {10.1103/PhysRevA.69.022308},
    journal = {Physical Review A},
    number = {2},
    pages = {022308},
    publisher = {APS},
    title = {Complete family of separability criteria},
    volume = {69},
    year = {2004}
}

@article{doran1999semisimplicity,
    author = {Doran IV, William F and Wales, David B and Hanlon, Philip J},
    doi = {10.1006/jabr.1998.7592},
    journal = {Journal of Algebra},
    number = {2},
    pages = {647--685},
    publisher = {Academic Press},
    title = {On the semisimplicity of the Brauer centralizer algebras},
    volume = {211},
    year = {1999}
}

@article{doty2019canonical,
    author = {Doty, Stephen and Lauve, Aaron and Seelinger, George H},
    doi = {10.4171/LEM/64-1/2-2},
    journal = {L’Enseignement Math{\'e}matique},
    number = {1},
    pages = {23--63},
    title = {Canonical idempotents of multiplicity-free families of algebras},
    volume = {64},
    year = {2019}
}

@article{faddeev1996algebraic,
    archiveprefix = {arXiv},
    author = {Faddeev, Ludvig},
    doi = {10.1142/9789814340960_0031},
    eprint = {hep-th/9605187},
    title = {How algebraic Bethe ansatz works for integrable model},
    year = {1996}
}

@article{grinko2022linear,
    archiveprefix = {arXiv},
    author = {Grinko, Dmitry and Ozols, Maris},
    eprint = {2207.05713},
    title = {Linear programming with unitary-equivariant constraints},
    year = {2022}
}

@article{heisenberg1928,
    author = {Heisenberg, Werner},
    doi = {10.1007/978-3-642-61659-4_35},
    journal = {Zeitschrift f{\"u}r Physik},
    pages = {619--636},
    publisher = {Springer},
    title = {Zur Theorie des Ferromagnetismus},
    volume = {49},
    year = {1928}
}

@article{horodecki1999reduction,
    author = {Horodecki, Micha{\l} and Horodecki, Pawe{\l}},
    doi = {10.1103/PhysRevA.59.4206},
    journal = {Physical Review A},
    number = {6},
    pages = {4206},
    publisher = {APS},
    title = {Reduction criterion of separability and limits for a class of distillation protocols},
    volume = {59},
    year = {1999}
}

@article{hudson1976locally,
    author = {Hudson, Robin L and Moody, Graham R},
    doi = {10.1007/BF00534784},
    journal = {Zeitschrift f{\"u}r Wahrscheinlichkeitstheorie und verwandte Gebiete},
    number = {4},
    pages = {343--351},
    publisher = {Springer},
    title = {Locally normal symmetric states and an analogue of de Finetti's theorem},
    volume = {33},
    year = {1976}
}

@article{jakab2022extendibility,
    archiveprefix = {arXiv},
    author = {Jakab, D{\'a}vid and Solymos, Adri{\'a}n and Zimbor{\'a}s, Zolt{\'a}n},
    eprint = {2208.13743},
    title = {Extendibility of Werner States},
    year = {2022}
}

@article{johnson2013compatible,
    author = {Johnson, Peter D and Viola, Lorenza},
    doi = {10.1103/PhysRevA.88.032323},
    journal = {Physical Review A},
    number = {3},
    pages = {032323},
    publisher = {APS},
    title = {Compatible quantum correlations: Extension problems for Werner and isotropic states},
    volume = {88},
    year = {2013}
}

@article{kent2011quantum,
    author = {Kent, Adrian and Munro, William J and Spiller, Timothy P},
    doi = {10.1103/PhysRevA.84.012326},
    journal = {Physical Review A},
    number = {1},
    pages = {012326},
    publisher = {APS},
    title = {Quantum tagging: Authenticating location via quantum information and relativistic signaling constraints},
    volume = {84},
    year = {2011}
}

@article{keyl1999optimal,
    author = {Keyl, Michael and Werner, Reinhard F},
    doi = {10.1063/1.532887},
    journal = {Journal of Mathematical Physics},
    number = {7},
    pages = {3283--3299},
    publisher = {American Institute of Physics},
    title = {Optimal cloning of pure states, testing single clones},
    volume = {40},
    year = {1999}
}

@article{king2022improved,
    archiveprefix = {arXiv},
    author = {King, Robbie},
    eprint = {2209.02589},
    title = {An improved approximation algorithm for quantum max-cut},
    year = {2022}
}

@article{koashi2004monogamy,
    author = {Koashi, Masato and Winter, Andreas},
    doi = {10.1103/PhysRevA.69.022309},
    journal = {Physical Review A},
    number = {2},
    pages = {022309},
    publisher = {APS},
    title = {Monogamy of quantum entanglement and other correlations},
    volume = {69},
    year = {2004}
}

@article{konig2005finetti,
    author = {K{\"o}nig, Robert and Renner, Renato},
    doi = {10.1063/1.2146188},
    journal = {Journal of Mathematical physics},
    number = {12},
    publisher = {AIP Publishing},
    title = {A de Finetti representation for finite symmetric quantum states},
    volume = {46},
    year = {2005}
}

@article{nechita2023asymmetric,
    author = {Nechita, Ion and Pellegrini, Cl{\'e}ment and Rochette, Denis},
    doi = {10.1007/s11005-023-01694-8},
    journal = {Letters in Mathematical Physics},
    number = {3},
    pages = {74},
    publisher = {Springer},
    title = {The asymmetric quantum cloning region},
    volume = {113},
    year = {2023}
}

@article{okada2016pieri,
    author = {Soichi Okada},
    doi = {10.37236/6214},
    eprint = {1606.02375},
    journal = {Electron. J. Comb.},
    pages = {4},
    title = {Pieri Rules for Classical Groups and Equinumeration between Generalized Oscillating Tableaux and Semistandard Tableaux},
    volume = {23},
    year = {2016}
}

@article{okounkov1996shifted,
    archiveprefix = {arXiv},
    author = {Okounkov, Andrei and Olshanski, Grigori},
    eprint = {q-alg/9605042},
    journal = {Algebra i Analiz},
    number = {2},
    pages = {73--146},
    title = {Shifted Schur functions},
    volume = {9},
    year = {1997}
}

@article{OV1996,
    author = {Okounkov, Andrei and Vershik, Anatoly},
    doi = {10.1007/BF02433451},
    journal = {Selecta Mathematica, New Series},
    number = {4},
    pages = {581--605},
    publisher = {Birkhäuser},
    title = {A new approach to representation theory of symmetric groups},
    volume = {2},
    year = {1996}
}

@article{parekh2021application,
    archiveprefix = {arXiv},
    author = {Parekh, Ojas and Thompson, Kevin},
    doi = {10.2172/1884683},
    eprint = {2105.05698},
    title = {Application of the Level-$2 $ Quantum Lasserre Hierarchy in Quantum Approximation Algorithms},
    year = {2021}
}

@article{park2023universal,
    author = {Park, Sang-Jun and Jung, Yeong-Gwang and Park, Jeongeun and Youn, Sang-Gyun},
    journal = {arXiv preprint arXiv:2301.03849},
    title = {A universal framework for entanglement detection under group symmetry},
    year = {2023}
}

@article{rui2005criterion,
    author = {Rui, Hebing},
    doi = {10.1016/j.jcta.2004.11.009},
    journal = {Journal of Combinatorial Theory, Series A},
    number = {1},
    pages = {78--88},
    publisher = {Academic Press},
    title = {A criterion on the semisimple Brauer algebras},
    volume = {111},
    year = {2005}
}

@article{rui2006criterion,
    author = {Rui, Hebing and Si, Mei},
    doi = {10.1016/j.jcta.2005.09.005},
    journal = {Journal of Combinatorial Theory, Series A},
    number = {6},
    pages = {1199--1203},
    publisher = {Elsevier},
    title = {A criterion on the semisimple Brauer algebras II},
    volume = {113},
    year = {2006}
}

@article{takahashi20232,
    archiveprefix = {arXiv},
    author = {Takahashi, Jun and Rayudu, Chaithanya and Zhou, Cunlu and King, Robbie and Thompson, Kevin and Parekh, Ojas},
    eprint = {2307.15688},
    title = {An SU (2)-symmetric Semidefinite Programming Hierarchy for Quantum Max Cut},
    year = {2023}
}

@article{terhal2003symmetric,
    author = {Terhal, Barbara M and Doherty, Andrew C and Schwab, David},
    doi = {10.1103/PhysRevLett.90.157903},
    journal = {Physical review letters},
    number = {15},
    pages = {157903},
    publisher = {APS},
    title = {Symmetric extensions of quantum states and local hidden variable theories},
    volume = {90},
    year = {2003}
}

@inproceedings{unruh2014quantum,
    author = {Unruh, Dominique},
    booktitle = {Advances in Cryptology--CRYPTO 2014: 34th Annual Cryptology Conference, Santa Barbara, CA, USA, August 17-21, 2014, Proceedings, Part II 34},
    doi = {10.1007/978-3-662-44381-1_1},
    organization = {Springer},
    pages = {1--18},
    title = {Quantum position verification in the random oracle model},
    year = {2014}
}

@article{VO2005,
    archiveprefix = {arXiv},
    author = {Vershik, Anatoly and Okounkov, Andrei},
    doi = {10.1007/s10958-005-0421-7},
    eprint = {math/0503040},
    journal = {Journal of Mathematical Sciences},
    pages = {5471--5494},
    title = {A new approach to the representation theory of the symmetric groups. II},
    volume = {131},
    year = {2005}
}

@article{vollbrecht2001entanglement,
    author = {Vollbrecht, Karl Gerd H and Werner, Reinhard F},
    doi = {10.1103/PhysRevA.64.062307},
    journal = {Physical Review A},
    number = {6},
    pages = {062307},
    publisher = {APS},
    title = {Entanglement measures under symmetry},
    volume = {64},
    year = {2001}
}

@article{watts2023relaxations,
    author = {Watts, Adam Bene and Chowdhury, Anirban and Epperly, Aidan and Helton, J William and Klep, Igor},
    journal = {arXiv preprint arXiv:2307.15661},
    title = {Relaxations and Exact Solutions to Quantum Max Cut via the Algebraic Structure of Swap Operators},
    year = {2023}
}

@article{wenzl1988structure,
    author = {Wenzl, Hans},
    doi = {10.2307/1971466},
    journal = {Annals of Mathematics},
    number = {1},
    pages = {173--193},
    publisher = {JSTOR},
    title = {On the structure of Brauer's centralizer algebras},
    volume = {128},
    year = {1988}
}

@article{werner1989quantum,
    author = {Werner, Reinhard F},
    doi = {10.1103/PhysRevA.40.4277},
    journal = {Physical Review A},
    number = {8},
    pages = {4277},
    publisher = {APS},
    title = {Quantum states with Einstein-Podolsky-Rosen correlations admitting a hidden-variable model},
    volume = {40},
    year = {1989}
}

@article{werner1998optimal,
    author = {Werner, Reinhard F},
    doi = {10.1103/PhysRevA.58.1827},
    journal = {Physical Review A},
    number = {3},
    pages = {1827},
    publisher = {APS},
    title = {Optimal cloning of pure states},
    volume = {58},
    year = {1998},
    archiveprefix={arXiv},
    eprint={quant-ph/9804001}
}

@article{wolf2003entanglement,
    author = {Wolf, Michael M and Verstraete, Frank and Ignacio Cirac, J},
    doi = {10.1142/S021974990300036X},
    journal = {International Journal of Quantum Information},
    number = {04},
    pages = {465--477},
    publisher = {World Scientific},
    title = {Entanglement and frustration in ordered systems},
    volume = {1},
    year = {2003},
    archiveprefix={arXiv},
    eprint={quant-ph/0311051}
}

@online{aubrun,
  author = {Guillaume Aubrun},
  title = {Schur-Weyl Duality},
  url = {https://math.univ-lyon1.fr/~aubrun/recherche/schur-weyl.pdf},
  year = {2018}
}

@article{fannes1988symmetric,
  title={Symmetric states of composite systems},
  author={Fannes, Mark and Lewis, John T and Verbeure, Andr{\'e}},
  journal={Letters in mathematical physics},
  volume={15},
  pages={255--260},
  year={1988},
  publisher={Springer}
}

@article{raggio1988quantum,
  title={Quantum statistical mechanics of general mean field systems},
  author={Raggio, GA and Werner, RF},
  journal={Helv. Phys. Acta},
  volume={62},
  number={DIAS-STP-88-49},
  pages={980},
  year={1988}
}

@article{Terhal2004,
  title = {Is entanglement monogamous?},
  volume = {48},
  ISSN = {0018-8646},
  url = {http://dx.doi.org/10.1147/rd.481.0071},
  DOI = {10.1147/rd.481.0071},
  number = {1},
  journal = {IBM Journal of Research and Development},
  publisher = {IBM},
  author = {Terhal,  B. M.},
  year = {2004},
  month = jan,
  pages = {71–78}
}

@article{renner2008security,
  title={Security of quantum key distribution},
  author={Renner, Renato},
  journal={International Journal of Quantum Information},
  volume={6},
  number={01},
  pages={1--127},
  year={2008},
  publisher={World Scientific}
}

@article{berta2021semidefinite,
  title={Semidefinite programming hierarchies for constrained bilinear optimization},
  author={Berta, Mario and Borderi, Francesco and Fawzi, Omar and Scholz, Volkher B},
  journal={Mathematical Programming},
  pages={1--49},
  year={2021},
  publisher={Springer}
}

@inproceedings{brandao2013product,
  title={Product-state approximations to quantum ground states},
  author={Brandao, Fernando GSL and Harrow, Aram W},
  booktitle={Proceedings of the forty-fifth annual ACM symposium on Theory of computing},
  pages={871--880},
  year={2013}
}

@article{walter2013entanglement,
  title={Entanglement polytopes: multiparticle entanglement from single-particle information},
  author={Walter, Michael and Doran, Brent and Gross, David and Christandl, Matthias},
  journal={Science},
  volume={340},
  number={6137},
  pages={1205--1208},
  year={2013},
  publisher={American Association for the Advancement of Science}
}

@inproceedings{schilling2015quantum,
  title={The quantum marginal problem},
  author={Schilling, Christian},
  booktitle={Mathematical Results in Quantum Mechanics: Proceedings of the QMath12 Conference},
  pages={165--176},
  year={2015},
  organization={World Scientific}
}

@article{klyachko2004quantum,
  title={Quantum marginal problem and representations of the symmetric group},
  author={Klyachko, Alexander A},
  journal={arXiv preprint quant-ph/0409113},
  year={2004}
}

@article{Klyachko_2006,
   title={Quantum marginal problem and N-representability},
   volume={36},
   ISSN={1742-6596},
   url={http://dx.doi.org/10.1088/1742-6596/36/1/014},
   DOI={10.1088/1742-6596/36/1/014},
   journal={Journal of Physics: Conference Series},
   publisher={IOP Publishing},
   author={Klyachko, Alexander A},
   year={2006},
   month=apr, 
   pages={72–86}
}

@article{lee2022optimizing,
  title={Optimizing quantum circuit parameters via SDP},
  author={Lee, Eunou},
  archiveprefix={arXiv},
  eprint={2209.00789},
  year={2022}
}

@article{lee2024improved,
  title={An improved Quantum Max Cut approximation via matching},
  author={Lee, Eunou and Parekh, Ojas},
  eprint={2401.03616},
  archivePrefix={arXiv},
  primaryClass={quant-ph},
  year={2024}
}

@article{Ryan_2022,
    title={On a Class of Orthogonal-Invariant Quantum Spin Systems on the Complete Graph},
    volume={2023},
    ISSN={1687-0247},
    url={http://dx.doi.org/10.1093/imrn/rnac034},
    DOI={10.1093/imrn/rnac034},
    number={7},
    journal={International Mathematics Research Notices},
    publisher={Oxford University Press (OUP)},
    author={Ryan, Kieran},
    year={2022},
    month=feb,
    pages={6078–6131}
}

@article{wernerExt1989,
	title = {An application of Bell's inequalities to a quantum state extension problem},
	volume = {17},
	issn = {1573-0530},
	url = {https://doi.org/10.1007/BF00399761},
	doi = {10.1007/BF00399761},
	pages = {359--363},
	number = {4},
	journaltitle = {Letters in Mathematical Physics},
	shortjournal = {Lett Math Phys},
	author = {Werner, Reinhard F.},
	urldate = {2024-08-22},
	date = {1989-05-01}
}

@article{doherty2014,
	title = {Entanglement and the shareability of quantum states},
	volume = {47},
	issn = {1751-8121},
	url = {https://dx.doi.org/10.1088/1751-8113/47/42/424004},
	doi = {10.1088/1751-8113/47/42/424004},
	pages = {424004},
	number = {42},
	journaltitle = {Journal of Physics A: Mathematical and Theoretical},
	shortjournal = {J. Phys. A: Math. Theor.},
	author = {Doherty, Andrew C.},
	urldate = {2024-08-22},
	date = {2014-10}
}

@article{jakabBilinearbiquadraticModelComplete2018,
	title = {The bilinear-biquadratic model on the complete graph},
	volume = {51},
	issn = {1751-8113, 1751-8121},
	url = {http://arxiv.org/abs/1709.06602},
	doi = {10.1088/1751-8121/aaa92b},
	pages = {105201},
	number = {10},
	journaltitle = {Journal of Physics A: Mathematical and Theoretical},
	shortjournal = {J. Phys. A: Math. Theor.},
	author = {Jakab, Dávid and Szirmai, Gergely and Zimborás, Zoltán},
	urldate = {2024-06-04},
	date = {2018-03-09},
	langid = {english},
	eprinttype = {arxiv},
	eprint = {1709.06602}
}

@article{jakabQuantumPhasesCollective2021,
	title = {Quantum phases of collective {SU}(3) spin systems with bipartite symmetry},
	volume = {103},
	issn = {2469-9950, 2469-9969},
	url = {http://arxiv.org/abs/2001.08310},
	doi = {10.1103/PhysRevB.103.214448},
	pages = {214448},
	number = {21},
	journaltitle = {Physical Review B},
	shortjournal = {Phys. Rev. B},
	author = {Jakab, Dávid and Zimborás, Zoltán},
	urldate = {2024-06-04},
	date = {2021-06-28},
	langid = {english},
	eprinttype = {arxiv},
	eprint = {2001.08310}
}

@thesis{jakabInterplayUnitaryPermutation2022,
	title = {The Interplay of Unitary and Permutation Symmetries in Composite Quantum Systems},
	url = {https://pea.lib.pte.hu/bitstream/handle/pea/34486/jakab-david-phd-2022.pdf},
	type = {phdthesis},
	author = {Jakab, Dávid},
	date = {2022},
	langid = {english}
}

%%%%%%%%%%%%%%%%%%%%%%%%%%%%%%%%%%%%%%%%%%%%%%%%%%%%%%%%%%%%%%%%%%%%%%%%%%%%%%%%%%%%%%%%%%%%%%%%%%%%%%%%%%%%%%%%%%%%%%%%%%%%%%%%%%%%%%%%%%%%%%%%%%%%%%%%%%%%%%%%%%%%%%%%%%%%%%%%%%%%%%%%%%%%%%%%%%%%%%%%%%%%%%%%%%%%%%%%%%
%%%%%%%%%%%%%%%%%%%%%%%%%%%%%%%%%%%%%%%%%%%%%%%%%%%%%%%%%%%%%%%%%%%%%%%%%%%%%%%%%%%%%%%%%%%%%%%%%%%%%%%%%%%%%%%%%%%%%%%%%%%%%%%%%%%%%%%%%%%%%%%%%%%%%%%%%%%%%%%%%%%%%%%%%%%%%%%%%%%%%%%%%%%%%%%%%%%%%%%%%%%%%%%%%%%%%%%%%%
\appendix

\section{\texorpdfstring{$K_n$}{Kn}-Extendibility of Werner states via primal SDP}\label{app:werner_primal}

Let $n \geq 1$ be the number of systems and $d \geq 2$ the local dimension of each system. We want to solve the optimization problem \cref{def:SDPworst} for $\Pi = \frac{\I-\F}{2}$, i.e.~our goal is to determine $q_W(n,d)$. Given an optimal $\rho$, we can assume without loss of generality that it has $U^{\otimes n}$ and $\S_n$ symmetry. This means, that we want to find an $n$-qudit density matrix $\rho$ whose any two-body marginal corresponding to an edge $e$ in the graph $K_n$ is
\begin{equation}
    \rho_{e}
    = p \rho_{\ydsm{1,1}}
    + (1-p) \rho_{\ydsm{2}}
\end{equation}
where $p \in [0,1]$ and
\begin{equation}
    \rho_{\ydsm{1,1}} \defeq \frac{\varepsilon_{\ydsm{1,1}}}{\Tr \varepsilon_{\ydsm{1,1}}}, \quad \rho_{\ydsm{2}} \defeq \frac{\varepsilon_{\ydsm{2}}}{\Tr \varepsilon_{\ydsm{2}}},
\end{equation}
with ranks $\Tr \varepsilon_{\ydsm{1,1}} = \frac{d(d-1)}{2}$ and $\Tr \varepsilon_{\ydsm{2}} = \frac{d(d+1)}{2}$.

Since the solution $\rho$ has $U^{\otimes n}$ and $\S_n$ symmetry, then we can assume
\begin{equation}
    \rho = \sum_{\substack{\lambda \pt n \\ l(\lambda) \leq d}} \beta_\lambda \rho_\lambda,
\end{equation}
where $\beta_\lambda$ is some probability distribution $\set{\beta_\lambda : \lambda \pt n}$ and $\rho_\lambda$ are normalized isotypical projectors $\varepsilon_\lambda$ onto the subspace $\lambda$ in the Schur--Weyl duality:
\begin{equation}
    \rho_\lambda
    \defeq \frac{\varepsilon_\lambda}{\Tr \varepsilon_\lambda}
    = \frac{\varepsilon_\lambda}{d(\lambda) m_d(\lambda)}.
\end{equation}
Their ranks $\Tr \varepsilon_\lambda = d(\lambda) m_d(\lambda)$ are given by dimensions of symmetric and unitary group irreps:
\begin{equation}
    d(\lambda) = \frac{n!}{ \prod_{u \in \lambda}\mathrm{hook}(u)}, \quad
    m_d(\lambda) = \prod_{u \in \lambda}\frac{d+\mathrm{cont(u)}}{\mathrm{hook}(u)}
\end{equation}
denote the dimensions of irreducible $\S_n$ and $\U(d)$ representations, known as hook formula and hook-content formula respectively. The key result which allows us to find $q_W(n,d)$ is the following lemma.

\begin{lemma}[\cite{Christandl2007}]
$\Tr_{[n] \backslash e} \rho_\lambda = \alpha^\lambda_{\ydsm{1,1}} \varepsilon_{\ydsm{1,1}} + \alpha^\lambda_{\ydsm{2}} \varepsilon_{\ydsm{2}}$ where
\begin{equation}
    \alpha^\lambda_{\ydsm{1,1}} = \frac{s^*_{\ydsm{1,1}}(\lambda)}{m_d(\yd{1,1})n(n-1)},
\end{equation}
where $s^*_{\mu}(\lambda)$ is the shifted Schur function $s^*_{\mu}$ evaluated for the partition $\lambda$, see \cite{okounkov1996shifted}.
\end{lemma}

Therefore, if we compute for every edge $e$ the reduced density matrix $\rho_e$ then we get
\begin{align*}
    \rho_e &= \Tr_{[n] \backslash e} \rho = \sum_{\substack{\lambda \pt n \\ l(\lambda) \leq d}} \beta_\lambda \Tr_{[n] \backslash e} \rho_\lambda = \sum_{\substack{\lambda \pt n \\ l(\lambda) \leq d}} \beta_\lambda \of{\alpha^\lambda_{\ydsm{1,1}} \varepsilon_{\ydsm{1,1}} + \alpha^\lambda_{\ydsm{2}} \varepsilon_{\ydsm{2}}} = \of[\bigg]{\sum_{\substack{\lambda \pt n \\ l(\lambda) \leq d}} \beta_\lambda \alpha^\lambda_{\ydsm{1,1}} \Tr \varepsilon_{\ydsm{1,1}}}  \rho_{\ydsm{1,1}} + \of[\bigg]{\sum_{\substack{\lambda \pt n \\ l(\lambda) \leq d}} \beta_\lambda \alpha^\lambda_{\ydsm{2}} \Tr \varepsilon_{\ydsm{2}}}  \rho_{\ydsm{2}},
\end{align*}
which means that 
\begin{equation}
    p = \sum_{\substack{\lambda \pt n \\ l(\lambda) \leq d}} \beta_\lambda \frac{d(\yd{1,1}) s^*_{\ydsm{1,1}}(\lambda)}{n(n-1)}.
\end{equation}
So after using a formula for the shifted Schur function \cite{okounkov1996shifted} we get the following solution to our optimization problem:
\begin{align}\label{app:werner_optimization}
    q_W(n,d) &=  \max_{\substack{\lambda \pt n \\ l(\lambda) \leq d}} \frac{d(\yd{1,1}) s^*_{\ydsm{1,1}}(\lambda)}{n(n-1)} =  \max_{\substack{\lambda \pt n \\ l(\lambda) \leq d}} \frac{\sum_{d \geq i > j \geq 1} \lambda_i (\lambda_j + 1)}{n(n-1)}
\end{align}

\begin{example}
When $d = 2$, the general formula for $\alpha$ is
\begin{equation}
    \frac{\lambda_2 (\lambda_1 + 1)}{n(n-1)}.
\end{equation}
We want to maximize this subject to $n \geq \lambda_1 \geq \lambda_2 \geq 0$ and $\lambda_1 + \lambda_2 = n$. The optimal value is
\begin{equation}
    \alpha =
    \begin{cases}
        \frac{1}{4} + \frac{3}{4n} & \text{if $n$ is odd}, \\
        \frac{1}{4} + \frac{3}{4(n-1)} & \text{if $n$ is even}.
    \end{cases}
\end{equation}
\end{example}

\begin{theorem}
    The general answer to \cref{app:werner_optimization} is
    \begin{equation}
        q_W(n,d)
        = \frac{d-1}{2d} \cdot
           \frac{(n+k+d)(n-k)}{n(n-1)} + \frac{k(k-1)}{n(n-1)}
    \end{equation}
    where $k = n \bmod d$.
\end{theorem}
\begin{proof}
We will guess the optimal solution for the Young diagram $\lambda$ and prove that it is not possible to improve on it. The conjectured optimal solution:
\begin{equation}
    \lambda_1 = \dotsc = \lambda_k = \frac{n-k}{d}+1, \quad \lambda_{k+1} = \dotsc = \lambda_d = \frac{n-k}{d}
\end{equation}
Consider an arbitrary Young diagram $\mu = \lambda + \Delta$, which you get by a perturbation $\Delta$, i.e.~$\Delta \defeq (\Delta_1, \dotsc, \Delta_d)$ with the properties $\sum_{i=1}^d \Delta_i = 0$ and $\sum_{i=1}^j \Delta_i \geq 0$ for all $j \in \set{1, \dotsc, d}$.
Then we can estimate the difference between shifted Schur functions as
\begin{align*}
    \frac{1}{d(\yd{1,1})} \of*{s^*_{\ydsm{1,1}}(\mu) - s^*_{\ydsm{1,1}}(\lambda)}&= \sum_{d \geq i > j \geq 1} \of[\big]{ (\lambda_i + \Delta_i) (\lambda_j + \Delta_j + 1) - \lambda_i (\lambda_j + 1) } = \sum_{d \geq i > j \geq 1} \of{ \Delta_i \lambda_j + \lambda_i \Delta_j + \Delta_i + \Delta_i \Delta_j } \\
    &= \sum_{i=2}^{d} \Delta_i \sum_{j=1}^{i-1} \lambda_j + \sum_{i=1}^{d-1} \Delta_i \sum_{j=i+1}^{d} \lambda_j + \sum_{i=2}^{d} \Delta_i (i-1) - \frac{1}{2} \sum_{i=1}^{d} \Delta_i^2 \\
    &= \Delta_d \sum_{j=1}^{d-1} \lambda_j + \Delta_1 \sum_{j=2}^{d} \lambda_j + \sum_{i=2}^{d-1} \Delta_i \sum_{j=1}^{d} \lambda_j - \sum_{i=2}^{d-1} \Delta_i \lambda_i + \sum_{j=2}^{d-1} \sum_{i=j}^{d} \Delta_i - \frac{1}{2} \sum_{i=1}^{d} \Delta_i^2 \\
    &= \Delta_d (n - \lambda_d) + \Delta_1 (n - \lambda_1) - n (\Delta_d + \Delta_1) - \sum_{i=2}^{d-1} \Delta_i \lambda_i - \sum_{j=2}^{d-1} \sum_{i=1}^{j-1} \Delta_i - \frac{1}{2} \sum_{i=1}^{d} \Delta_i^2 \\
    &= -\sum_{i=1}^{d} \Delta_i \lambda_i - \sum_{j=2}^{d-1} \sum_{i=1}^{j-1} \Delta_i - \frac{1}{2} \sum_{i=1}^{d} \Delta_i^2 \\
    &= -\sum_{i=1}^{k} \Delta_i \of*{\frac{n-k}{d}+1} -\sum_{i=k+1}^{d} \Delta_i \frac{n-k}{d} - \sum_{j=2}^{d-1} \sum_{i=1}^{j-1} \Delta_i - \frac{1}{2} \sum_{i=1}^{d} \Delta_i^2 \\
    &= -\sum_{i=1}^{k} \Delta_i  - \sum_{j=2}^{d-1} \sum_{i=1}^{j-1} \Delta_i - \frac{1}{2} \sum_{i=1}^{d} \Delta_i^2 \leq 0,
\end{align*}
which shows that $\lambda$ is actually the optimal solution. Evaluating the shifted Schur function at $\lambda$ gives the answer.
\end{proof}

%%%%%%%%%%%%%%%%%%%%%%%%%%%%%%%%%%%%%%%%%%%%%%%%%%%%%%%%%%%%%%%%%%%%%%%%%%%%%%%%%%%%%%%%%%%%%%%%%%%%%%%%%%%%%%%%%%%%%%%%%%%%%%%%%%%%%%%%%%%%%%%%%%%%%%%%%%%%%%%%%%%%%%%%%%%%%%%%%%%%%%%%%%%%%%%%%%%%%%%%%%%%%%%%%%%%%%%%%%
%%%%%%%%%%%%%%%%%%%%%%%%%%%%%%%%%%%%%%%%%%%%%%%%%%%%%%%%%%%%%%%%%%%%%%%%%%%%%%%%%%%%%%%%%%%%%%%%%%%%%%%%%%%%%%%%%%%%%%%%%%%%%%%%%%%%%%%%%%%%%%%%%%%%%%%%%%%%%%%%%%%%%%%%%%%%%%%%%%%%%%%%%%%%%%%%%%%%%%%%%%%%%%%%%%%%%%%%%%
\section{The dual SDPs for \texorpdfstring{$K_n$}{Kn}-extendibility}

\subsection{Isotropic states}
\label{app:dualSDPIsotropicStates}

In this section we are going to solve the optimisation problem \cref{eq:SDPStandard}:
\begin{equation*}
    \begin{aligned}
        p'_{I}(n,d) = \max_{\rho,p'} \quad & p' \qquad
        \textrm{s.t.} \quad
         \rho_e = \frac{p'}{d} \cdot \W + \frac{1-p'}{d^2} \cdot \I \quad \forall e \in E, \quad
         \rho \succeq 0.
    \end{aligned}
\end{equation*}
The Lagrangian \cite{boyd2004convex} associated with it is defined as,
\begin{equation}
    L \big{(} p, \rho, {(h_e)}_e, Z \big{)} \defeq p + \sum_{e \in E} \frob[\Big]{ h_e , \rho_e - \frac{p}{d} \W - \frac{ (1 - p)}{d^2} \I} + \frob[\big]{Z , \rho},
\end{equation}
where ${(h_e)}_e$ is a family of Hermitian matrices acting on ${\big{(} \C^d \big{)}}^{\x 2}$, $Z$ is a Hermitian matrix acting on ${\big{(} \C^d \big{)}}^{\x n}$ and $\langle \cdot , \cdot \rangle$ denotes the Frobenius inner product. The Min-Max principle states that
\begin{align*}
    \max_{\rho,p} \quad \min_{\substack{{Z,(h_e)}_e \\ Z \succeq 0 }} \quad L \of[\big]{ p, \rho, {Z,(h_e)}_e} \leq \min_{\substack{{Z,(h_e)}_e \\ Z \succeq 0 }} \quad \max_{\rho,p} \quad L \of[\big]{ p, \rho, {Z,(h_e)}_e}.
\end{align*}
In fact, Slater's condition holds for our SDP (take $\rho = \frac{1}{d^n} \cdot \I^{\otimes n}$ and $p = 0$) and we have the equality. Rewriting the Lagrangian gives,
\begin{align*}
    L \big{(} p, \rho, {(h_e)}_e, Z \big{)} &= - \frac{\Tr [H]}{d^{n}} + p \of[\bigg]{ 1 - \sum_{e \in E} \frob[\Big]{ h_e , \frac{\W}{d}  - \frac{\I}{d^2} }} + \frob[\big]{ H + Z, \rho},
\end{align*}
where $H \defeq \sum_{e \in E} (h_e \x I_{\bar{e}})$ and $I_{\bar{e}}$ is the identity on all vertices except those of $e$. Therefore the dual SDP of \cref{eq:SDPStandard} is
\begin{equation}\label{eq:dualSDP}
    \begin{aligned}
        p'_I(n,d) = \min_{\mathclap{{(h_e)}_e,Z}} \quad - \frac{\Tr[H]}{d^{n}}  \quad
        \textrm{s.t.} \quad
         \sum_{e \in E} \Big{\langle} h_e , \frac{\W}{d} - \frac{\I}{d^2} \Big{\rangle} = 1, \quad
         H + Z = 0, \quad Z \succeq 0.
    \end{aligned}
\end{equation}
Recall that for any Hermitian matrix $M$ the smallest $\lambda \in \R$ such that $\lambda I \succeq M$ is equal to the largest eigenvalue of $M$, denoted $\lambda_{\mathrm{max}}(M)$. Then using the substitution  $\of{h_e - \frac{\Tr[h_e]}{d^2}I_e} \mapsto h_e $ the \cref{eq:dualSDP} can be rewritten as
\begin{equation}\label{eq:dualSDPSimplified}
    \begin{aligned}
        p'_I = \min_{{(h_e)}_e} \quad  \lambda_{\mathrm{max}}\of[\bigg]{\sum_{e \in E} h_e \x I_{\bar{e}}}  \quad
        \textrm{s.t.} \quad
        \sum_{e \in E} \Big{\langle} h_e , \frac{\W}{d} \Big{\rangle} = 1, \quad
        \Tr[h_e] = 0 \quad \forall e \in E.
    \end{aligned}
\end{equation}

We can simplify \cref{eq:dualSDPSimplified} even further, using the commutation relation of the isotropic states, i.e
\begin{equation}
    [\rho, \bar{U} \otimes U] = [\rho, O \otimes O] = 0, \qquad \forall U \in \U_d, \: \forall O \in \Orth_d,
\end{equation}
by twirling the operator $H$ with respect to the orthogonal group. Let ${(h^*_e)}_e = \arg \; p'_I(n,d)$ be a family of Hermitian matrices acting on ${\big{(} \C^d \big{)}}^{\x 2}$ optimal for the dual problem $p'_I(n,d)$. For all edges $e \in E$, define the twirling of $h^*_e$ as follows:
\begin{equation}
    \widetilde{h^*_e} \defeq \int_{\mathrm{O}_d} O^{\x 2} h^*_e {\of[\big]{ O^* }}^{\x 2} \operatorname{d} O,
\end{equation}
where the integral is taken with respect to the normalized Haar measure on the orthogonal group. Twirling in the same fashion the operator $H^*$ we get $\widetilde{H^*}$. Then by the convexity of $\lambda_{\text{max}}$,
\begin{align*}
    \lambda_{\text{max}} (H^*) &\geq \lambda_{\text{max}} \of[\big]{ \widetilde{H^*} }.
\end{align*}
The constraints of the \cref{eq:dualSDPSimplified} are also satisfied with ${\big{(} \widetilde{h^*_e} \big{)}}_e$ by the cyclic property of the trace, and hence we can restrict the dual problem to twirled ${\big{(} \widetilde{h_e} \big{)}}_e$.

Since each $\widetilde{h_e}$ commutes with the action of the orthogonal group, we can write (see \cref{sec:Schur_Weyl_duality_orthogonal}),
\begin{equation}
    \widetilde{h_e} = \alpha_e \I_e + \beta_e \W_e + \gamma_e \F_e,
\end{equation}
where $\F \defeq \sum^d_{i,j = 1} \ketbra{ij}{ji}$ is the flip operator. Note that condition $\Tr[\widetilde{h_e}] = 0$ is equivalent to $\alpha_e = -\frac{\beta_e+\gamma_e}{d}$. Then \cref{eq:dualSDPSimplified} becomes
\begin{equation}\label{eq:dualSDPSimplifiedOrthogonal}
    \begin{aligned}
        p'_I(n,d) = \min_{{(\beta_e, \gamma_e)}_e} \quad & \lambda_{\text{max}} \of*{ \sum_{e \in E} \of*{ -\frac{\beta_e+\gamma_e}{d} \I_{e} + \beta_e \W_e + \gamma_e \F_e} \x I_{\bar{e}} } \\ 
        \textrm{s.t.} \quad
        & (d^2-1) \of*{\sum_{e \in E}\beta_e} + (d-1) \of*{\sum_{e \in E}\gamma_e} = d.
    \end{aligned}
\end{equation}
Similarly, using the commutation relation
\begin{equation}
    [\rho, \psi(\pi)] = 0, \qquad \forall \pi \in \S_n,
\end{equation}
we get that for all edges $e \in E$ the values of $\beta_e$ and $\gamma_e$ are equal, i.e.~we can write for some $\beta \in \R$ and $\gamma \in \R$:
\begin{equation}
   \forall e \in E: \quad \beta_e = \beta \quad \text{ and } \quad \gamma_e = \gamma.
\end{equation}
Therefore we can rewrite \cref{eq:dualSDPSimplifiedOrthogonal} as follows:
\begin{equation}\label{eq:dualSDPSimplifiedOrthogonalSymmetric}
    \begin{aligned}
        p'_I(n,d) = \min_{\beta, \gamma} \quad & \lambda_{\text{max}} \of*{ \sum_{e \in E} \of*{ -\frac{\beta+\gamma}{d} \I_{e} + \beta \W_e + \gamma \F_e} \x I_{\bar{e}} }\\
        \textrm{s.t.} \quad
        & (d^2-1) \abs{E} \beta + (d-1) \abs{E} \gamma = d.
    \end{aligned}
\end{equation}
Using the constraint $(d^2-1) \abs{E} \beta + (d-1) \abs{E} \gamma = d$ to eliminate $\gamma$, and by the change of variable $-\beta \mapsto x$ and rewriting
\begin{equation}
    f(x) \defeq \lambda_{\text{max}} \of*{ H(x) }, \qquad H(x) \defeq \sum_{e \in E} \of*{ \of*{\frac{1}{\abs{E}(1-d)} - x} \of{\I_e - d \,\F_e} + x \of{\F_e - \W_e} }\x I_{\bar{e}},
\end{equation}
we reformulate \cref{eq:dualSDPSimplifiedOrthogonalSymmetric} as
\begin{equation}
    \begin{aligned}
    p'_I(n,d) = \min_{x \in \R} f(x).
     \end{aligned}
\end{equation}

%%%%%%%%%%%%%%%%%%%%%%%%%%%%%%%%%%%%%%%%%%%%%%%%%%%%%%%%%%%%%%%%%%%%%%%%%%%%%%%%%%%%%%%%%%%%%%%%%%%%%%%%%%%%%%%%%%%%%%%%%%%%%%%%%%%%%%%%%%%%%%%%%%%%%%%%%%%%%%%%%%%%%%%%%%%%%%%%%%%%%%%%%%%%%%%%%%%%%%%%%%%%%%%%%%%%%%%%%%
\subsection{Brauer states when \texorpdfstring{$q=0$}{q equals 0}}
\label{app:brauer_q=0}

In this section, we are going to solve the following optimization problem:
\begin{equation} \label{def:opt_q=0}
    \begin{aligned}
        p^*(n,d) = \max_{\rho,p} \quad & p \qquad
        \textrm{s.t.} \quad
         \rho_e = p \cdot \Pi_\0 + \of*{1 - p} \cdot \tfrac{\Pi_{\ydsm{2}}}{\Tr \Pi_{\ydsm{2}}} \quad \forall e \in E, \quad \rho \succeq 0.
\end{aligned}
\end{equation}
Similarly to \Cref{app:dualSDPIsotropicStates}, the Slater's condition holds for our SDP: take $\lambda = (n)$, and set $p = \frac{1}{1 + \Tr \Pi_{\ydsm{2}}}$ and $\rho = \frac{\varepsilon_{\lambda}}{\Tr \varepsilon_{\lambda}}$ is a normalized projector onto the symmetric subspace of $n$ qudits. Therefore, we get the following dual SDP:
\begin{equation}\label{eq:opt_q=0_dual_init}
    \begin{aligned}
         p^*(n,d) = \min_{\mathclap{{(h_e)}_e,Z}} \quad -\sum_{e \in E} \frob[\big]{h_e , \frac{\Pi_{\ydsm{2}}}{\Tr \Pi_{\ydsm{2}}}}  \quad
        \textrm{s.t.} \quad
         \sum_{e \in E} \frob[\Big]{ h_e , \Pi_\0 - \frac{\Pi_{\ydsm{2}}}{\Tr \Pi_{\ydsm{2}}} } = 1, \quad
         H + Z = 0, \quad Z \succeq 0.
    \end{aligned}
\end{equation}
This SDP has orthogonal symmetry and $\S_n$ of the complete graph $K_n$, therefore similarly to \cref{app:dualSDPIsotropicStates}, we can assume that $h_e = \alpha \I_e + \beta \W_e + \gamma \F_e$ for every $e \in E$. In the end, this leads to the following simplification of \cref{eq:opt_q=0_dual_init}:
\begin{equation}\label{eq:opt_q=0_dual_simplified}
    \begin{aligned}
         p^*(n,d)  = \min_{x \in \mathbb{R}} \quad & \frac{1}{d \cdot \abs{E}} \lambda_{\mathrm{max}} \of[\bigg]{ \sum_{e \in E} \W_e - x \cdot \F_e } + \frac{x}{d}.
    \end{aligned}
\end{equation}
But the spectrum of the Hamiltonian $\sum_{e \in E} \W_e - x \cdot \F_e $ can be easily obtained, see \cref{sec:JM_elemnts}. Therefore we can simplify \cref{eq:opt_q=0_dual_simplified} to
\begin{equation}\label{eq:opt_q=0_dual}
    \begin{aligned}
         p^*(n,d)  = \min_{x \in \mathbb{R}} \max_{ (\lambda,\mu) \in \Omega_{n,d} } \quad & \frac{1}{d \cdot \abs{E}} \of[\bigg]{ \cont(\mu) - \cont(\lambda) + \frac{(n - \abs{\lambda})(d-1)}{2}} + \frac{x}{d} \of[\bigg]{1 - \frac{\cont(\mu)}{\abs{E}}}.
    \end{aligned}
\end{equation}

%%%%%%%%%%%%%%%%%%%%%%%%%%%%%%%%%%%%%%%%%%%%%%%%%%%%%%%%%%%%%%%%%%%%%%%%%%%%%%%%%%%%%%%%%%%%%%%%%%%%%%%%%%%%%%%%%%%%%%%%%%%%%%%%%%%%%%%%%%%%%%%%%%%%%%%%%%%%%%%%%%%%%%%%%%%%%%%%%%%%%%%%%%%%%%%%%%%%%%%%%%%%%%%%%%%%%%%%%%
\subsection{Brauer states with fixed \texorpdfstring{$p$}{p}} 
\label{app:brauer_p}

In this section, we are going to solve the following optimization problem:
\begin{equation}
    \begin{aligned}
        q(p,n,d) = \max_{\rho,q} \quad & q \\
        \textrm{s.t.} \quad & \rho_e = p \cdot \Pi_{\0} + q \cdot \tfrac{\Pi_{\ydsm{1,1}}}{\Tr \Pi_{\ydsm{1,1}}} + (1 - p - q) \cdot \tfrac{\Pi_{\ydsm{2}}}{\Tr \Pi_{\ydsm{2}}}, \quad \forall e \in E\\
        &\rho \succeq 0.
\end{aligned}
\end{equation}
Slater's condition holds for our SDP as in \Cref{app:dualSDPIsotropicStates}. Therefore the dual SDP of \cref{def:opt_q=0} is
\begin{equation}\label{eq:opt_fixed_p_dual_init}
    \begin{aligned}
         q(p,n,d) = \min_{\mathclap{{(h_e)}_e,Z}} \quad - \sum_{e \in E} \frob[\big]{h_e , p \cdot \Pi_{\0} + (1 - p) \cdot \tfrac{\Pi_{\ydsm{2}}}{\Tr \Pi_{\ydsm{2}}}}  \quad
        \textrm{s.t.} \quad
         \sum_{e \in E} \frob[\Big]{ h_e , \tfrac{\Pi_{\ydsm{1,1}}}{\Tr \Pi_{\ydsm{1,1}}} - \tfrac{\Pi_{\ydsm{2}}}{\Tr \Pi_{\ydsm{2}}}} = 1, \; H + Z = 0, \; Z \succeq 0.
    \end{aligned}
\end{equation}
Similarly to \cref{app:brauer_q=0}, we can assume that $h_e = \alpha \I_e + \beta \W_e + \gamma \F_e$ for every $e \in E$. This leads to the following simplification of \cref{eq:opt_fixed_p_dual_init}:
\begin{equation}\label{eq:opt_fixed_p_dual_simplified}
    \begin{aligned}
         q(p,n,d)  = \min_{x \in \mathbb{R}} \frac{1}{\abs{E}} \lambda_{\mathrm{max}} \of[\bigg]{ \sum_{e \in E} x \frac{\W_e}{d} - \frac{\F_e}{2} } + \frac{1}{2} - p \cdot x.
    \end{aligned}
\end{equation}
But the spectrum of the Hamiltonian $\sum_{e \in E} x \frac{\W_e}{d} - \frac{\F_e}{2}$ can be easily obtained, see \cref{sec:JM_elemnts}. Therefore we can simplify \cref{eq:opt_fixed_p_dual_init} to
\begin{equation}\label{eq:opt_fixed_p_dual}
    \begin{aligned}
         q(p,n,d)  = \min_{x \in \mathbb{R}} \max_{ (\lambda,\mu) \in \Omega_{n,d} } \of*{\frac{x}{d \cdot \abs{E}} \of[\bigg]{ \cont(\mu) - \cont(\lambda) + \frac{(n - \abs{\lambda})(d-1)}{2}} + \frac{1}{2} \of[\bigg]{1 - \frac{\cont(\mu)}{\abs{E}}} - p \cdot x}.
    \end{aligned}
\end{equation}

%%%%%%%%%%%%%%%%%%%%%%%%%%%%%%%%%%%%%%%%%%%%%%%%%%%%%%%%%%%%%%%%%%%%%%%%%%%%%%%%%%%%%%%%%%%%%%%%%%%%%%%%%%%%%%%%%%%%%%%%%%%%%%%%%%%%%%%%%%%%%%%%%%%%%%%%%%%%%%%%%%%%%%%%%%%%%%%%%%%%%%%%%%%%%%%%%%%%%%%%%%%%%%%%%%%%%%%%%%
%%%%%%%%%%%%%%%%%%%%%%%%%%%%%%%%%%%%%%%%%%%%%%%%%%%%%%%%%%%%%%%%%%%%%%%%%%%%%%%%%%%%%%%%%%%%%%%%%%%%%%%%%%%%%%%%%%%%%%%%%%%%%%%%%%%%%%%%%%%%%%%%%%%%%%%%%%%%%%%%%%%%%%%%%%%%%%%%%%%%%%%%%%%%%%%%%%%%%%%%%%%%%%%%%%%%%%%%%%
\section{PPT criterion for Brauer states} \label{app:BrauerStates}

The two-parameter $(p,q)$ family of Brauer states, given by \cref{eq:BrauerStates_proj} as a sum of orthogonal projectors, is positive semi-definite if and only if $p \geq 0, q \geq 0$ and $p + q \leq 1$. This family of states can be alternatively parametrized as:
\begin{equation}\label{eq:BrauerStates}
    p^\prime \cdot \frac{\W}{d} + q^\prime \cdot \frac{\F}{d} + (1 - p^\prime - q^\prime) \cdot \frac{\I}{d^2}.
\end{equation}
The parameters $(p,q)$ of \cref{eq:BrauerStates_proj} and $(p',q')$ of \cref{eq:BrauerStates} are related via
\begin{equation} \label{eq:BrauerStatesChangeVariable}
    \begin{cases}
        p' = p - \frac{2 (1 - p - q)}{d (d + 1) - 2} \\
        q' = \frac{-q}{d - 1} + \frac{d (1 - p - q)}{d (d + 1) - 2}
    \end{cases}
    \qquad\text{and}\qquad
    \begin{cases}
        p = \frac{p' (d^2 - 1) + q' (d - 1) + 1}{d^2} \\
        q = - \frac{p' (d - 1) + q' (d^2 - 1) - d + 1}{2 d},
    \end{cases}
\end{equation}
such that \cref{eq:BrauerStates} is positive semi-definite if and only if the following holds
\begin{equation} \label{eq:positivesemiDefiniteBrauerStates}
    \begin{cases}
        p' (d^2 - 1) + q' (d - 1) + 1 \geq 0 \\
        - p' - q' (d + 1) + 1 \geq 0 \\
        2 - d + d^2 + p' (d^2 + d - 2) - q' (d^3 - 3d + 2) \leq 2 d^2.
    \end{cases}
\end{equation}

A bipartite state satisfies the \textsc{PPT} criterion (or simply is \textsc{PPT}) if its partial transpose is positive semi-definite. Consequently, the set of \textsc{PPT} states contains the set of separable states. In the context of the two-parameter family of Brauer states, the \textsc{PPT} criterion is equivalent to the separability \cite{vollbrecht2001entanglement,park2023universal}. The partial transpose of a Brauer state of the form \cref{eq:BrauerStates} becomes
\begin{equation}
    p^\prime \cdot \frac{\F}{d} + q^\prime \cdot \frac{\W}{d} + (1 - p^\prime - q^\prime) \cdot \frac{\I}{d^2},
\end{equation}
i.e.~it is just the change of variable: $(p^\prime, q^\prime) \mapsto (q^\prime, p^\prime)$. A Brauer state's separability (or equivalently \textsc{PPT}) is determined by the intersection of the region defined by:
\begin{equation}
    \begin{cases}
        q' (d^2 - 1) + p' (d - 1) + 1 \geq 0 \\
        - q' - p' (d + 1) + 1 \geq 0 \\
        2 - d + d^2 + q' (d^2 + d - 2) - p' (d^3 - 3d + 2) \leq 2 d^2,
    \end{cases}
\end{equation}
and \cref{eq:positivesemiDefiniteBrauerStates} (see \cref{fig:ConvexBrauerStates}). Alternatively, using the relations \cref{eq:BrauerStatesChangeVariable} a Brauer state for any fixed $d \geq 2$ is separable if and only if
\begin{equation}
    \begin{cases}
        0 \leq q \leq \frac{1}{2} \\
        0 \leq p \leq \frac{1}{d} 
    \end{cases}
\end{equation}

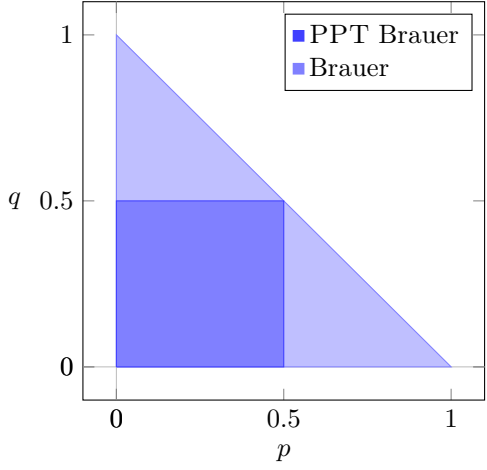
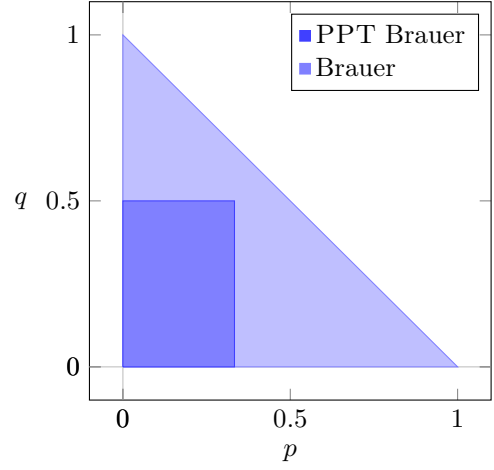
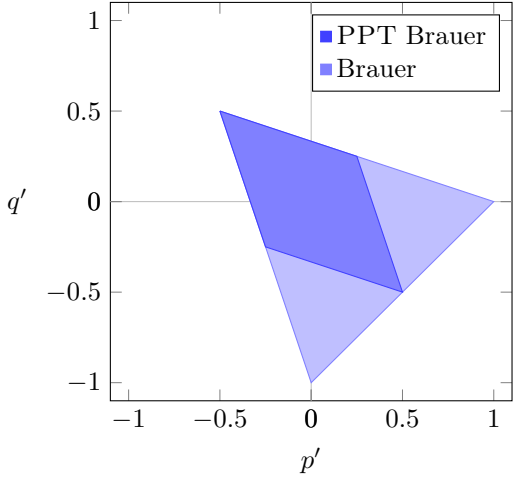
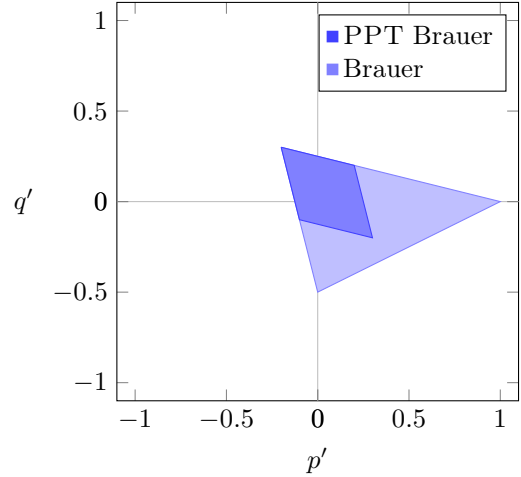
\begin{figure}[!h]
    \centering
    \begin{subfigure}[b]{0.4\textwidth}
        \centering
        \begin{tikzpicture}
            \begin{axis}[use as bounding box,
                         xmin =  {0 - 0.1},
                         xmax = {1 + 0.1},
                         ymin = {0 - 0.1},
                         ymax = {1 + 0.1},
                         axis equal image,
                         xlabel = {$p$},
                         ylabel = {$q$},
                         ylabel style = {rotate = -90},
                         xtick = {0, 0.5, 1},
                         ytick = {0, 0.5, 1},
                         extra x ticks = {0},
                         extra x tick style = {grid = major},
                         extra y ticks = {0},
                         extra y tick style = {grid = major},
                         legend pos = north east,
                         legend cell align={left},
                         scale = 0.75,
                         ylabel style={overlay},
                         yticklabel style={overlay},
                         ]            
    
                \addlegendimage{only marks, mark = square*, blue!75}
                \addlegendentry{\textsc{PPT} Brauer};
    
                \addlegendimage{only marks, mark = square*, blue!50}
                \addlegendentry{Brauer};
    
                \path[draw = blue!50, fill = blue!25] (0,0) -- (0,1) -- (1,0) -- (0,0);
                \path[draw = blue!75, fill = blue!50] (0,0) -- (0,1/2) -- (1/2,1/2) -- (1/2,0) -- (0,0);
            \end{axis}
        \end{tikzpicture}
        \caption{Brauer states of the form \cref{eq:BrauerStates_proj} ($d=2$).}
    \end{subfigure}
    \hspace{4em}
    \begin{subfigure}[b]{0.4\textwidth}
        \centering
        \begin{tikzpicture}
            \begin{axis}[use as bounding box,
                         xmin =  {0 - 0.1},
                         xmax = {1 + 0.1},
                         ymin = {0 - 0.1},
                         ymax = {1 + 0.1},
                         axis equal image,
                         xlabel = {$p$},
                         ylabel = {$q$},
                         ylabel style = {rotate = -90},
                         xtick = {0, 0.5, 1},
                         ytick = {0, 0.5, 1},
                         extra x ticks = {0},
                         extra x tick style = {grid = major},
                         extra y ticks = {0},
                         extra y tick style = {grid = major},
                         legend pos = north east,
                         legend cell align={left},
                         scale = 0.75,
                         ylabel style={overlay},
                         yticklabel style={overlay},
                         ]            
    
                \addlegendimage{only marks, mark = square*, blue!75}
                \addlegendentry{\textsc{PPT} Brauer};
    
                \addlegendimage{only marks, mark = square*, blue!50}
                \addlegendentry{Brauer};
    
                \path[draw = blue!50, fill = blue!25] (0,0) -- (0,1) -- (1,0) -- (0,0);
                \path[draw = blue!75, fill = blue!50] (0,0) -- (0,1/2) -- (1/3,1/2) -- (1/3,0) -- (0,0);
            \end{axis}
        \end{tikzpicture}
        \caption{Brauer states of the form \cref{eq:BrauerStates_proj} ($d=3$).}
    \end{subfigure}
    \\[2em]
    \begin{subfigure}[b]{0.4\textwidth}
        \centering
        \begin{tikzpicture}
            \begin{axis}[use as bounding box,
                         xmin =  {-1 - 0.1},
                         xmax = {1 + 0.1},
                         ymin = {-1 - 0.1},
                         ymax = {1 + 0.1},
                         axis equal image,
                         xlabel = {$p^\prime$},
                         ylabel = {$q^\prime$},
                         ylabel style = {rotate = -90},
                         xtick = {-1, -0.5, 0, 0.5, 1},
                         ytick = {-1, -0.5, 0, 0.5, 1},
                         extra x ticks = {0},
                         extra x tick style = {grid = major},
                         extra y ticks = {0},
                         extra y tick style = {grid = major},
                         legend pos = north east,
                         legend cell align={left},
                         scale = 0.75,
                         ylabel style={overlay},
                         yticklabel style={overlay},
                         ]            
    
                \addlegendimage{only marks, mark = square*, blue!75}
                \addlegendentry{\textsc{PPT} Brauer};
    
                \addlegendimage{only marks, mark = square*, blue!50}
                \addlegendentry{Brauer};
    
                \path[draw = blue!50, fill = blue!25] (-0.5,0.5) -- (0,-1) -- (1,0) -- (-0.5,0.5);
                \path[draw = blue!75, fill = blue!50] (-0.5,0.5) -- (1/4,1/4) -- (0.5,-0.5) -- (-1/4, -1/4) -- (-0.5,0.5);
            \end{axis}
        \end{tikzpicture}
        \caption{Brauer states of the form \cref{eq:BrauerStates} ($d=2$).}
    \end{subfigure}
    \hspace{4em}
    \begin{subfigure}[b]{0.4\textwidth}
        \centering
        \begin{tikzpicture}
            \begin{axis}[use as bounding box,
                         xmin =  {-1 - 0.1},
                         xmax = {1 + 0.1},
                         ymin = {-1 - 0.1},
                         ymax = {1 + 0.1},
                         axis equal image,
                         xlabel = {$p^\prime$},
                         ylabel = {$q^\prime$},
                         ylabel style = {rotate = -90},
                         xtick = {-1, -0.5, 0, 0.5, 1},
                         ytick = {-1, -0.5, 0, 0.5, 1},
                         extra x ticks = {0},
                         extra x tick style = {grid = major},
                         extra y ticks = {0},
                         extra y tick style = {grid = major},
                         legend pos = north east,
                         legend cell align={left},
                         scale = 0.75,
                         ylabel style={overlay},
                         yticklabel style={overlay},
                         ]            
    
                \addlegendimage{only marks, mark = square*, blue!75}
                \addlegendentry{\textsc{PPT} Brauer};
    
                \addlegendimage{only marks, mark = square*, blue!50}
                \addlegendentry{Brauer};
    
                \path[draw = blue!50, fill = blue!25] (-1/5,3/10) -- (1,0) -- (0,-1/2) -- (-1/5,3/10);
                \path[draw = blue!75, fill = blue!50] (-1/5,3/10) -- (1/5,1/5) -- (3/10,-1/5) -- (-1/10,-1/10) -- (-1/5,3/10);
            \end{axis}
        \end{tikzpicture}
        \caption{Brauer states of the form \cref{eq:BrauerStates} ($d=3$).}
    \end{subfigure}
    \caption{The two-parameter family of Brauer states, and the subset of separable (\textsc{PPT}) states.}
    \label{fig:ConvexBrauerStates}
\end{figure}

\end{document}